\documentclass[12pt]{article}
\usepackage[margin=1in]{geometry}
\usepackage{amsmath,amssymb,amsthm,mathtools}
\usepackage[dvipsnames]{xcolor}
\usepackage{enumitem}
\usepackage{microtype}
\usepackage{bm}
\usepackage{mleftright}
\usepackage{booktabs,tabularx,longtable,array}
\usepackage{placeins}
\usepackage{thmtools}
\usepackage{thm-restate}
\usepackage{hyperref}
\newtheorem{theorem}{Theorem}[section]
\newtheorem{proposition}[theorem]{Proposition}
\newtheorem{lemma}[theorem]{Lemma}
\newtheorem{corollary}[theorem]{Corollary}

\theoremstyle{definition}
\newtheorem{definition}[theorem]{Definition}
\newtheorem{remark}[theorem]{Remark}

\newcommand{\BQP}{\mathsf{BQP}}
\newcommand{\NP}{\mathsf{NP}}
\newcommand{\PP}{\mathsf{PP}}

\newcommand{\poly}{\operatorname{poly}}

\newcommand{\ket}[1]{\lvert #1\rangle}
\newcommand{\bra}[1]{\langle #1\rvert}
\newcommand{\braket}[2]{\langle #1 | #2\rangle}

\providecommand{\cmcQubits}{n}
\providecommand{\cmcDimension}{D}
\providecommand{\cmcKeyLength}{L}
\providecommand{\cmcSamples}{m}
\providecommand{\cmcErrorExponent}{s}

\providecommand{\cmcField}{\mathbb{F}_2}
\providecommand{\cmcComplex}{\mathbb{C}}
\providecommand{\cmcBitStrings}[1]{\{0,1\}^{#1}}
\providecommand{\cmcUnary}[1]{1^{#1}}
\providecommand{\cmcUnitaryGroup}[1]{\mathrm{U}(#1)}
\providecommand{\cmcHaar}{\mathrm{Haar}}
\providecommand{\cmcProbability}{\Pr}
\providecommand{\cmcExpectation}{\mathbb{E}}
\providecommand{\cmcDiag}{\operatorname{diag}}

\providecommand{\cmcKeySpace}{\mathcal{K}_{\cmcQubits}}
\providecommand{\cmcAdversary}{\mathcal{A}}
\providecommand{\cmcUnitary}{U}
\providecommand{\cmcClifford}{C}
\providecommand{\cmcMonomial}{M}
\providecommand{\cmcPermutation}{\pi}
\providecommand{\cmcPhase}{\omega}
\providecommand{\cmcHadamard}{H}
\providecommand{\cmcPhaseGate}{S}
\providecommand{\cmcCNOT}{\mathrm{CNOT}}
\providecommand{\cmcProjector}{\Pi}
\providecommand{\cmcTranspose}{\mathsf{T}}
\providecommand{\cmcAdjoint}{\dagger}

\providecommand{\cmcKet}[1]{\lvert #1\rangle}
\providecommand{\cmcBra}[1]{\langle #1\rvert}
\providecommand{\cmcBellLabel}[1]{\beta_{#1}}
\providecommand{\cmcAbs}[1]{\left\lvert #1\right\rvert}
\providecommand{\cmcNorm}[1]{\left\lVert #1\right\rVert}
\providecommand{\cmcVectorNorm}[1]{\cmcNorm{#1}_2}

\providecommand{\cmcLabelMap}[1]{T_{#1}}
\providecommand{\cmcCorrected}[2]{#1^{(#2)}}
\providecommand{\cmcAcceptance}[1]{p_{#1}}

\providecommand{\cmcPhi}{\lvert\Phi_{\cmcDimension}\rangle}
\providecommand{\cmcPauli}[1]{P_{#1}}
\providecommand{\cmcTransposePauli}[1]{Q_{#1}}
\providecommand{\cmcLinearMap}[1]{F_{#1}}
\providecommand{\cmcOffset}[1]{t_{#1}}
\providecommand{\cmcRange}[1]{R_{#1}}
\providecommand{\cmcSwap}{\mathsf{Swap}}
\providecommand{\cmcPsi}[1]{\cmcKet{\psi^{-}_{#1}}}

\providecommand{\cmcPauliX}{X}
\providecommand{\cmcPauliZ}{Z}
\providecommand{\cmcPauliGroup}{\mathcal{P}_{\cmcQubits}}
\providecommand{\cmcParity}[1]{q(#1)}

\providecommand{\cmcProbeLabel}{v_\star}
\providecommand{\cmcProbe}{\cmcKet{\cmcBellLabel{\cmcProbeLabel}}}
\providecommand{\cmcRising}[2]{\left(#1\right)^{\overline{#2}}}

\providecommand{\cmcDisp}[1]{\operatorname{disp}(#1)}
\providecommand{\cmcSector}[1]{\mathcal{H}_{#1}}

\hypersetup{colorlinks=true,linkcolor=MidnightBlue,citecolor=MidnightBlue,
  urlcolor=MidnightBlue,
  pdftitle={How Not to Build Cryptography in Microcrypt},
  pdfauthor={Aditya Gulati and Dakshita Khurana and Kabir Tomer}}

\title{\bfseries How Not to Build Microcrypt}

\date{}

\author{
Aditya Gulati\\ \normalsize UCSB \\ \normalsize \texttt{adityagulati@ucsb.edu} \and Dakshita Khurana\\ \normalsize UIUC and NTT Research \\ \normalsize \texttt{dakshita@illinois.edu}
\and
Kabir Tomer\\ \normalsize UIUC \\ \normalsize \texttt{kabirtomer@gmail.com}
}

\begin{document}
\maketitle

\begin{abstract}
A key challenge in quantum cryptography is to build quantum one-wayness and pseudorandomness without the use of (quantum computable) one-way functions. So far, this has turned out to be a difficult task, with only a few proposed candidates that are not directly built from one-way functions. 
Even within these few proposed candidates, we have lacked the techniques to analyze when proposed constructions may inadvertently yield one-way functions.

In this work, we introduce efficient \(\NP\)-aided shadow tomography of quantum states. 
We prove that collections of {\em computable} pure states can be learned via efficient \(\NP\)-aided shadow tomography, where we say that a state is computable if the amplitude and phase on any computational basis term can be classically efficiently computed given the description of an efficient preparation circuit for the state. 
We also give new algorithms for \(\NP\)-aided learning of unitaries given polynomially many queries to the unitary.

By building on this, we show that many existing architectures for PRS and PRU, including some that were explicitly introduced for the purposes of avoiding one-way functions (e.g., Hamiltonian Phase States, Bostanci et. al., TQC 2025), actually do imply the existence of one-way functions or imply \(\NP\) hardness. We hope that these no-go results will inform future investigations into building PRS and PRUs from assumptions that are plausibly outside the complexity class \(\NP\).

\end{abstract}

\newpage
\tableofcontents

\newpage
\section{Introduction}

Quantum cryptography has experienced a renaissance in recent years, fueled in part by the introduction of new computational primitives like pseudorandom states (PRS)~\cite{JLS}, pseudorandom unitaries (PRU)~\cite{JLS} and one-way state generators (OWSG)~\cite{MorimaeYamakawa} which underlie constructions of quantum money, commitments and encryption.

A central appeal of these objects is that unlike classical cryptography, they may exist even if (post-quantum) one-way functions do not exist~\cite{Kretschmer,KQST,KQT}.
An important design challenge is therefore to build quantum pseudorandom primitives from computational hardness assumptions that are plausibly \emph{weaker than} classical one-way functions: for example, problems that can remain unbroken even in the presence of an $\NP$ oracle (also known as Microcrypt). However, identifying such assumptions has proved difficult. 

In the several years since the inception of Microcrypt, only a handful of candidate assumptions have been proposed. Security for these primitives against NP attacks is poorly understood:  a prominent candidate pseudorandom state family, \emph{Hamiltonian Phase States}~\cite{HPS}, which were conjectured to rely on assumptions weaker than one-way functions. 
We show in this work that HPS pseudorandomness {\em does} imply a classical one-way function,
and therefore cannot realize cryptography in Microcrypt. 

Thus it is fair to say that existing approaches to building pseudorandom and one-way states have been ad-hoc, and there has been little systematic cryptanalysis of these candidates against \(\NP\)-aided adversaries. This has made it hard to justify or even {\em understand} whether the problem implies a one-way function. 

Our works aims to develop guardrails to help us better understand how to build cryptographic primitives in Microcrypt. 
We ask what \textit{natural} recipes {\em do not work}, that is, what kinds of constructions end up inadvertently implying one-way functions.\\

\noindent{\bf I. Using \(\NP\) Oracles to Invert One-way State Generators.}
Let us begin with the example of one-way state generators (OWSG), which are a keyed family of states where an efficient sampler outputs a state $\ket{\psi_k}$ given the key $k$. The security guarantee is that given $\ket{\psi_k}$ for random $k$, it is hard for quantum polynomial time machines to output $k'$ for which $\braket{\psi_{k'}}{\psi_k}$ is noticeable.

The promise of microcrypt is security even when \(\NP\) is easy, this means any proposed OWSG constructions in Microcrypt should remain secure even against adversaries that have access to \(\NP\) oracles.

Now suppose our OWSG adversary additionally has (quantum-query) access to a \emph{classical} $\NP$ oracle. Importantly, an $\NP$ oracle takes a classical description of an instance;  it \emph{cannot} meaningfully parse a general unknown quantum state.
Our approach to using this oracle is to try and turn the quantum challenge into a classical search problem. When can this succeed, i.e.,
\begin{center}
{\em When can a $\BQP^\NP$ algorithm invert a one-way state generator?}
\end{center}

This question is closely related to classical shadow tomography~\cite{Aaronsontomography,HKP20}. It is known that {\em efficient} measurements on many copies of a pure state can produce a compact classical record called a shadow. This shadow contains enough information to (inefficiently) estimate the challenge state's overlap with each candidate state in the family.

The remaining challenge is to use this information to recover a key for a state close to the challenge state. 
Our hope is that searching for such a key, given appropriately and efficiently generated shadows, becomes an \(\NP\) search task. 
We denote the process of 
recovering a key enabling efficient synthesis of a state close to the challenge as {\em inversion}, and ask:
\begin{center}
{\em Which families of states admit efficient shadows that can be inverted by $\NP$ oracles?}
\end{center}
Our first set of results prove that for a large family of states that we call {\em computable}, it is possible to prepare efficient shadows that can be inverted in $\NP$.
In more detail, we consider a family of states $\{\ket{\psi_k}\}_{k}$ indexed by keys, where the key enables efficient preparation of the corresponding state.
We say that the family is computable if the amplitude and phase on any computational basis term in $\ket{\psi_k}$ can be efficiently computed given the key $k$. We prove the following theorem.

\begin{theorem}(Informal)
    All OWSG families with computable states are invertible in $\mathsf{BQP}^{\NP}$.\footnote{Strictly speaking we need a stronger than usual correctness guarantee that is nevertheless satisfied by all proposed constructions.}
\end{theorem}

The above theorem only establishes that computable  OWSGs require {\em worst-case} \(\NP\) hardness, and does not yet provide an implication to one-way functions. 
%
Towards that goal, we obtain an implication to one-way functions from OWSGs that are {\em sampleable}, in addition to being computable. Informally, sampleability requires the ability to efficiently and deterministically sample from (approximately) the induced distribution $\mathcal{D}_k$ on computational basis terms obtained when measuring $\ket{\psi_k}$. This is different from simply measuring $\ket{\psi_k}$: we want an efficient classical process that on input randomness $r$ and key $k$ outputs a deterministic sample from (a distribution close to) $\mathcal{D}_k$.   

\begin{theorem}
    (Informal) Sampleable and computable OWSGs\footnote{We will assume for the purposes of the overview that OWSG keys are uniform. Our techniques can also handle keys that are not uniformly sampled, provided the key generation process is classically simulatable.} imply one-way functions.
\end{theorem}


Notable examples of sampleable OWSGs are Hamiltonian Phase States~\cite{HPS}, which prior to this work had been conjectured to not require one-way functions. Similarly, Morimae and Xagawa~\cite{QGA} proposed hardness assumptions for IQP-based quantum group actions as candidates for cryptography without one-way functions: we show that both imply one-way functions. 
\\

\noindent{\bf II. Using \(\NP\) Oracles to Break Pseudorandom Unitaries.} 

A pseudorandom {\em unitary} (PRU) refers to a family of efficient circuits implementing a unitary that cannot be distinguished from Haar random with polynomially many (quantum) oracle queries to the circuit. PRUs are perched atop the Microcrypt food chain, i.e., they imply many other cryptographic primitives, and yet, relative to certain oracles, they are known to exist without one-way functions. Obtaining pseudorandom unitaries from assumptions weaker than \(\NP\) hardness would therefore demonstrate that most of Microcrypt can be similarly concretely instantiated without the need for \(\NP\) hardness. This makes the problem of instantiating PRUs in Microcrypt one of the central problems in the area.

While there has been remarkable progress realizing pseudorandom unitaries from one-way functions themselves (or stronger assumptions)~\cite{BM24,PermHam24,PFC24,MH25,LRFC25,BlockedLRFC25,HiddenHam25,PC26,Kac26,Distinct26},  
\
there are no known concrete instantiations in the plain model from concrete assumptions weaker than one-way functions. It is plausible that existing methods and architectures to build PRUs, that implicitly or explicitly rely on one-way functions, could be secure even if the underlying building block is not a one-way function and satisfies a much weaker property. 
We study various known architectures to see which ones can potentially be instantiated without one-way functions, and we demonstrate generic ways to break pseudorandomness of some important architectures via an $\NP$ oracle.

First, we note that our attacks on pseudorandom states also apply to certain families of pseudorandom unitaries, for example, the $PFC$ unitary,  where $C$'s correspond to Clifford unitaries, $P$ denotes a computational-basis permutation and $F$ denotes a diagonal phase unitary (with arbitrary general phase). 
More importantly, we go beyond these attacks and show that one can also learn broader classes of unitary families.



\begin{theorem}(Informal)
    Any unitary family  $\{U_k\}_{k}$ where $U_k = C_{1,k}P_kF_kC_{2,k}$ 
    $C$'s correspond to Clifford unitaries, $P$ denotes a computational-basis permutation and $F$ denotes a diagonal phase unitary (with arbitrary general phase)
    can be learned (and therefore, also distinguished from Haar random) by a BQP machine with query access 
    to an \(\NP\) oracle.
\end{theorem}

This includes the first candidates that were proven to be secure PRUs~\cite{MH25,PFC} as well as more recent constructions~\cite{PC26,BlockedLRFC25,ConstantTime25,Distinct26}. Our theorem says that these candidates cannot be instantiated from weaker assumptions than \(\NP\) hardness, and one needs to look at fundamentally different architectures to obtain a Microcrypt PRU.

Perhaps the most important conceptual takeaway from this paper is a deeper understanding of what candidates and architectures survive \(\NP\) attacks, and can be plausible routes to constructions in Microcrypt.
Near the end of the paper, Table~\ref{tab:con-uncovered} lists several such constructions for which
our arguments do not provide attacks.



\subsection{Related Work}
Kretschmer~\cite[Appendix~C]{Kretschmer} gave an efficient quantum attack with an \(\NP\) oracle on the binary-phase PRS construction of Ji, Liu, and Song~\cite{JLS18} and Brakerski--Shmueli~\cite{BSBinary}. The attack measures copies of the state to obtain classical constraints on the phase function, then uses the oracle to test whether a key consistent with those constraints exists. Our results extend the study of \(\NP\)-aided attacks beyond binary-phase states to general computable pure-state families.

In a recently posted manuscript, Carrasco et al.~\cite{CEGHLS26} also establish the HPS-to-OWF implication using a state certification framework, which they instantiate for phase states whose computational basis probabilities are uniformly distributed. They also develop a general reduction from efficiently certifiable OWSGs to efficiently verifiable one-way puzzles and, under suitable classical simulability assumptions, a direct construction of one-way functions that avoids distributional inversion. 
Beyond the HPS-to-OWF implication, our general results appear to be incomparable, in particular we obtain one-way functions from OWSG beyond the flat setting: our \(\NP\)-aided tomography for arbitrary computable
pure-state families includes states with highly nonuniform amplitude
magnitudes. We further show that computable samplable
pure-state OWSG imply
OWF, again without requiring
flatness. Beyond the state setting, we also develop new techniques for \(\NP\)-aided learning of certain unitary families. 

Our techniques were developed concurrently and independently from~\cite{CEGHLS26}.
In particular, our results including the implication from  Search HPS to quantum-secure classical one-way functions and our NP-aided unitary learning algorithms, were presented
by Dakshita Khurana at the Simons Institute for the Theory of Computing
on July 14, 2026~\cite{KhuranaSimons26}.

\section{Technical Overview}

\subsection{\texorpdfstring{\(\NP\)}{NP}-aided inversion of computable OWSG}
Our first goal is to understand when access to \(\NP\) oracles (or, weaker one-way function inversion oracles) can help invert one-way state generators.

We begin by describing our $\NP$-aided tomography algorithm for
computable states. Let
\[
    |\psi_k\rangle
    =
    \sum_{x\in\{0,1\}^n} a_k(x)|x\rangle,
    \qquad k\in\mathcal K,
\]
be a family of computable states. Recall that the family is computable if the amplitude and phase on any computational basis term in $\ket{\psi_k}$ can be efficiently computed given the key $k$.
Given polynomially many copies of
$|\psi_k\rangle$ for an unknown key $k$, our goal is to recover some
$h\in\mathcal K$ such that $|\psi_h\rangle$ has high fidelity with
$|\psi_k\rangle$.

A natural starting point is classical-shadow tomography~\cite{HKP20}.
Classical shadows apply efficient measurements to copies of an unknown
state and produce a classical transcript from which many observables of
the state can later be estimated. In particular, such techniques can be
used to compare an unknown state against an exponentially large family
of candidate states. However, the resulting classical post-processing
need not be efficient: the only known upper bounds on the post-processing capability needed is $\mathsf{PP}$. In particular, any OWSG can be inverted by a quantum polynomial-time algorithm with classical access to a $\PP$ oracle~\cite{KT24,BoundedCopies25,HirokaHsieh24}. Such an oracle is believed to be much more powerful than an $\NP$ oracle. 
Our goal in this work is to exploit the additional structure of
\emph{computable} states to reduce the required classical power to
$\NP$.

The reason this may be possible is that, although the challenge state itself is
quantum, a substantial amount of information about it may still be
classically computable given the key. In particular, for computable state
families we can efficiently evaluate the complex amplitude
\[
    a_k(x)=\langle x|\psi_k\rangle
\]
on any computational-basis vector $x$.

This suggests a natural strategy: extract enough amplitudes from copies of the
unknown state, and then use an $\NP$ oracle to search for a key whose amplitudes
match the extracted values. The difficulty is that there is no measurement
which simply reveals the complex amplitude $a_k(x)$ of an unknown quantum
state. Measurement outcomes are probabilities, and in particular direct
computational-basis measurements are correlated with $|a_k(x)|^2$ rather than $a_k(x)$,
discarding its phase.

We therefore want to replace the goal of extracting amplitudes themselves by the more
attainable goal of obtaining samples from measurement distributions whose
probabilities are efficiently computable given a candidate key. 

In our approach, we first
measure copies of the challenge state to obtain a classical transcript. For each
candidate key $h$, we can then classically compute the probability that
$|\psi_h\rangle$ would have produced this transcript. An $\NP$ oracle can
therefore search over the exponentially many candidate keys for one whose
predicted transcript distribution best explains the observations.

The simplest example is computational-basis measurement. Define
\[
    p_k(x):=|a_k(x)|^2.
\]
By amplitude computability, $p_h(x)$ can be efficiently approximated
for every candidate key $h$ and basis string $x$. We measure several
copies of the challenge in the computational basis, obtaining
\[
    X_1,\ldots,X_s,
\]
and use classical $\NP$ queries to find a key $a$ maximizing the
likelihood
\[
    \prod_{i=1}^s p_a(X_i).
\]
Using the properties of maximum likelihood estimation it can be shown that, with
sufficiently many samples, the resulting distribution $p_a$ is close
to the true distribution $p_k$.

This is unfortunately not enough to learn the quantum state. Computational-basis
measurement discards all phase information. For example, all states
of the form
\[
    |\psi_k\rangle
    =
    \frac1{\sqrt D}
    \sum_x e^{i\theta_k(x)}|x\rangle
\]
have the same uniform computational-basis distribution, even though
different choices of the phases can produce nearly orthogonal states.
We therefore need a second measurement that extracts information about
relative phases.

A natural starting point is the phase-extraction procedure used in
\cite{KT25}. There, one chooses a random shift $\Delta$ and applies a
two-to-one function satisfying
\[
    f(x)=f(x\oplus\Delta).
\]
Measuring the function value collapses the state to be supported on a random pair
$x,y=x\oplus\Delta$, attaining the form
\[
    \cos(\vartheta/2)|x\rangle
    +
    \sin(\vartheta/2)e^{i\phi}|y\rangle.
\]
for some angles $\vartheta$ and $\phi$.
Measuring this state in the two bases
\[
    \{|x\rangle\pm|y\rangle\}
    \qquad\text{and}\qquad
    \{|x\rangle\pm i|y\rangle\}
\]
reveals respectively information about $\cos\phi$ and $\sin\phi$.
Thus, in principle, such measurements expose the relative phase
between two computational-basis amplitudes.

There is, however, an important obstruction. The strength of the phase
``signal" is proportional to
\[
    \sin\vartheta
    =
    \frac{2|a(x)a(y)|}{|a(x)|^2+|a(y)|^2}.
\]
If one amplitude is much larger than the other, this quantity is small,
and changing their relative phase barely changes the measured state.
A prior construction in \cite{KT25} deals with a similar problem by first applying a
random unitary $2$-design, which flattens the computational-basis
weights sufficiently that the two amplitudes are usually not too
imbalanced. We cannot directly use this solution: applying a generic
randomizing unitary to the challenge may destroy the amplitude
computability on which our subsequent $\NP$ search relies.

This suggests looking for a way to obtain the same balancing effect
without scrambling the challenge, inspired by \cite{Ben21}. 
We will create a reference state and force our unknown OWSG output state to interfere with the reference state.
Write the unknown challenge state as
\[
    |\psi\rangle
    =
    \sum_{z\in\{0,1\}^n}\sqrt{p(z)}e^{i\theta(z)}|z\rangle,
\]
where $p(z)$ is its computational-basis probability distribution and
$\theta(z)$ is its phase at $z$. As a first attempt, use the uniform state
\[
    |+\rangle^{\otimes n}
    =
    \frac1{\sqrt D}\sum_z|z\rangle,
    \qquad D=2^n,
\]
as a reference. Prepare an ancilla in $|+\rangle$ and apply a controlled
SWAP between the challenge and reference registers. Then measure both
registers in the computational basis, obtaining $x$ in the first register
and $y$ in the second. Conditioned on these outcomes, the ancilla state
is proportional to
\[
    \sqrt{p(x)}e^{i\theta(x)}|0\rangle
    +
    \sqrt{p(y)}e^{i\theta(y)}|1\rangle.
\]
This already improves on choosing an unrelated pair: the pair
$(x,y)$ has distribution
\[
    \frac{p(x)+p(y)}{2D}.
\]
Equivalently, one endpoint is sampled according to the target
distribution $p$, the other is sampled uniformly, and their order is
randomized. Thus the measurement automatically sees coordinates on
which the state has mass.

But the main problem remains. If $p$ is highly nonuniform, the endpoint
sampled uniformly will typically have much smaller target amplitude
than the endpoint sampled from $p$. The two branches of the ancilla
are again badly imbalanced, and consequently contain very little phase
information. For a state supported roughly uniformly on only $S$
basis vectors, for example, the useful phase signal obtained from the
uniform reference can be suppressed by a factor on the order of
$S/D$, which may be exponentially small.

The uniform reference nevertheless suggests the right generalization:
instead of fixing its magnitude distribution to be uniform, let us use
a reference
\[
    |q^+\rangle
    =
    \sum_x\sqrt{q(x)}|x\rangle
\]
for some distribution $q$. After the controlled SWAP and computational
basis measurement of the two state registers, the unnormalized ancilla
state is
\[
    \sqrt{p(x)q(y)}e^{i\theta(x)}|0\rangle
    +
    \sqrt{q(x)p(y)}e^{i\theta(y)}|1\rangle.
\]
The pair $(x,y)$ now has marginal distribution
\[
    \frac{p(x)q(y)+q(x)p(y)}2.
\]
Equivalently, one string is sampled from $p$ and one from $q$, with
their order randomized.

The failure of our previous example reveals what properties we want from the reference. We want it to place mass on the same coordinates as the challenge and, more
importantly, to balance the two branches of the interference
experiment. The ideal choice, therefore, is to choose $q=p$.
In this case the observed pair satisfies
\[
    (x,y)\sim p\times p,
\]
and the ancilla state becomes
\[
    \sqrt{p(x)p(y)}
    \left(
        e^{i\theta(x)}|0\rangle
        +
        e^{i\theta(y)}|1\rangle
    \right).
\]
The two branches now have exactly equal magnitude for every observed
pair. After normalization, the ancilla is simply
\[
    \frac{
        e^{i\theta(x)}|0\rangle
        +
        e^{i\theta(y)}|1\rangle
    }{\sqrt2}.
\]
Measuring the ancilla in equatorial bases discussed earlier allows us to extract information about
the cosine and sine of the relative phase
$\theta(x)-\theta(y)$. Moreover, both $x$ and $y$ are drawn according
to the actual mass of the state. Thus this procedure works equally
well for flat states and states with highly nonuniform/spiky probability distributions. In our proofs we show that if two states produce similar distributions (when measured as above with their corresponding ideal reference states) they must have high fidelity. Since this measurement distribution is also computable, it therefore suffices to measure challenge states to obtain samples from this ideal distribution.

The reader may wonder if the procedure discussed above can really work, since preparing the necessary reference state
appears to require knowledge of the unknown key $k$. But this is
exactly what the first, computational-basis stage of our algorithm has
already solved! That stage returns a key ${k_0}$ for which
\[
    p_{k_0}\approx p_k.
\]
Since ${k_0}$ is known, computability lets us prepare its state $\ket{\psi_{{k_0}}}$ and then uncompute its phases to obtain
\[
    |\psi_{{k_0}}^+\rangle
    =
    \sum_x|a_{{k_0}}(x)||x\rangle.
\]

Our algorithm therefore proceeds as follows: The first stage finds a reference whose
magnitudes approximately match those of the challenge. We then use the reference to sample from the distribution described above, followed by one last search for the key that achieves the maximum likelihood for the final samples.

Finally, for the special case where the distributions are also classically samplable given the key, the same
two-stage strategy yields a classical one-way function. Namely, we
construct classical functions that generate the two measurement
transcripts using the random key and some additional random coins. A
distributional inverter replaces each $\NP$ maximum-likelihood search:
given a transcript, it approximately samples a key conditioned on
producing that transcript. If such inversion were easy, the same argument  as above would break the one-way state generator;
hence the transcript function is distributionally one-way and therefore
implies a classical one-way function. To enable efficient computation of our one-way function, we end up having to assume that the OWSG is {\em sampleable}, that is, it is possible to efficiently and classically sample from the distribution that assigns probabilities 
\[  p_k(x):=|a_k(x)|^2.
\]
to each $x$.
These conditions already cover concrete candidates proposed for quantum
cryptography without classical one-way functions: Hamiltonian phase
states~\cite{HPS} and the IQP group-action states of~\cite{QGA}, after
applying a public Hadamard gate to each qubit of the latter.
Both families then have uniform amplitude magnitudes and phases that
can be efficiently computed from the classical key by summing the
contributions of polynomially many diagonal gates; hence their
amplitudes are efficiently computable, and their measurement
distributions in this basis can be sampled exactly using one
independent fair bit per qubit.

\subsection{An \texorpdfstring{\(\NP\)}{NP}-aided attack on CMC constructions}

Next, we consider unitary constructions of the form
\[
    U_k=C_{2,k}M_kC_{1,k},
    \qquad
    M_k\ket{x}=\omega_k(x)\ket{\pi_k(x)},
\]
where $C_{1,k},C_{2,k}$ are Clifford unitaries and $M_k$ is
{\em monomial}: it permutes computational-basis states and adds phases. An example is the construction $C_2PFC_1$
from~\cite{MH25}, whose middle layer $PF$ is monomial.

Our \(\NP\)-aided tomography attack described in the previous subsection already helps distinguish $PFC$ unitaries from random, because the state $PFC \ket{0}$ is computable for every known permutation, function and Clifford.
Adding the output Clifford $C_2$, however, can mix the
amplitudes and destroy the computability that attack
relies on. We therefore look for a different property
of the construction that survives both Clifford layers and can be exploited by an \(\NP\) oracle.

We build up to our approach by first understanding a specific way to distinguish unitaries of the form 
\[
    U_k=M_k,
    \qquad
    M_k\ket{x}=\omega_k(x)\ket{\pi_k(x)},
\]
for Haar random. Let us see what happens
when $M_k$ is applied to both registers of a Bell state. 
Let $D=2^n$. A Bell state (indexed by $(a,c)$) has the form
\[
    \ket{\beta_{a,c}}
      =\frac{1}{\sqrt D}\sum_x(-1)^{c\cdot x}\ket{x,x\oplus a}.
\]
Applying $M_k$ to both registers gives
\[
    (M_k\otimes M_k)\ket{\beta_{a,c}}
      =\frac{1}{\sqrt D}\sum_x
        (-1)^{c\cdot x}\omega_k(x)\omega_k(x\oplus a)
        \ket{\pi_k(x),\pi_k(x\oplus a)}.
\]
If we now measure both registers in the computational basis,
the XOR of the two observed strings must satisfy the relation
\[
    R_{\pi_k}(a)
      =\{\Delta_{k,a} = (\pi_k(x)\oplus\pi_k(x\oplus a)) :x\in\{0,1\}^n\}.
\]
We call these XORed values $\Delta_{k,a}$ the allowed displacements.
Note that replacing the index $x$ by $x\oplus a$ simply swaps
the two strings in the XOR, leaving the displacement
unchanged. We will assume $a\neq 0$, then these are distinct indices,
so every displacement arises from at least two out of the
$D$ indices. Hence $R_{\pi_k}(a)$ contains at most
$D/2$ distinct displacements.

For a Haar-random unitary, the displacement is uniform
over the $D-1$ nonzero strings when averaged over the
choice of unitary. Thus, any fixed set of at most $D/2$
displacements misses a constant fraction of the
measurement probability.

This means that we can repeat the experiment with the same $(a,c)$ value on fresh copies
of the same Bell state. In the monomial case, every
outcome remains in its allowed set, whereas in the Haar case, the probability that all outcomes remain in any
fixed candidate set decreases exponentially with the
number of repetitions.



The remaining obstacle is checking whether the observed
displacements belong to an allowed set.
The key is unknown, and even given a candidate key $h$,
listing $R_{\pi_h}(a)$ may take exponential time.
However, membership has a short classical witness:
to certify that $b\in R_{\pi_h}(a)$, it suffices to give
a string $x$ satisfying
$\pi_h(x)\oplus\pi_h(x\oplus a)=b$.
An \(\NP\) oracle can therefore check whether there exist
one common key $h$ and strings $x_1,\ldots,x_m$ explaining
all the observed displacements.

This is the key idea:
we have found a measurement whose
outcomes, for the case of PRU, have classical
certificates that an \(\NP\) oracle can use to test consistency
with the family.

To handle our actual unitaries $U_k$ that have the two secret Clifford layers, we replace
computational-basis measurements with Bell-basis measurements.
A Bell outcome has two labels $(b,d)$, where $b$ records
a displacement. For each fixed $b$, the Bell states
$\{\ket{\beta_{b,d}}\}_d$ form another orthonormal basis
for the subspace
\[
    \operatorname{span}\{\ket{y,y\oplus b}:y\in\{0,1\}^n\}.
\]
So, if we ignore $d$, Bell measurements give exactly the
same displacement distribution as measuring both
registers and taking their XOR. Our restriction on
allowed displacements therefore still applies.


But Bell measurements have another useful feature: {\em applying
the same Clifford to both registers simply relabels
the Bell states}. To see why, write a Bell state as
$(I\otimes P)\ket{\Phi}$, where $P$ is a Pauli and
$\ket{\Phi}=D^{-1/2}\sum_x\ket{x,x}$. Then
\[
    (C\otimes C)(I\otimes P)\ket{\Phi}
      =(I\otimes CPC^T)\ket{\Phi}.
\]
Here $CPC^T$ is a Pauli.\footnote{
Write $CPC^T=(CPC^\dagger)(CC^T)$.
Clifford conjugation maps Paulis to Paulis, and
$CC^T$ is also a Pauli up to phase.}
The resulting state is therefore another Bell state. Given the Clifford description, one can
efficiently compute this relabeling and undo the relabeling
classically. Moreover, this relabeling is a bijection on all Bell labels, so our set of potential displacements (after inverting the bijection) is still atmost $D/2$. 

Now consider a candidate key $h$ for
$U_h=C_{2,h}M_hC_{1,h}$.
Starting with the Bell state $\ket{\beta_{a,c}}$, the first Clifford
layer would produce
\[
    (C_{1,h}\otimes C_{1,h})\ket{\beta_{a,c}}
      =\ket{\beta_{a_h,c_h}},
\]
where $(a_h,c_h)$ is classically computable from $(h,a,c)$.
Let $z_i$ be a Bell measurement of $(U_{h}\otimes U_{h})\ket{\beta_{a,c}}$, we can also run
the output Clifford $C_{2,h}$'s relabeling backwards to obtain
a {\em corrected} label $(b_i^{(h)},d_i^{(h)})$. If $h$ is the true key, this corrected outcome has exactly
the distribution obtained by applying $M_h$ to both
registers of $\ket{\beta_{a_h,c_h}}$ and then measuring.
We can therefore check whether the candidate permutation
allows the corrected displacement. 

The full attack is as follows: measure $m$ Bell outcomes and ask the \(\NP\) oracle whether one candidate key can explain all of them.
Specifically, the \(\NP\) query asks whether there exist
a key $h$ and strings $x_1,\ldots,x_m$ such that
\[
    \pi_h(x_i)\oplus\pi_h(x_i\oplus a_h)=b_i^{(h)}
    \qquad\text{for every }i.
\]
We chose $\ket{\beta_{a,c}}$ to ensures $a_h\neq0$ for every candidate,
so each candidate allows at most $D/2$ displacements.
For a genuine CMC unitary, the true key
explains every outcome, so the attack accepts with
probability one.

In the Haar case, for any candidate key, our measurement lands outside the set with constant probability. 
Hence with repeated measurements, we can rule out any fixed
candidate key with high probability. Also, increasing the number of repetitions makes the probability
of a candidate surviving exponentially small. With sufficiently
many repetitions, this probability is small enough to union
bound over all candidate keys. Thus, with high probability,
no key explains the entire measurement record, and the
\(\NP\) query rejects.

The attack therefore distinguishes CMC constructions from
Haar using polynomially many nonadaptive forward queries
and a single classical \(\NP\) query. In fact, the same attack also learns the unitary: NP search-to-decision reductions will {\em find} a key whose unitary is close to the desired unitary.

This attack applies in particular to unitaries of the form
$C_2PFC_1$ (which are strong secure PRU~\cite{MH25,LRFC25}) and Luby-Rackoff-Function-Clifford unitaries
$C_2S_RF S_LC_1$ (which are also strong PRUs~\cite{LRFC25}). Both of these constructions have middle layers that are
monomial. We discuss all other constructions that fall into this category in
Section~\ref{subsec:con-cmc}.

\section{Preliminaries}
\label{sec:prelims}
For a random variable $X$, write $\mathcal L(X)$ for its distribution.
We write
\[
    T(\rho,\sigma):=\frac12\|\rho-\sigma\|_1
\]
for trace distance. For pure states,
\[
    T\bigl(|u\rangle\langle u|,|v\rangle\langle v|\bigr)
    =\sqrt{1-|\langle u|v\rangle|^2}
    \le\bigl\|\,|u\rangle-|v\rangle\,\bigr\|_2.
\]

We relate distances between distributions to inner products and
Euclidean distances between their qsamples.

\begin{definition}[Qsamples]
\label{def:qsample}
Let $P$ be a probability distribution on a finite set
$\mathcal{Z}$. Its \emph{qsample} is
\[
    |P^+\rangle
    :=\sum_{z\in\mathcal{Z}}\sqrt{P(z)}\,|z\rangle.
\]
This state is normalized and has real, nonnegative amplitudes.

\end{definition}

\begin{definition}[Basis distribution]
\label{def:basis-distribution}
For a normalized pure state
$|\psi\rangle=\sum_{z\in\mathcal{Z}}a(z)|z\rangle$, its
\emph{basis distribution} is the probability distribution
\[
    p_\psi(z):=|a(z)|^2.
\]
\end{definition}

\begin{definition}[Positive-amplitude state]
\label{def:positive-amplitude-state}
For a normalized pure state
$|\psi\rangle=\sum_{z\in\mathcal{Z}}a(z)|z\rangle$, its
\emph{positive-amplitude state} is the qsample of its
basis distribution:
\[
    |\psi^+\rangle
    :=|p_\psi^+\rangle
    =\sum_{z\in\mathcal{Z}}|a(z)|\,|z\rangle.
\]
\end{definition}

\begin{definition}[Distribution distances]
\label{def:distribution-distances}
For probability distributions $P,Q$ on the same finite set
$\mathcal{Z}$, define their \emph{Bhattacharyya coefficient},
\emph{Hellinger distance}, and \emph{total variation distance} by
\begin{align*}
    B(P,Q)
    &:=\sum_{z\in\mathcal{Z}}\sqrt{P(z)Q(z)},\\
    H(P,Q)
    &:=\left(
        \sum_{z\in\mathcal{Z}}
        \bigl(\sqrt{P(z)}-\sqrt{Q(z)}\bigr)^2
    \right)^{1/2},\\
    \operatorname{TV}(P,Q)
    &:=\frac12\sum_{z\in\mathcal{Z}}|P(z)-Q(z)|.
\end{align*}
\end{definition}

\begin{lemma}[Qsample identities]
\label{lem:qsample-identities}
For probability distributions $P,Q$ on the same finite set,
\begin{align*}
    B(P,Q)&=\langle P^+|Q^+\rangle,\\
    H(P,Q)&=\bigl\|\,|P^+\rangle-|Q^+\rangle\,\bigr\|_2,\\
    H^2(P,Q)&=2-2B(P,Q).
\end{align*}
In particular, $H$ satisfies the triangle inequality.
\end{lemma}

\begin{proof}
The first two identities follow by expanding in the standard
basis. The third follows by expanding the squared norm and using
normalization. The triangle inequality follows from that of the
Euclidean norm.
\end{proof}

\begin{lemma}[Total variation versus Hellinger distance]
\label{lem:tv-hellinger}
For probability distributions $P,Q$ on the same finite set,
\[
    \operatorname{TV}(P,Q)\le H(P,Q).
\]
\end{lemma}

\begin{proof}
Factoring each difference and applying Cauchy--Schwarz gives
\begin{align*}
    \operatorname{TV}(P,Q)
    &=\frac12\sum_z
      \bigl|\sqrt{P(z)}-\sqrt{Q(z)}\bigr|
      \bigl(\sqrt{P(z)}+\sqrt{Q(z)}\bigr)\\
    &\le\frac12 H(P,Q)\sqrt{2+2B(P,Q)}
    \le H(P,Q).
\end{align*}
The last inequality uses $B(P,Q)\le1$, also by Cauchy--Schwarz.
\end{proof}

\begin{definition}[Monomial layer]
\label{def:monomial-layer}
A unitary $M$ on $n$ qubits is a \emph{monomial layer} if there exist
a permutation $\pi:\mathbb{F}_2^n\to\mathbb{F}_2^n$ and a function
$\omega:\mathbb{F}_2^n\to\mathbb{C}$ with $|\omega(x)|=1$ such that
\[
    M|x\rangle=\omega(x)|\pi(x)\rangle
    \qquad\text{for every }x\in\mathbb{F}_2^n.
\]
Here we identify $\mathbb{F}_2^n$ with the computational-basis labels
$\{0,1\}^n$. Thus $M$ permutes computational-basis states and multiplies
them by phases. 
\end{definition}
\begin{definition}[Pure-state one-way state generator]
\label{def:owsg}
A pure-state one-way state generator is a tuple of QPT algorithms $(\mathsf{KeyGen},\mathsf{StateGen},\mathsf{Ver})$.
The algorithm $\mathsf{KeyGen}(1^n)$ outputs a classical key
$k\in\mathcal K_n := \mathsf{Supp}(\mathsf{KeyGen}(1^n))$.  For every $k\in\mathcal K_n$,
$\mathsf{StateGen}(1^n,k)$ prepares a pure state
$|\psi_{n,k}\rangle$.  The verifier $\mathsf{Ver}(1^n,h,\rho)$
takes a candidate key $h$ and a state $\rho$ and outputs a bit.
The algorithms satisfy the following conditions.

\begin{enumerate}
    \item \emph{Correctness.}
    For some negligible function $\nu$,
    \[
        \Pr_{k\sim\mathsf{KeyGen}(1^n)}
        \left[
                \mathsf{Ver}
                (1^n,k,\psi_{n,k})=1
        \right]
        \ge 1-\nu(n).
    \]

    \item \emph{One-wayness.}
    There exists a negligible function $\mu$ such that for every polynomial $t=t(n)$ and every QPT adversary
    $\mathcal A$, for large enough $n$, 
    \[
        \Pr\!\left[
            \mathsf{Ver}
            (1^n,k',\psi_{n,k})=1
        \right]
        \le \mu(n),
    \]
    where
    \[
        k\leftarrow\mathsf{KeyGen}(1^n),
        \qquad
        k'\leftarrow
        \mathcal A(1^n,\psi_{n,k}^{\otimes t(n)}).
    \]
\end{enumerate}
\end{definition}
\section{Computable and Samplable States}
\label{owsg:model}

Here $n$ denotes the security parameter. All algorithms are uniform.

\begin{definition}[Computable states]
\label{def:computable-states}
Let $\ell:\mathbb{N}\to\mathbb{N}$ be a polynomially bounded,
polynomial-time computable function. For each $n\ge1$, let
$\mathcal{K}_n\subseteq\{0,1\}^{\ell(n)}$ be nonempty, with membership
decidable in polynomial time.
Consider a family
\[
    \Psi=\{\Psi_n\}_{n\ge1},
    \qquad
    \Psi_n=\{|\psi_{n,k}\rangle\}_{k\in\mathcal{K}_n}
\]
of normalized $n$-qubit pure states, with amplitudes
$a_{n,k}(x):=\langle x|\psi_{n,k}\rangle$. Note that we fix a global phase for each state.

We call $\Psi$ \emph{computable} if the following conditions hold.
\begin{enumerate}
    \item \textbf{Quantum generation.}
    A quantum polynomial-time algorithm $\mathsf{Gen}$, on input
    $(1^n,k)$, outputs
    $|\psi_{n,k}\rangle$ for every
    $n\ge1$ and $k\in\mathcal{K}_n$.

    \item \textbf{Classical amplitude computation.}
    A deterministic classical algorithm $\mathsf{Amp}$, on input
    $(k,x,1^b)$, runs in polynomial time
    and outputs a complex number $\widetilde a$  satisfying
    \[
        |\widetilde a-a_{n,k}(x)|\le2^{-b}
    \]
    for every $n,b\ge1$, $k\in\mathcal{K}_n$, and
    $x\in\{0,1\}^{n}$.
\end{enumerate}
\end{definition}

\begin{definition}[Computable samplable states]
\label{def:computable-samplable-states}
A computable state family $\Psi$ is \emph{computable samplable}
if there exists a PPT algorithm $\mathsf{Samp}$
with the following property. For every $n,t\ge1$ and
$k\in\mathcal{K}_n$, the algorithm $\mathsf{Samp}(1^n,k,1^t)$ outputs a string in
$\{0,1\}^{n}$ whose distribution $\widetilde p_{n,k,t}$ satisfies
\[
    \operatorname{TV}(\widetilde p_{n,k,t},p_{n,k})\le\frac1t,
    \qquad
    p_{n,k}(x):=p_{\psi_{n,k}}(x)=|a_{n,k}(x)|^2.
\]
\end{definition}

\begin{definition}[Computable and samplable OWSGs]
\label{def:structured-owsg}
Let
$\Pi=(\mathsf{KeyGen},\mathsf{StateGen},\mathsf{Ver})$
be a pure-state OWSG as in Definition~\ref{def:owsg}, with output family
$\Psi=\{\Psi_n\}_{n\ge1}$, where
\[
    \Psi_n
    =
    \{|\psi_{n,k}\rangle:k\in\mathcal K_n\}.
\]
We call $\Pi$
\begin{enumerate}
    \item a \emph{computable OWSG} if $\Psi$ is \textit{computable} in the sense
    of Definition~\ref{def:computable-states};
    \item a \emph{computable samplable OWSG} if $\Psi$ is \textit{computable samplable} in the sense
    of Definition~\ref{def:computable-samplable-states} and $\mathsf{KeyGen}$ is a classical PPT algorithm.
\end{enumerate}
\end{definition}
 \section{\texorpdfstring{$\NP$}{NP}-aided tomography of computable states}

\label{sec:computable-pru-avgnp}
\label{owsg:section}



In this section we show efficient $\NP$-aided  tomography algorithms for computable states. The algorithms receive polynomially many copies of an unknown computable state $\ket{\psi_k}$ indexed by $k$ in keyspace $\mathcal{K}$ and output $h \in \mathcal{K}$ such that $\ket{\psi_k}$ and $\ket{\psi_h}$ have high overlap. The algorithms  only make classical queries to $\NP$.

\begin{theorem}[\(\NP\)-oracle recovery of a high-fidelity key]
\label{thm:owsg:learning}
Let $\Psi$ be a computable state family (Definition~\ref{def:computable-states}).
There exists a uniform quantum algorithm $\mathsf{Learn}$ with
only classical access to an $\NP$ oracle satisfying the following
guarantee.

For every $n\ge 1$, $k\in\mathcal{K}_n$, and accuracy pair
$0<\varepsilon,\delta<1/4$, the algorithm receives
$(1^n,\varepsilon,\delta)$ and copies of $|\psi_{n,k}\rangle$.
It uses
\[
    O\!\left(
        \frac{\ell(n)+\log(1/\delta)}{\varepsilon}
    \right)
\]
copies, outputs a key $h\in\mathcal{K}_n$ such that
\begin{equation}
    \Pr\!\left[
        |\langle\psi_{n,h}|\psi_{n,k}\rangle|^2
        \ge 1-\varepsilon
    \right]
    \ge 1-\delta.
    \label{owsg:eq-9}
\end{equation}
and runs in time $
    \operatorname{poly}\!\left(
        n,\frac1\varepsilon,\log\frac1\delta
    \right)
$.
\end{theorem}

\noindent Before describing the algorithm $\mathsf{Learn}$ and proving Theorem \ref{thm:owsg:learning}, we first introduce some notation and describe procedures that will be used as subroutines.
Fix the security parameter $n$ and write
\[
    \mathcal{K}=\mathcal{K}_n,\qquad L=\ell(n),\qquad
    \mathcal{X}=\{0,1\}^{n},\qquad D=|\mathcal{X}|=2^{n}.
\]
We suppress $n$ in the notation when clear from context. As in
Definitions~\ref{def:basis-distribution} and
\ref{def:positive-amplitude-state}, for $h \in \mathcal{K}$, for $a_{h}(x):=\langle x|\psi_{h}\rangle$,
\begin{align}
    p_h(x)&:=p_{\psi_h}(x)=|a_h(x)|^2,
    \label{owsg:eq-2}\\
    |\psi_h^+\rangle&=|p_h^+\rangle
      =\sum_{x\in\mathcal{X}}|a_h(x)|\,|x\rangle.
    \label{owsg:eq-3}
\end{align}


\paragraph{Controlled-swap measurement.}
The measurement $\mathcal{M}$ takes two registers $A,B$ supported on $\mathcal{X}$ and performs the following steps.
\begin{enumerate}
    \item Prepare an ancilla $C$ in $|+\rangle$. Apply a SWAP between
    $A$ and $B$, controlled on $C$ being $1$.
    \item Choose $\alpha\in\{0,\pi/2\}$ uniformly. Measure $C$ in the
    orthonormal basis
    \[
        \left\{
            \frac{|0\rangle+ e^{i\alpha}|1\rangle}{\sqrt2}, \frac{|0\rangle-e^{i\alpha}|1\rangle}{\sqrt2}
        \right\},
    \]
    recording the outcome $\sigma \in \{\pm 1\}$.
    \item Measure $A,B$ in the computational basis, obtaining the
    ordered pair $(x,y)$. Output $(\alpha,\sigma,x,y)$.
\end{enumerate}
The outcome space is
\[
    \mathcal{Z}
    :=\{0,\pi/2\}\times\{-1,1\}\times\mathcal{X}\times\mathcal{X}.
\]
For a normalized state $|\psi\rangle$ on $\mathcal{X}$ and a
probability distribution $q$ on $\mathcal{X}$, let $R_{\psi,q}$ be the
output distribution of $\mathcal{M}$ on
$|\psi\rangle_A\otimes|q^+\rangle_B$.
For $j,k_0\in\mathcal{K}$, define
\begin{equation}
    R_{j,k_0}:=R_{\psi_j,p_{k_0}}.
    \label{owsg:eq-7}
\end{equation}

\begin{lemma}[Outcome distribution]
\label{owsg:outcome-distribution}
For $|\psi\rangle=\sum_x\psi(x)|x\rangle$ and
$(\alpha,\sigma,x,y)\in\mathcal{Z}$,
\begin{equation}
    R_{\psi,q}(\alpha,\sigma,x,y)
    =\frac18\left|
        \psi(x)\sqrt{q(y)}
        +\sigma e^{-i\alpha}\sqrt{q(x)}\psi(y)
    \right|^2.
    \label{owsg:born}
\end{equation}
In particular, for $j,k_0\in\mathcal{K}$,
\begin{equation}
    R_{j,k_0}(\alpha,\sigma,x,y)
    =\frac18\left|
        a_j(x)|a_{k_0}(y)|
        +\sigma e^{-i\alpha}|a_{k_0}(x)|a_j(y)
    \right|^2.
    \label{owsg:eq-8}
\end{equation}
\end{lemma}

\begin{proof}
After the controlled SWAP, the joint state is
\[
    \frac{
        |0\rangle|\psi\rangle|q^+\rangle
        +|1\rangle|q^+\rangle|\psi\rangle
    }{\sqrt2}.
\]
For fixed $\alpha$, its amplitude for outcome $(\sigma,x,y)$ is
\[
    \frac12\left(
        \psi(x)\sqrt{q(y)}
        +\sigma e^{-i\alpha}\sqrt{q(x)}\psi(y)
    \right).
\]
Squaring the modulus and multiplying by the probability $1/2$ of
choosing $\alpha$ proves \eqref{owsg:born}. Substituting
$\psi=\psi_j$ and $q=p_{k_0}$ gives \eqref{owsg:eq-8}.
\end{proof}
\paragraph{Learning Algorithm.}
On input $(1^n,\varepsilon,\delta)$ and copies of an unknown state,
$\mathsf{Learn}^{\NP}$ uses the following parameters:
\begin{align}    s&:=\left\lceil\frac{800}\varepsilon
        \left(L+\left\lceil\log_2\frac{16}{\delta}\right\rceil\right)
        \right\rceil,
    \label{owsg:precision}\\
    \eta_1&:=\frac{\delta}{24sD^s},\qquad
    \eta_2:=\frac{\delta}{24s(4D^2)^s},
    \label{owsg:probability-tolerances}\\
    \zeta&:=\frac{\delta}{6s}.
    \label{owsg:eq-86}
\end{align}
The algorithm uses exactly $2s$ input copies.
\begin{enumerate}
    \item \textbf{Find match for basis distribution.}
    Measure the first $s$ copies in the computational basis, obtaining
    $X_1,\ldots,X_s$. For any $j,x$ let
    $\widetilde p_j(x) \in [0,1]$ represent $\pm\eta_1$ approximations of $p_j(x)$ computed using $\mathsf{Amp}$, i.e. \[
    |\widetilde p_j(x)-p_j(x)|\le \eta_1.
    \]
    Use classical $\NP$ queries to binary search and find the lexicographically first
    maximizer $k_0\in\mathcal{K}$ of
    \begin{equation}
        \prod_{i=1}^s\widetilde p_j(X_i)
        \qquad\text{over }j\in\mathcal{K}.
        \label{owsg:eq-12}
    \end{equation}

    \item \textbf{Prepare reference states.}
Prepare $|\psi_{k_0}\rangle$ using $\mathsf{Gen}(1^n,k_0)$.
Coherently for each $x$ in the superposition, compute amplitude approximations $\widetilde a_{k_0}(x)$ satisfying
\[
    |\widetilde a_{k_0}(x)-a_{k_0}(x)|\le\frac{\zeta}{16\sqrt D},
\]
 apply the phase $e^{-i\arg (\widetilde a_{k_0}(x))}$ (with $\arg(0):=0$), and
uncompute the auxiliary registers. Implement this phase operation
to operator-norm error at most $\zeta/2$. As we show later, this allows the algorithm to obtain trace distance $\zeta$ approximations of
$|\psi_{k_0}^+\rangle$ which we refer to as reference states.

    \item \textbf{Collect measurement $\mathcal{M}$ outcomes.}
    For each of the remaining $s$ copies, independently prepare a fresh
    reference state by Step~2 and apply $\mathcal{M}$ to the copy
    and reference. Record the outcomes $Z_1,\ldots,Z_s$.
\item \textbf{Find match for $\mathcal{M}$ outcome distribution.}
For any $j,z$, let $\widetilde R_{j,k_0}(z)\in[0,1]$ represent an
additive $\eta_2$ approximation of $R_{j,k_0}(z)$ computed using
$\mathsf{Amp}$ and \eqref{owsg:eq-8}, i.e.,
\[
    |\widetilde R_{j,k_0}(z)-R_{j,k_0}(z)|\le\eta_2.
\]
Use classical $\NP$ queries to perform binary search and find the lexicographically first
maximizer $h\in\mathcal{K}$ of
\begin{equation}
    \prod_{i=1}^s\widetilde R_{j,k_0}(Z_i)
    \qquad\text{over }j\in\mathcal{K}.
    \label{owsg:eq-15}
\end{equation}
Return $h$.
Note that the formula for $R_{j,k_0}(z)$ involves only constantly
many sums, products, and moduli of bounded amplitudes, so additive errors are amplified by at most a constant factor and the above precision can be achieved using $\mathsf{Amp}$.

\end{enumerate}
\paragraph{Analysis.}
\label{owsg:analysis}
Fix the unknown key $k\in\mathcal{K}$. We first show that
the basis distributions of $\ket{\psi_{k_0}}$ and $\ket{\psi_k}$ match with high probability, i.e. that $\ket{\psi^+_{k_0}}$ and $\ket{\psi^+_k}$ are close. Then we show that with high probability for the $h$ the algorithm obtains, the distributions we would obtain by applying $\mathcal{M}$ to $\ket{\psi_k}$ and to $\ket{\psi_h}$ (in each case using the reference state generated from $k_0$ in step 2) are close. In other words, the $\NP$ searches for the maximum likelihood distribution do indeed produce keys ($k_0$ and $h$) that produce similar distributions as the hidden key $k$. These follow primarily from the statistical properties of maximum likelihood. The challenge is then to show that any state producing similar distributions must also be close to  $\ket{\psi_k}$ in fidelity.\\

\noindent We begin by proving a strong success lemma for maximum likelihood search.

\begin{lemma}
\label{owsg:mle}
Let $J$ be a nonempty finite ordered set with $|J|\le M$. For each
$j\in J$, let $P_j$ be a probability distribution on a set of size $N$.
Suppose numbers
$\widetilde P_j(z)\in[0,1]$ satisfy
\[
    |\widetilde P_j(z)-P_j(z)|\le\eta
    \qquad\text{for every }j,z,
\]
where $\eta>0$. Let $s\ge1$, and suppose the joint distribution of
$Z=(Z_1,\ldots,Z_s)$ is within total variation distance $\xi$ of
$P_k^{\otimes s}$ for some $k\in J$. Let $h$ be the lexicographically first maximizer of
\[
    \prod_{i=1}^s\widetilde P_j(Z_i)
    \qquad\text{over }j\in J.
\]
For $0<\gamma<2$,
\begin{equation}
    \Pr[H^2(P_k,P_h)>\gamma]
    \le4sN^s\eta+\sqrt2\,M e^{-s\gamma/2}+\xi.
    \label{owsg:eq-37}
\end{equation}
In particular, this probability is at most $\beta$ for $0<\beta<1$ if
\begin{align}
    s&\ge\frac2\gamma
        \left(\ln M+\ln\frac{3\sqrt2}{\beta}\right),
        \label{owsg:eq-25}\\
    \eta&\le\frac{\beta}{12sN^s},
        \label{owsg:eq-26}\\
    \xi&\le\frac\beta3.
        \label{owsg:eq-27}
\end{align}
\end{lemma}

\begin{proof}
For a sequence $z=(z_1,\ldots,z_s)$, write
\[
    \mathcal{L}_j(z):=\prod_{i=1}^sP_j(z_i),\qquad
    \widetilde{\mathcal{L}}_j(z):=\prod_{i=1}^s\widetilde P_j(z_i).
\]
Since all factors lie in $[0,1]$, telescoping gives
\begin{equation}
    |\widetilde{\mathcal{L}}_j(z)-\mathcal{L}_j(z)|\le s\eta.
    \label{owsg:eq-31}
\end{equation}
First take $Z\sim P_k^{\otimes s}$. Define
$E=\{z:\mathcal{L}_k(z)<4s\eta\}$. There are at most $N^s$ sequences, so
\[
    \Pr[Z \in E]\le4sN^s\eta.
\]
Outside $E$, maximality of the computed likelihood implies
\begin{equation}
    \mathcal{L}_h(Z)\ge\widetilde{\mathcal{L}}_h(Z)-s\eta
    \ge\widetilde{\mathcal{L}}_k(Z)-s\eta
    \ge \mathcal{L}_k(Z)-2s\eta\ge\frac12\mathcal{L}_k(Z).
    \label{owsg:eq-34}
\end{equation}

For each fixed $j$, the inequality $\mathcal{L}_j(z)\ge \mathcal{L}_k(z)/2$ implies
$\mathcal{L}_k(z)\le\sqrt2\sqrt{\mathcal{L}_k(z)\mathcal{L}_j(z)}$. Consequently, by union bound over $z$,
\begin{align}
    \Pr\!\left[\mathcal{L}_j(Z)\ge\frac12\mathcal{L}_k(Z)\right]
    &\le\sqrt2\sum_z\sqrt{\mathcal{L}_k(z)\mathcal{L}_j(z)}\notag\\
    &=\sqrt2\,B(P_k,P_j)^s.
    \label{owsg:eq-35}
\end{align}
If $H^2(P_k,P_j)>\gamma$, then
\[
    B(P_k,P_j)=1-\frac12H^2(P_k,P_j)
    <1-\frac\gamma2\le e^{-\gamma/2}.
\]
A union bound over at most $M$ candidates, together with
\eqref{owsg:eq-34}, bounds the failure probability by
$4sN^s\eta+\sqrt2\,M e^{-s\gamma/2}$.
Changing the joint observation distribution by total variation
$\xi$ adds at most $\xi$ to this probability. This proves
\eqref{owsg:eq-37}. Conditions~\eqref{owsg:eq-25}--\eqref{owsg:eq-27}
bound each of its three terms by $\beta/3$.
\end{proof}

First we use Lemma \ref{owsg:mle} to show that $p_k$ and $p_{k_0}$ are close with high probability. 
Apply the lemma to $P_j=p_j$, with
$J=\mathcal{K}$, $M=|\mathcal{K}|$, $N=D$, and
\[
    \gamma=\frac{\varepsilon}{400},\qquad
    \beta=\frac\delta2,\qquad \eta=\eta_1,\qquad \xi=0.
\]
Plugging in the parameter values gives gives
\begin{equation}
    \Pr\!\left[H^2(p_k,p_{k_0})>\frac{\varepsilon}{400}\right]
    \le\frac\delta2.
    \label{owsg:good-first}
\end{equation}
Next we show that $R_{k,k_0}$ and $R_{h,k_0}$ are close with high probability. Fix any $k_0$ obtained in the first step. 
If the reference states were prepared exactly, Step~3 would by definition produce independent samples from
$R_{k,k_0}$. Although the reference states are not exact, we may bound the error as follows.

\begin{lemma}[Accuracy of reference preparation]
\label{owsg:positive}
For every $k_0\in\mathcal{K}$, Step~2 prepares a reference
$\rho_{k_0,\zeta}^+$ satisfying
\begin{equation}
    T\bigl(\rho_{k_0,\zeta}^+,
      |\psi_{k_0}^+\rangle\langle\psi_{k_0}^+|\bigr)\le\zeta,
    \label{owsg:eq-39}
\end{equation}
where $T(\rho,\sigma):=\tfrac12\|\rho-\sigma\|_1$.
\end{lemma}

\begin{proof}
Let $\widetilde a_{k_0}(x)$ be the amplitude approximations computed in Step~2,
and let $V_{k_0}$ be the diagonal unitary
\[
    V_{k_0}|x\rangle=e^{-i\arg \widetilde a_{k_0}(x)}|x\rangle,
    \qquad \arg0:=0.
\]
For each $x$,
\begin{equation}
    \left|a_{k_0}(x)e^{-i\arg \widetilde a_{k_0}(x)}-|a_{k_0}(x)|\right|
    \le2|a_{k_0}(x)-\widetilde a_{k_0}(x)|.
    \label{owsg:eq-42}
\end{equation}
For $\widetilde a_{k_0}(x)\ne0$, this follows by adding and subtracting
$\widetilde a_{k_0}(x)e^{-i\arg \widetilde a_{k_0}(x)}=|\widetilde a_{k_0}(x)|$ and using the triangle
inequality. For $\widetilde a_{k_0}(x)=0$, the left side is at most $2|a_{k_0}(x)|$.
Summing the squared errors,
\begin{equation}
    \bigl\|V_{k_0}|\psi_{k_0}\rangle-|\psi_{k_0}^+\rangle\bigr\|_2
    \le2\sqrt D\,\frac{\zeta}{16\sqrt D}
    =\frac\zeta8.
    \label{owsg:eq-44}
\end{equation}
Adding the operator-norm implementation error $\zeta/2$ gives
Euclidean error at most $\zeta/8+\zeta/2<\zeta$, including the
auxiliary registers. Since trace distance is upper bounded by this Euclidean
error we are done.
\end{proof}

By Lemma~\ref{owsg:positive}, each actual reference differs from
$|\psi_{k_0}^+\rangle$ by trace distance at most $\zeta$.
Measurement cannot increase trace distance. Since each measurement
uses a fresh target copy and reference preparation, replacing the
$s$ references one at a time gives
\begin{equation}
    \operatorname{TV}\!\left(
        \{Z_1,\ldots,Z_s\},R_{k,k_0}^{\otimes s}
    \right)\le s\zeta=\frac\delta6,
    \label{owsg:eq-88}
\end{equation}
where $\{Z_1,\ldots,Z_s\}$ denotes the joint distribution in step 3.

Apply Lemma~\ref{owsg:mle} to $P_j=R_{j,k_0}$ with
$J=\mathcal{K}$, $M=|\mathcal{K}|$, $N=4D^2$,
$\gamma=\varepsilon/400$, $\beta=\delta/2$,
$\eta=\eta_2$, and $\xi=\delta/6$.
It follows that
\begin{equation}
    \Pr\!\left[
        H^2(R_{k,k_0},R_{h,k_0})>\frac{\varepsilon}{400}
    \right]\le\frac\delta2.
    \label{owsg:good-second}
\end{equation}
By \eqref{owsg:good-first} and \eqref{owsg:good-second}, except with
probability at most $\delta$, both
\[
    H(p_k,p_{k_0})\le\sqrt{\frac{\varepsilon}{400}},\qquad
    H(R_{k,k_0},R_{h,k_0})\le\sqrt{\frac{\varepsilon}{400}}.
\]
We must now show that these two distributional guarantees imply
high fidelity between $|\psi_k\rangle$ and $|\psi_h\rangle$.
The following theorem supplies exactly this implication, with
$\psi=\psi_k$, $\phi=\psi_h$, and $q=p_{k_0}$.
\begin{theorem}
\label{owsg:fidelity}
Let $|\psi\rangle,|\phi\rangle$ be normalized pure states on
$\mathcal{X}$, with basis distributions $p,r$, and let $q$ be any
probability distribution on $\mathcal{X}$. Write
\[
    g:=H(p,q),\qquad
    t:=H(R_{\psi,q},R_{\phi,q}).
\]
Then
\begin{equation}
    1-|\langle\phi|\psi\rangle|^2\le100(t+g)^2.
    \label{owsg:eq-49}
\end{equation}
\end{theorem}

We prove the theorem in three steps. First, we show that if $t$ and $g$ are small, this
forces the distribution $r$ to be close to $q$. Since $p$ and $q$ are also close, we may replace the reference state $\ket{q^+}$ with $\ket{\psi^+}$ when sampling from $R_{\psi,q}$ and with $\ket{\phi^+}$ when sampling from $R_{\phi,q}$ without increasing the distance between the distributions by much.
Second, this lets us replace the common reference by each state's
own positive-amplitude state. Finally, we show that the distance between the distributions thus obtained can directly be connected to the fidelity of the states.

First we note that the third marginals in $R_{\psi,q}$ and $R_{\phi,q}$ are distributed as $(p+q)/2$ and
$(r+q)/2$. Since if the two distributions are close, their marginals must be close, and the marginals are simply mixtures of $p$ and $r$ with $q$, $p$ and $r$ must be close.

\begin{lemma}
\label{owsg:magnitudes}
$H(r,q)\le4(t+g)$
\end{lemma}

\begin{proof}
Since the third marginal  of distribution $R_{\psi,q}$ is obtained by measuring the second register of 
\[
    \frac{
        |0\rangle|\psi\rangle|q^+\rangle
        +|1\rangle|q^+\rangle|\psi\rangle
    }{\sqrt2}.
\]
we can trace out the first and last register, leaving a uniform mixture of $p$ and $q$, i.e the measurement outcome is distributed as $(p+q)/2$. Similarly for $R_{\phi,q}$ the marginal is distributed as $(r+q)/2$. Taking a marginal cannot
increase Hellinger distance: by Cauchy--Schwarz, merging outcomes
can only increase the Bhattacharyya coefficient. Thus
\begin{equation}
    H\!\left(\frac{p+q}{2},\frac{r+q}{2}\right)\le t.
    \label{owsg:eq-52}
\end{equation}
For nonnegative $u,v$,
\begin{equation}
    \left|\sqrt{\frac{u+v}{2}}-\sqrt v\right|
    \le|\sqrt u-\sqrt v|
    \le4\left|\sqrt{\frac{u+v}{2}}-\sqrt v\right|.
    \label{owsg:mixture-comparison}
\end{equation}
The first inequality follows because $\sqrt{(u+v)/2}$ lies between
$\sqrt u$ and $\sqrt v$. For the second, when $u\ne v$,
rationalizing gives
\[
    \frac{|\sqrt u-\sqrt v|}
         {|\sqrt{(u+v)/2}-\sqrt v|}
    =
    \frac{2\bigl(\sqrt{(u+v)/2}+\sqrt v\bigr)}
         {\sqrt u+\sqrt v}
    \le4.
\]
The case $u=v$ is immediate.
Applying \eqref{owsg:mixture-comparison} coordinatewise,
\[
    H\!\left(\frac{p+q}{2},q\right)\le H(p,q)=g.
\]
Together with \eqref{owsg:eq-52} and the triangle inequality, this gives
\[
    H\!\left(\frac{r+q}{2},q\right)\le t+g.
\]
A second application of \eqref{owsg:mixture-comparison} yields
\[
    H(r,q)\le4H\!\left(\frac{r+q}{2},q\right)\le4(t+g).
\]
\end{proof}
Having shown that $p$, $q$, $r$ are close, we replace the reference state in each distribution with the corresponding positive-amplitude state.
The following lemma bounds the effect of changing a reference.

\begin{lemma}
\label{owsg:reference-stability}
Fix a target state $|\chi\rangle$ and reference states $|b\rangle,|c\rangle$, as well a
measurement on the target and reference registers. Let
$P_{\chi,b},P_{\chi,c}$ be its outcome distributions for each reference with the same target state. Then
\begin{equation}
    H(P_{\chi,b},P_{\chi,c})
    \le\bigl\|\,|b\rangle-|c\rangle\,\bigr\|_2.
    \label{owsg:contractivity}
\end{equation}
In particular, for distributions $p,q$,
\[
    H(R_{\chi,p},R_{\chi,q})\le H(p,q).
\]
\end{lemma}

\begin{proof}
Write $|u\rangle=|\chi\rangle|b\rangle$ and
$|v\rangle=|\chi\rangle|c\rangle$. Choose measurement operators
$M_z$ with $\sum_zM_z^\dagger M_z=I$.
The square root of each outcome probability is the norm of the
corresponding unnormalized postmeasurement vector. Hence the triangle inequality gives
\begin{align*}
    H^2(P_{\chi,b},P_{\chi,c})
    &=\sum_z
      \bigl(\|M_z|u\rangle\|_2-\|M_z|v\rangle\|_2\bigr)^2\\
    &\le\sum_z\|M_z(|u\rangle-|v\rangle)\|_2^2\\
    &=\bigl\|\,|u\rangle-|v\rangle\,\bigr\|_2^2\\
    &=\bigl\|\,|b\rangle-|c\rangle\,\bigr\|_2^2.
\end{align*}
The specialization follows from Lemma~\ref{lem:qsample-identities}.
\end{proof}
Apply this lemma to replace $|q^+\rangle$ by $|\psi^+\rangle$
in the experiment corresponding to $R_{\psi,q}$ and by $|\phi^+\rangle$ in the experiment corresponding to $R_{\phi,q}$.
By the hypothesis on $p$ and Lemma~\ref{owsg:magnitudes},
\[
    H(R_{\psi, p},R_{\psi,q})\le H(p,q) \leq g,\qquad
    H(R_{\phi,q},R_{\phi, r})\le H(r,q)\le4(t+g).
\]
$R_{\psi,q}$ and $R_{\phi,q}$ are at most $t$ far, so
\begin{align}
    H(R_{\psi, p},R_{\phi, r})
    &\le H(R_{\psi, p},R_{\psi,q})
        +H(R_{\psi,q},R_{\phi,q})
        +H(R_{\phi,q},R_{\phi, r})\notag\\
    &\le g+t+4(t+g)=5(t+g).
    \label{owsg:eq-62}
\end{align}
All that remains is to show that closeness of these
distributions forces high fidelity.

\begin{lemma}
\label{owsg:matched-fidelity}
Let $|\psi\rangle,|\phi\rangle$ be normalized states on
$\mathcal{X}$, with basis distributions $p,r$. Then
\begin{equation}
    H^2(R_{\psi,p},R_{\phi,r})
    \ge\frac{1-|\langle\phi|\psi\rangle|^2}{4}.
    \label{owsg:eq-72}
\end{equation}
\end{lemma}

\begin{proof}
Let
\[
    \rho_\psi:=|\psi\rangle\langle\psi|,
    \qquad
    \rho_\phi:=|\phi\rangle\langle\phi|.
\]
We show that the squared Hellinger distance contribution of each pair $(x,y)$
can be used to control the difference between the corresponding entries of
$\rho_\psi$ and $\rho_\phi$.

Fix $(x,y)$ and write
\[
    c_\psi:=(\rho_\psi)_{xy}
        =\psi(x)\overline{\psi(y)}=ae^{i\theta},
    \qquad
    c_\phi:=(\rho_\phi)_{xy}
        =\phi(x)\overline{\phi(y)}=be^{i\varphi},
\]
where
\[
    a=\sqrt{p(x)p(y)},\qquad
    b=\sqrt{r(x)r(y)}.
\]
Writing $\Delta:=\theta-\varphi$, the squared difference between
the two density-matrix entries is
\begin{align}
    |c_\psi-c_\phi|^2
    &=a^2+b^2-2ab\cos\Delta \notag\\
    &=(a-b)^2+2ab(1-\cos\Delta).
    \label{owsg:entry-error}
\end{align}
We now compare this quantity with the contribution of $(x,y)$
to the squared Hellinger distance. Define
\begin{equation}
    Q_\vartheta(\alpha,\sigma)
    :=\frac{1+\sigma\cos(\vartheta+\alpha)}4,
    \qquad
    \alpha\in\{0,\pi/2\},\quad
    \sigma\in\{-1,1\}.
    \label{owsg:eq-63}
\end{equation}
The measurement outcome distributions have probabilities,
\[
    R_{\psi,p}(\alpha,\sigma,x,y)
        =a^2Q_\theta(\alpha,\sigma),
    \qquad
    R_{\phi,r}(\alpha,\sigma,x,y)
        =b^2Q_\varphi(\alpha,\sigma).
\]
We first relate the distance between $Q_\theta$ and $Q_\varphi$ to the differnce in the corresponding phases.

\medskip
\noindent\textbf{Claim.}
For all $\vartheta,\omega$,
\begin{equation}
    H^2(Q_\vartheta,Q_\omega)
    \ge
    \frac{1-\cos(\vartheta-\omega)}4.
    \label{owsg:eq-64}
\end{equation}

\noindent\emph{Proof of claim.}
For a fixed value of $\alpha$, let
\[
    P_{\vartheta,\alpha}(\sigma)
    :=Q_\vartheta(\sigma\mid\alpha)
    =\frac{1+\sigma\cos(\vartheta+\alpha)}2.
\]
 Computing the TVD, 
\[
    \operatorname{TV}
        (P_{\vartheta,\alpha},P_{\omega,\alpha})
    =
    \frac12
    |\cos(\vartheta+\alpha)-\cos(\omega+\alpha)|.
\]
For any distributions $P,Q$,
\[
    \operatorname{TV}(P,Q)\le H(P,Q),
\]
which follows from
\[
\begin{aligned}
    \operatorname{TV}(P,Q)
    &=
    \frac12\sum_z
    |\sqrt{P(z)}-\sqrt{Q(z)}|
    (\sqrt{P(z)}+\sqrt{Q(z)})\\
    &\le
    \frac12 H(P,Q)
    \left(
        \sum_z(\sqrt{P(z)}+\sqrt{Q(z)})^2
    \right)^{1/2}\\
    &\le H(P,Q).
\end{aligned}
\]
Therefore
\[
    H^2(P_{\vartheta,\alpha},P_{\omega,\alpha})
    \ge
    \frac14
    |\cos(\vartheta+\alpha)-\cos(\omega+\alpha)|^2.
\]

The marginal distribution of $\alpha$ is uniform in both
$Q_\vartheta$ and $Q_\omega$. Thus
\begin{align*}
    H^2(Q_\vartheta,Q_\omega)
    &=
    \frac12
    \sum_{\alpha\in\{0,\pi/2\}}
    H^2(P_{\vartheta,\alpha},P_{\omega,\alpha})\\
    &\ge
    \frac18\Bigl(
        (\cos\vartheta-\cos\omega)^2
        +(\sin\vartheta-\sin\omega)^2
    \Bigr)\\
    &=
    \frac{1-\cos(\vartheta-\omega)}4.
\end{align*}
This proves the claim.
\hfill$\square$

\medskip

Return now to the fixed pair $(x,y)$. Its contribution to
$H^2(R_{\psi,p},R_{\phi,r})$ is
\[
    d_{xy}
    :=
    \sum_{\alpha,\sigma}
    \left(
        a\sqrt{Q_\theta(\alpha,\sigma)}
        -
        b\sqrt{Q_\varphi(\alpha,\sigma)}
    \right)^2.
\]
Expanding the square and using
$\sum_{\alpha,\sigma}Q_\vartheta(\alpha,\sigma)=1$ gives
\begin{align*}
    d_{xy}
    &=
    a^2+b^2
    -2ab
    \sum_{\alpha,\sigma}
        \sqrt{
            Q_\theta(\alpha,\sigma)
            Q_\varphi(\alpha,\sigma)
        }.
\end{align*}
On the other hand,
\begin{align*}
    ab\,H^2(Q_\theta,Q_\varphi)
    &=
    ab\sum_{\alpha,\sigma}
        \left(
            \sqrt{Q_\theta(\alpha,\sigma)}
            -
            \sqrt{Q_\varphi(\alpha,\sigma)}
        \right)^2\\
    &=
    2ab
    -2ab
    \sum_{\alpha,\sigma}
        \sqrt{
            Q_\theta(\alpha,\sigma)
            Q_\varphi(\alpha,\sigma)
        }.
\end{align*}
Comparing the two expressions yields the exact decomposition
\begin{equation}
    d_{xy}
    =
    (a-b)^2
    +
    ab\,H^2(Q_\theta,Q_\varphi).
    \label{owsg:pair-decomposition}
\end{equation}
Thus the first term measures the difference in the magnitudes of
the two density-matrix entries, while the second measures their
phase difference.

Applying the claim,
\[
    d_{xy}
    \ge
    (a-b)^2
    +\frac{ab}{4}(1-\cos\Delta).
\]
Comparing this with \eqref{owsg:entry-error},
\[
    \frac18|c_\psi-c_\phi|^2
    =
    \frac18(a-b)^2
    +\frac{ab}{4}(1-\cos\Delta).
\]
Hence
\begin{equation}
    d_{xy}
    \ge
    \frac18|c_\psi-c_\phi|^2.
    \label{owsg:pair-entry-bound}
\end{equation}

Finally, summing over all pairs gives
\begin{align*}
    H^2(R_{\psi,p},R_{\phi,r})
    &=\sum_{x,y}d_{xy}\\
    &\ge
    \frac18\sum_{x,y}
        |(\rho_\psi)_{xy}-(\rho_\phi)_{xy}|^2\\
    &=
    \frac18\|\rho_\psi-\rho_\phi\|_F^2.
\end{align*}
Since $\rho_\psi$ and $\rho_\phi$ are pure states, and since the Frobenius norm for a Hermitian matrix is the trace of its square,
\begin{align*}
    \|\rho_\psi-\rho_\phi\|_F^2
    &= \operatorname{Tr}((\rho_\psi-\rho_\phi)^2)\\
    &=
    \operatorname{Tr}(\rho_\psi^2)
    +\operatorname{Tr}(\rho_\phi^2)
    -2\operatorname{Tr}(\rho_\psi\rho_\phi)\\
    &=
    2\bigl(1-|\langle\phi|\psi\rangle|^2\bigr).
\end{align*}
Substituting this above proves \eqref{owsg:eq-72}.
\end{proof}
We can put these lemmas together to prove Theorem~\ref{owsg:fidelity} as described.
\begin{proof}[Proof of Theorem~\ref{owsg:fidelity}]
By \eqref{owsg:eq-62},
\[
    H(R_{\psi,p},R_{\phi,r})\le 5(t+g).
\]
Lemma~\ref{owsg:matched-fidelity} therefore gives
\[
    1-|\langle\phi|\psi\rangle|^2
    \le 4H^2(R_{\psi,p},R_{\phi,r})
    \le 100(t+g)^2.
\]
\end{proof}
Finally, we prove Theorem~\ref{thm:owsg:learning}
\begin{proof}[Proof of Theorem~\ref{thm:owsg:learning}]
Run $\mathsf{Learn}^{\NP}$ as described above.
By \eqref{owsg:good-first} and \eqref{owsg:good-second}, except with
probability at most $\delta$,
\[
    H(p_k,p_{k_0})\le\sqrt{\frac{\varepsilon}{400}},
    \qquad
    H(R_{k,k_0},R_{h,k_0})
    \le\sqrt{\frac{\varepsilon}{400}}.
\]
Applying Theorem~\ref{owsg:fidelity} with
$\psi=\psi_k$, $\phi=\psi_h$, and $q=p_{k_0}$ gives
\[
\begin{aligned}
    1-|\langle\psi_h|\psi_k\rangle|^2
    &\le
    100\bigl(
        H(R_{k,k_0},R_{h,k_0})
        +H(p_k,p_{k_0})
    \bigr)^2 \\
    &\le
    100\left(2\sqrt{\frac{\varepsilon}{400}}\right)^2
    =\varepsilon.
\end{aligned}
\]
\end{proof}

\subsection{Inverting computable pure-state OWSGs}
\label{owsg:inversion}

Let $\Pi=(\mathsf{KeyGen},\mathsf{StateGen},\mathsf{Ver})$ be a
computable pure-state OWSG in the sense of Definition~\ref{def:structured-owsg}.
Its output family is computable, so we may apply
Theorem~\ref{thm:owsg:learning} to copies of the OWSG challenge state and
return the learned key.

\begin{corollary}[\(\NP\)-oracle attack on computable pure-state OWSGs]
\label{owsg:attack}
For $0<\varepsilon,\delta<1/4$, there is a uniform quantum algorithm
$\mathcal{A}^{\NP}$ using only classical oracle queries and
\[
    s_{\mathrm{atk}}=O\!\left(\frac{\ell(n)+\log(1/\delta)}{\varepsilon}\right)
\]
challenge copies, with running time
$\operatorname{poly}(n,1/\varepsilon,\log(1/\delta))$.
In the experiment
\[
    k\leftarrow\mathsf{KeyGen}(1^n),\qquad
    h\leftarrow\mathcal{A}^{\NP}
       (1^n,\varepsilon,\delta,|\psi_k\rangle^{\otimes s_{\mathrm{atk}}}),
\]
the following guarantees hold.
\begin{enumerate}
    \item If verification on key $h$ has acceptance operator
    $|\psi_h\rangle\langle\psi_h|$, then
    \[
        \Pr[\mathsf{Ver}(1^n,h,|\psi_k\rangle)=1]
        \ge(1-\delta)(1-\varepsilon).
    \]
    \item More generally, suppose verification satisfies the pointwise
    correctness bound
    \[
        \Pr[\mathsf{Ver}(1^n,h,|\psi_h\rangle)=1]
        \ge1-\nu(n)
        \qquad\text{for every }h\in\mathcal{K}.
    \]
    Then
    \[
        \Pr[\mathsf{Ver}(1^n,h,|\psi_k\rangle)=1]
        \ge1-\delta-\nu(n)-\sqrt\varepsilon.
    \]
\end{enumerate}
When $\nu$ is negligible, these generators are not one-way against
quantum polynomial-time adversaries with an $\NP$ oracle.
\end{corollary}

\begin{proof}
Run $\mathsf{Learn}^{\NP}$ and return its key $h$.
For every fixed $k$, except with probability at most $\delta$,
\[
    F_{h,k}:=|\langle\psi_h|\psi_k\rangle|^2\ge1-\varepsilon.
\]
For canonical verification, the conditional acceptance probability is
$F_{h,k}$. Averaging over the learner and then over key generation gives
$(1-\delta)(1-\varepsilon)$.

For a general verifier, let $E_h$ be its acceptance effect, with
$0\le E_h\le I$. On successful learning,
\[
    T\bigl(|\psi_h\rangle\langle\psi_h|,
           |\psi_k\rangle\langle\psi_k|\bigr)
    =\sqrt{1-F_{h,k}}\le\sqrt\varepsilon.
\]
Consequently,
\[
    \langle\psi_k|E_h|\psi_k\rangle
    \ge\langle\psi_h|E_h|\psi_h\rangle-\sqrt\varepsilon
    \ge1-\nu(n)-\sqrt\varepsilon.
\]
Allowing probability $\delta$ for failed learning proves the bound.
For example, $\varepsilon=\delta=1/16$ gives success at least
$11/16-\nu(n)$ with $O(\ell(n)+1)$ copies. For any inverse-polynomial
$0<\tau<1/2$, choosing $\varepsilon\le\tau^2/4$ and
$\delta\le\tau/2$ gives success at least $1-\nu(n)-\tau$ in polynomial
time.
\end{proof}

\begin{restatable}[One-way functions from computable samplable states]
    {theorem}{computablesamplableowf}
\label{thm:computable-samplable-owf}
Every computable samplable pure-state OWSG implies a classical one-way
function secure against quantum polynomial-time inversion.
\end{restatable}

\begin{proof}
The proof is given in Appendix~\ref{sec:owf-samplable}.
\end{proof}
\begin{remark}[Verification assumptions and scope]
\label{owsg:scope}
Part~2 of Corollary~\ref{owsg:attack} uses pointwise correctness, which is
stronger than the average correctness in Definition~\ref{def:owsg}.
This is needed because the worst-case $\NP$ learner's output key need not
follow the honest $\mathsf{KeyGen}$ distribution.
Appendix~\ref{sec:owf-samplable} avoids this stronger correctness assumption by using
posterior sampling.
\end{remark}

\section{Properties of Clifford and Bell measurements}
\label{sec:clifford-bell}
\label{sec:attack-preliminaries}

In this section, we prove several lemmas on the action of the
Clifford group on Bell states and on properties of the Haar
distribution. These results will support our \(\NP\)-aided attack
in Section~\ref{sec:cmc}, which distinguishes CMC constructions
from Haar-random unitaries. We begin by fixing our notation
for Pauli operators and Bell states.

Throughout, let $\cmcQubits\geq1$ and $\cmcDimension=2^{\cmcQubits}$.
We think of $\cmcBitStrings{\cmcQubits}$ as $\cmcField^{\cmcQubits}$;
addition and inner products of strings are over $\cmcField$.
For unit vectors or unitary operators, $\propto$ denotes equality up
to global phase.

\subsection{Pauli and Bell conventions}
\label{subsec:attack-paulis}
\label{subsec:attack-bell}

\begin{definition}[Pauli operators and labels]
\label{def:cmc-pauli}
For $x\in\cmcField$, let
$\cmcPauliX\cmcKet{x}=\cmcKet{x+1}$ and
$\cmcPauliZ\cmcKet{x}=(-1)^x\cmcKet{x}$.
For a label $v=(a,c)\in\cmcField^{2\cmcQubits}$, define
\[
    \cmcPauli{v}=\cmcPauli{a,c}
      =\bigotimes_{j=1}^{\cmcQubits}
          \cmcPauliX^{a_j}\cmcPauliZ^{c_j},
    \qquad
    \cmcParity{v}=a\cdot c.
\]
Thus $\cmcPauli{a,c}\cmcKet{x}
=(-1)^{c\cdot x}\cmcKet{x+a}$.
The Pauli group is
$\cmcPauliGroup=\{i^\ell\cmcPauli{v}:
\ell\in\{0,1,2,3\},\ v\in\cmcField^{2\cmcQubits}\}$.
The label of a Pauli up to global phase is the corresponding $v$.
\end{definition}

\begin{definition}[Bell states and swap]
\label{def:cmc-bell}
Define the maximally entangled state $\cmcPhi$ and Bell vectors $\cmcKet{\cmcBellLabel{a,c}}$ by
\[
    \cmcPhi=\frac1{\sqrt{\cmcDimension}}\sum_x\cmcKet{x,x},
    \quad
    \cmcKet{\cmcBellLabel{a,c}}
      =(I\otimes\cmcPauli{a,c})\cmcPhi
      =\frac1{\sqrt{\cmcDimension}}\sum_x
          (-1)^{c\cdot x}\cmcKet{x,x+a}.
\]
The swap operator exchanges the two entire registers:
$\cmcSwap\cmcKet{x,y}=\cmcKet{y,x}$.
Its $+1$ and $-1$ eigenspaces are the symmetric and antisymmetric
subspaces, respectively. Their projectors are $(I+\cmcSwap)/2$
and $(I-\cmcSwap)/2$, respectively.
\end{definition}

The labels above let us express Pauli multiplication and Bell-state
symmetry in terms of binary strings. In particular, the parity
$\cmcParity{v}$ determines whether a Bell state is symmetric or
antisymmetric, as the following identities show.

\begin{lemma}[Pauli identities and Bell symmetry]
\label{lem:cmc-pauli-bell}
For $v=(a,c)$ and $w=(b,d)$,
\[
    \cmcPauli{v}\cmcPauli{w}
      =(-1)^{c\cdot b}\cmcPauli{v+w},
    \qquad
    \cmcPauli{v}^{\cmcTranspose}
      =(-1)^{\cmcParity{v}}\cmcPauli{v}.
\]
Pauli labels are unique up to global phase: if
$\cmcPauli{v}\propto\cmcPauli{w}$, then $v=w$.
The Bell vectors form an orthonormal basis, and
\begin{equation}
\label{eq:cmc-swap-sign}
    \cmcSwap\cmcKet{\cmcBellLabel{v}}
       =(-1)^{\cmcParity{v}}\cmcKet{\cmcBellLabel{v}}.
\end{equation}
In particular, the symmetric Bell vectors have $\cmcParity{v}=0$,
and the antisymmetric Bell vectors have $\cmcParity{v}=1$. For any $\cmcDimension\times\cmcDimension$ matrices
$A,B$, we also have
\begin{equation}
\label{eq:cmc-transpose-trick}
    (A\otimes B)\cmcPhi=(I\otimes BA^{\cmcTranspose})\cmcPhi.
\end{equation}
\end{lemma}

\begin{proof}
The Pauli multiplication and transpose identities follow from
$\cmcPauliZ\cmcPauliX=-\cmcPauliX\cmcPauliZ$.

The Bell inner products are
\[
    \cmcBra{\cmcBellLabel{a,c}}\cmcKet{\cmcBellLabel{b,d}}
       =\frac{\delta_{a,b}}{\cmcDimension}
          \sum_x(-1)^{(c+d)\cdot x}
       =\delta_{a,b}\delta_{c,d}.
\]
There are $\cmcDimension^2$ such vectors in a space of dimension $\cmcDimension^2$, so they form an
orthonormal basis. Pauli-label uniqueness follows:
$\cmcPauli{v}\propto\cmcPauli{w}$ implies
$\cmcKet{\cmcBellLabel{v}}\propto
 \cmcKet{\cmcBellLabel{w}}$,
so orthogonality forces $v=w$.

For swap, substituting $y=x+a$ gives
\begin{align*}
    \cmcSwap\cmcKet{\cmcBellLabel{a,c}}
       &=\frac1{\sqrt{\cmcDimension}}
          \sum_x(-1)^{c\cdot x}\cmcKet{x+a,x}\\
       &=\frac1{\sqrt{\cmcDimension}}
          \sum_y(-1)^{c\cdot(y+a)}\cmcKet{y,y+a}\\
       &=(-1)^{a\cdot c}\cmcKet{\cmcBellLabel{a,c}}.
\end{align*}

Finally, the coefficient of $\cmcKet{j,k}$ on either side of
\eqref{eq:cmc-transpose-trick} is
$\cmcDimension^{-1/2}\sum_x A_{j,x}B_{k,x}$.
\end{proof}

We will group Bell outcomes by their displacement coordinate. The
corresponding subspaces also have a computational-basis description,
which lets us analyze displacement probabilities in either basis.

\begin{definition}[Displacement subspaces]
\label{def:attack-displacement}
The displacement of a computational-basis pair $\cmcKet{x,y}$ is
$x+y$. For a Bell label $(a,c)$, define $\cmcDisp{a,c}=a$.
For each $a\in\cmcField^{\cmcQubits}$, let
\[
    \cmcSector{a}
      =\operatorname{span}\{\cmcKet{x,x+a}:
                            x\in\cmcField^{\cmcQubits}\}.
\]
The Bell vectors $\{\cmcKet{\cmcBellLabel{a,c}}\}_c$ form an
orthonormal basis of $\cmcSector{a}$. By
Lemma~\ref{lem:cmc-pauli-bell}, its antisymmetric part is spanned
by those Bell vectors with $a\cdot c=1$ and, for $a\neq0$, has
dimension $\cmcDimension/2$.
For a set $R$, a Bell outcome has displacement in a set $R$ if its first label
coordinate belongs to $R$; its phase coordinate is unrestricted
by this condition.
\end{definition}

\subsection{Clifford action and classical correction}
\label{subsec:attack-cliffords}

\begin{definition}[Clifford unitaries and their descriptions]
\label{def:cmc-clifford}
A Clifford unitary is a unitary implemented, up to global phase, by a
circuit over the gates
\[
    \cmcHadamard=\frac1{\sqrt2}
       \begin{pmatrix}1&1\\1&-1\end{pmatrix},
    \qquad \cmcPhaseGate=\cmcDiag(1,i),
    \qquad
    \cmcCNOT\cmcKet{x,y}=\cmcKet{x,x+y}.
\]
The one-qubit gates may act on any qubit, and CNOT may act on any
ordered pair of distinct qubits. Whenever an algorithm receives a
Clifford, its input is such a circuit description. Running times are
measured in $\cmcQubits$ and the number of gates in that description.
\end{definition}

We first describe how a Clifford transforms Pauli labels under
conjugation. Ignoring global phases turns this action into an
efficiently computable linear map on binary strings.

\begin{lemma}[Clifford conjugation of Paulis]
\label{lem:cmc-pauli-conjugation}
For every Clifford $\cmcClifford$, there is a unique invertible
linear map $\cmcLinearMap{\cmcClifford}$ on
$\cmcField^{2\cmcQubits}$ such that
\[
    \cmcClifford\cmcPauli{v}\cmcClifford^{\cmcAdjoint}
       \propto\cmcPauli{\cmcLinearMap{\cmcClifford}v}
    \quad\text{for every }v.
\]
The matrix of this map is computable in polynomial time from the
circuit description. For Cliffords $A,B$,
$\cmcLinearMap{AB}=\cmcLinearMap{A}\cmcLinearMap{B}$.
\end{lemma}

\begin{proof}
Each generating gate conjugates Pauli generators to Paulis. Up to
phase, Hadamard exchanges $\cmcPauliX$ and $\cmcPauliZ$, while
$\cmcPhaseGate$ sends $\cmcPauliX$ to $\cmcPauliX\cmcPauliZ$
and keeps $\cmcPauliZ$ the same. For a CNOT from $i$ to $j$, the updates are
\[
    \cmcPauliX_i\mapsto\cmcPauliX_i\cmcPauliX_j,
    \quad \cmcPauliX_j\mapsto\cmcPauliX_j,
    \quad \cmcPauliZ_i\mapsto\cmcPauliZ_i,
    \quad \cmcPauliZ_j\mapsto\cmcPauliZ_i\cmcPauliZ_j.
\]
These identities follow directly from the gate definitions; other
qubits are unaffected. They establish the Pauli conjugation property
for a circuit. Uniqueness of labels and the multiplication identity
in Lemma~\ref{lem:cmc-pauli-bell} show that the induced map is linear.
Conjugation is injective, so this linear map is invertible.
Tracking the $2\cmcQubits$ single-qubit Pauli generators through the
circuit computes its columns in polynomial time. Applying conjugation
first by $B$ and then by $A$ proves the composition formula.
\end{proof}

To understand the action of
$\cmcClifford\otimes\cmcClifford$ on Bell states, the transpose
identity~\eqref{eq:cmc-transpose-trick} introduces the factor
$\cmcClifford\cmcClifford^{\cmcTranspose}$ in addition to Pauli
conjugation. The next lemma shows that this factor is itself a Pauli
up to phase and determines the offset in the Bell-label map.

\begin{lemma}[A Clifford times its transpose]
\label{lem:cmc-transpose-pauli}
For every Clifford $\cmcClifford$, there is a unique label
$\cmcOffset{\cmcClifford}$ such that
\begin{equation}
\label{eq:cmc-transpose-pauli}
    \cmcClifford\cmcClifford^{\cmcTranspose}
       \propto\cmcPauli{\cmcOffset{\cmcClifford}}.
\end{equation}
This label is computable in polynomial time. For Cliffords $A,B$,
\[
    \cmcOffset{AB}
      =\cmcLinearMap{A}\cmcOffset{B}+\cmcOffset{A}.
\]
\end{lemma}

\begin{proof}
Put $\cmcTransposePauli{\cmcClifford}
=\cmcClifford\cmcClifford^{\cmcTranspose}$. The generating gates satisfy
\[
    \cmcTransposePauli{\cmcHadamard_j}=I,
    \qquad \cmcTransposePauli{\cmcPhaseGate_j}=\cmcPauliZ_j,
    \qquad \cmcTransposePauli{\cmcCNOT_{i\to j}}=I.
\]
For a product,
\[
    \cmcTransposePauli{AB} = (AB)(AB)^{\cmcTranspose}=A(BB^{\cmcTranspose})A^{\cmcTranspose}=A\cmcTransposePauli{B}A^{\cmcAdjoint}
          (AA^{\cmcTranspose})
       =A\cmcTransposePauli{B}A^{\cmcAdjoint}
          \cmcTransposePauli{A}.
\]
Lemma~\ref{lem:cmc-pauli-conjugation} and induction on the circuit
therefore show that this is always a Pauli up to phase. Uniqueness
follows from Lemma~\ref{lem:cmc-pauli-bell}. The above identity proves the recursion and computes the offset
one generator at a time.
\end{proof}

Applying the same Clifford to both registers permutes the Bell basis. The next lemma describes this map and shows how classical postprocessing recovers the measurement distribution before the Clifford was applied.

\begin{lemma}[The Clifford map on Bell labels]
\label{lem:cmc-clifford-bell}
For every Clifford $\cmcClifford$, there exists a map
$\cmcLabelMap{\cmcClifford}:\cmcField^{2\cmcQubits}
\to\cmcField^{2\cmcQubits}$ with the following properties.
\begin{enumerate}
    \item \emph{Affine form and offset.}
    $\cmcLabelMap{\cmcClifford}(v)
       =\cmcLinearMap{\cmcClifford}v+\cmcOffset{\cmcClifford}$,
    where $\cmcLinearMap{\cmcClifford}$ and
    $\cmcOffset{\cmcClifford}$ are as in
    Lemmas~\ref{lem:cmc-pauli-conjugation} and
    \ref{lem:cmc-transpose-pauli}. In particular,
    $\cmcLabelMap{\cmcClifford}(0)=\cmcOffset{\cmcClifford}$.
    \item \emph{Invertibility.} The map is a bijection, with
    $\cmcLabelMap{\cmcClifford}^{-1}(z)
       =\cmcLinearMap{\cmcClifford}^{-1}
          (z+\cmcOffset{\cmcClifford})$.
    \item \emph{Bell action.} For every label $v$,
    \begin{equation}
    \label{eq:cmc-affine-action}
        (\cmcClifford\otimes\cmcClifford)
           \cmcKet{\cmcBellLabel{v}}
           \propto\cmcKet{\cmcBellLabel{\cmcLabelMap{\cmcClifford}(v)}}.
    \end{equation}
    \item \emph{Preservation of parity.}
    $\cmcParity{\cmcLabelMap{\cmcClifford}(v)}=\cmcParity{v}$.
    Equivalently, if $(a',c')=\cmcLabelMap{\cmcClifford}(a,c)$,
    then $a'\cdot c'=a\cdot c$ in $\cmcField$.
    \item \emph{Composition.} For Cliffords $A,B$,
    $\cmcLabelMap{AB}=\cmcLabelMap{A}\circ\cmcLabelMap{B}$.
    \item \emph{Efficient computation.} The affine description,
    its inverse, and their evaluations are computable in polynomial
    time from the circuit description of $\cmcClifford$.
    \item \emph{Classical correction of outcomes.} For any density
    operator $\rho$ on two $\cmcQubits$-qubit registers, measuring
    $(\cmcClifford\otimes\cmcClifford)\rho
     (\cmcClifford^{\cmcAdjoint}\otimes\cmcClifford^{\cmcAdjoint})$
    in the Bell basis and replacing the outcome $z$ by
    $\cmcLabelMap{\cmcClifford}^{-1}(z)$ gives exactly the Bell
    measurement distribution of $\rho$.
\end{enumerate}
\end{lemma}

\begin{proof}
Define $\cmcLabelMap{\cmcClifford}(v)
=\cmcLinearMap{\cmcClifford}v+\cmcOffset{\cmcClifford}$.
The Lemma~\ref{lem:cmc-pauli-conjugation} and Lemma~\ref{lem:cmc-transpose-pauli} give its affine form, invertibility,
composition rule, and efficient computation. Since, $\cmcLinearMap{\cmcClifford}$ is invertible, we have $\cmcLabelMap{\cmcClifford}^{-1}(z)
       =\cmcLinearMap{\cmcClifford}^{-1}
          (z+\cmcOffset{\cmcClifford})$.  
Using \eqref{eq:cmc-transpose-trick}, we obtain
\begin{align*}
    (\cmcClifford\otimes\cmcClifford)
       \cmcKet{\cmcBellLabel{v}}
      &=(I\otimes\cmcClifford\cmcPauli{v}
                         \cmcClifford^{\cmcTranspose})\cmcPhi\\
      &=(I\otimes
         (\cmcClifford\cmcPauli{v}\cmcClifford^{\cmcAdjoint})
         (\cmcClifford\cmcClifford^{\cmcTranspose}))\cmcPhi\\
      &\propto(I\otimes
         \cmcPauli{\cmcLinearMap{\cmcClifford}v}
         \cmcPauli{\cmcOffset{\cmcClifford}})\cmcPhi
       \propto\cmcKet{\cmcBellLabel{\cmcLabelMap{\cmcClifford}(v)}}.
\end{align*}
This proves the Bell action.

To prove parity preservation, note that swap commutes with
$\cmcClifford\otimes\cmcClifford$. Write
$(\cmcClifford\otimes\cmcClifford)\cmcKet{\cmcBellLabel{v}}
=\eta_v\cmcKet{\cmcBellLabel{\cmcLabelMap{\cmcClifford}(v)}}$,
where $\cmcAbs{\eta_v}=1$. Applying swap and using
\eqref{eq:cmc-swap-sign} on the right side gives the eigenvalue
$(-1)^{\cmcParity{\cmcLabelMap{\cmcClifford}(v)}}$.
Commuting swap past $\cmcClifford\otimes\cmcClifford$ on the left side instead gives
$(-1)^{\cmcParity{v}}$. Hence we have that $(-1)^{\cmcParity{\cmcLabelMap{\cmcClifford}(v)}} = (-1)^{\cmcParity{v}}$ or $\cmcParity{\cmcLabelMap{\cmcClifford}(v)} = \cmcParity{v}$. In particular, $\cmcParity{\cmcOffset{\cmcClifford}}=0$.

Finally, the Bell action implies the exact projector identity
\[
    (\cmcClifford\otimes\cmcClifford)
      \cmcKet{\cmcBellLabel{v}}\cmcBra{\cmcBellLabel{v}}
    (\cmcClifford^{\cmcAdjoint}\otimes\cmcClifford^{\cmcAdjoint})
       =\cmcKet{\cmcBellLabel{\cmcLabelMap{\cmcClifford}(v)}}
        \cmcBra{\cmcBellLabel{\cmcLabelMap{\cmcClifford}(v)}}.
\]
Applying the transformation backwards through bell measurements outcomes, we get the
classical correction property.
\end{proof}

We now choose an antisymmetric Bell state as the probe for the attack.
This choice ensures that Bell outcomes have nonzero displacement even
after the same unitary is applied to both registers.

Let $e_1\in\cmcField^{\cmcQubits}$ be the first standard basis vector, and fix
\begin{equation}
\label{eq:cmc-probe}
    \cmcProbeLabel=(e_1,e_1),
    \qquad
    \cmcPsi{V}=(V\otimes V)\cmcProbe.
\end{equation}
By Lemma~\ref{lem:cmc-pauli-bell}, the probe $\cmcProbe$ is
antisymmetric. Since $V\otimes V$ commutes with swap, so
$\cmcPsi{V}$ is antisymmetric for every unitary $V$. 

\subsection{Haar acceptance bounds}
\label{subsec:attack-haar}

We next bound the probability that repeated Bell measurements yield
displacements in a fixed set. For a fixed unitary, fresh preparations
give independent outcomes, so the probability that all
$\cmcSamples$ outcomes lie in the set is the $\cmcSamples$-th power
of the single-trial probability. Averaging over a Haar-random unitary
therefore requires a moment bound. We begin with the corresponding
formula for projections of a Haar-random state.

\begin{lemma}[Moments of a Haar-state projection]
\label{lem:cmc-haar-state-moment}
Let $\cmcKet{\vartheta}$ be a Haar-random unit vector in
$\cmcComplex^N$, and let $\cmcProjector$ be an orthogonal projector
of rank $r$, where $1\leq r\leq N$. For every positive integer
$\cmcSamples$,
\[
    \cmcExpectation_{\vartheta}
       \left[\cmcVectorNorm{\cmcProjector\cmcKet{\vartheta}}^{2\cmcSamples}\right]
       =\frac{\binom{r+\cmcSamples-1}{\cmcSamples}}
               {\binom{N+\cmcSamples-1}{\cmcSamples}}
       =\frac{\cmcRising{r}{\cmcSamples}}
               {\cmcRising{N}{\cmcSamples}},
\]
where $\cmcRising{x}{\cmcSamples}
=\prod_{j=0}^{\cmcSamples-1}(x+j)$ denotes the rising factorial.
\end{lemma}

\begin{proof}
By unitary invariance of the Haar distribution, we may choose a basis
such that
\[
    \cmcProjector=\sum_{i=1}^{r}\cmcKet{i}\cmcBra{i}.
\]

Let
\[
    \mathcal{T}
      =\left\{
          t=(t_1,\ldots,t_N)\in\mathbb{Z}_{\geq0}^{N}
          : \sum_{i=1}^{N}t_i=\cmcSamples
        \right\}
\]
be the set of type vectors. For each $t\in\mathcal{T}$, let $S_t$
contain the strings in $[N]^{\cmcSamples}$ in which the symbol $i$
occurs exactly $t_i$ times. The corresponding normalized type state is
\[
    \cmcKet{T_t}
      =\frac{1}{\sqrt{|S_t|}}
         \sum_{x\in S_t}\cmcKet{x},
    \qquad
    |S_t|=\frac{\cmcSamples!}{\prod_{i=1}^{N}t_i!}.
\]
Different types have disjoint supports, so these states are orthonormal.

The Haar-state moment identity of~\cite[Lemma~3.9]{AGL24} states that
\[
    \cmcExpectation_{\vartheta}
      \left[
        (\cmcKet{\vartheta}\cmcBra{\vartheta})^{\otimes\cmcSamples}
      \right]
      =\frac{1}{|\mathcal{T}|}
         \sum_{t\in\mathcal{T}}\cmcKet{T_t}\cmcBra{T_t}.
\]
Since $\cmcProjector$ is an orthogonal projector,
\[
    \cmcVectorNorm{\cmcProjector\cmcKet{\vartheta}}^{2\cmcSamples}
      =\operatorname{Tr}\!\left[
          \cmcProjector^{\otimes\cmcSamples}
          (\cmcKet{\vartheta}\cmcBra{\vartheta})^{\otimes\cmcSamples}
        \right].
\]
Taking expectations therefore gives
\[
    \cmcExpectation_{\vartheta}
      \left[
        \cmcVectorNorm{\cmcProjector\cmcKet{\vartheta}}^{2\cmcSamples}
      \right]
      =\frac{1}{|\mathcal{T}|}
         \sum_{t\in\mathcal{T}}
           \cmcBra{T_t}\cmcProjector^{\otimes\cmcSamples}
           \cmcKet{T_t}.
\]

The projector $\cmcProjector^{\otimes\cmcSamples}$ keeps precisely
the basis strings whose symbols all lie in $[r]$. Consequently,
\[
    \cmcProjector^{\otimes\cmcSamples}\cmcKet{T_t}
      =
      \begin{cases}
        \cmcKet{T_t}, & t_i=0\text{ for every }i>r,\\
        0,           & \text{otherwise}.
      \end{cases}
\]
Thus $\cmcExpectation_{\vartheta}
      \left[
        \cmcVectorNorm{\cmcProjector\cmcKet{\vartheta}}^{2\cmcSamples}
      \right]$ is the fraction of type vectors supported on the
first $r$ coordinates. By a combinatorial argument, there are
$\binom{N+\cmcSamples-1}{\cmcSamples}$ types in total and
$\binom{r+\cmcSamples-1}{\cmcSamples}$ supported on those coordinates.
Hence
\[
    \cmcExpectation_{\vartheta}
      \left[
        \cmcVectorNorm{\cmcProjector\cmcKet{\vartheta}}^{2\cmcSamples}
      \right]
      =
      \frac{\binom{r+\cmcSamples-1}{\cmcSamples}}
           {\binom{N+\cmcSamples-1}{\cmcSamples}}
      =
      \frac{\prod_{j=0}^{\cmcSamples-1}(r+j)}
           {\prod_{j=0}^{\cmcSamples-1}(N+j)}
      =
      \frac{\cmcRising{r}{\cmcSamples}}
           {\cmcRising{N}{\cmcSamples}}.
\]
\end{proof}

For the antisymmetric probe, the displacement acceptance probability
can be written as an average of projection probabilities in dimension
$\cmcDimension-1$. Each term has the Haar-state moments computed
above, and convexity bounds the moments of their average.

\begin{lemma}[Haar acceptance of displacement sets]
\label{lem:cmc-haar-displacement}
Let $R$ be a fixed nonempty subset of
$\cmcField^{\cmcQubits}\setminus\{0\}$, and write $r=\cmcAbs{R}$.
For a unitary $V$, let $\cmcAcceptance{R}(V)$ be the probability
that a Bell measurement of $\cmcPsi{V}$ has displacement in $R$.
Then, for every positive integer $\cmcSamples$,
\begin{equation}
\label{eq:cmc-haar-displacement-moment}
    \cmcExpectation_{V\leftarrow\cmcHaar(\cmcUnitaryGroup{\cmcDimension})}
       \left[\cmcAcceptance{R}(V)^{\cmcSamples}\right]
      \leq\frac{\cmcRising{r}{\cmcSamples}}
                   {\cmcRising{\cmcDimension-1}{\cmcSamples}}.
\end{equation}
In particular, if $r\leq\cmcDimension/2$ and
$1\leq\cmcSamples\leq\cmcDimension-2$, then
\begin{equation}
\label{eq:cmc-haar-displacement-geometric}
    \cmcExpectation_V\left[\cmcAcceptance{R}(V)^{\cmcSamples}\right]
      \leq\left(\frac34\right)^{\cmcSamples}.
\end{equation}
\end{lemma}

\begin{proof}
For each computational-basis label $x$, define
\[
    \cmcKet{\vartheta_x(V)}
      =V\cmcPauli{\cmcProbeLabel}
         V^{\cmcTranspose}\cmcKet{x}.
\]
Expanding $\cmcPsi{V}$ gives
\begin{equation}
\label{eq:cmc-skew-coefficients}
\begin{aligned}
    \cmcPsi{V}
      &=\frac1{\sqrt{\cmcDimension}}
          \sum_z V\cmcKet{z}\otimes
                 V\cmcPauli{\cmcProbeLabel}\cmcKet{z}\\
      &=\frac1{\sqrt{\cmcDimension}}
          \sum_x \cmcKet{x}\otimes\cmcKet{\vartheta_x(V)}.
\end{aligned}
\end{equation}
Each $\cmcKet{\vartheta_x(V)}$ is normalized because
$V\cmcPauli{\cmcProbeLabel}V^{\cmcTranspose}$ is unitary.
Moreover,
$\cmcPauli{\cmcProbeLabel}^{\cmcTranspose}
=-\cmcPauli{\cmcProbeLabel}$, so this matrix is skew-symmetric.
Its diagonal entries are therefore $0$:
\[
    \cmcBra{x}\cmcKet{\vartheta_x(V)}=0.
\]
Thus, measuring the first register gives a uniformly random $x$,
and the second register is a unit vector orthogonal to $\cmcKet{x}$.

We claim that, for each fixed $x$, $\cmcKet{\vartheta_x(V)}$ is Haar-random
in the $(\cmcDimension-1)$-dimensional subspace orthogonal to
$\cmcKet{x}$. To see this, we show that $A\cmcKet{\vartheta_x(V)}$ has the same distribution as $\cmcKet{\vartheta_x(V)}$ for any $A$ that fixes $\cmcKet{x}$. That is
$A^{\cmcTranspose}\cmcKet{x}=\cmcKet{x}$, and hence
\[
    \cmcKet{\vartheta_x(AV)}
      =AV\cmcPauli{\cmcProbeLabel}
          V^{\cmcTranspose}A^{\cmcTranspose}\cmcKet{x}
      =A\cmcKet{\vartheta_x(V)}.
\]
Since $AV$ has the same distribution as $V$ ($V$ is sampled from the Haar distribution), the distribution of
$\cmcKet{\vartheta_x(V)}$ is invariant under every $A$. Hence $\cmcKet{\vartheta_x(V)}$ is Haar-random
in the $(\cmcDimension-1)$-dimensional subspace orthogonal to
$\cmcKet{x}$.

For each displacement $b$, the Bell vectors
$\{\cmcKet{\cmcBellLabel{b,d}}\}_d$ form an orthonormal basis of
$\operatorname{span}\{\cmcKet{x,x+b}\}_x$.
Consequently, the total probability of displacement $b$ is the
same whether we measure in the Bell basis or the computational basis.
Using \eqref{eq:cmc-skew-coefficients}, we therefore have
\[
    \cmcAcceptance{R}(V)
      =\frac1{\cmcDimension}\sum_x s_x(V),
    \qquad
    s_x(V)
      =\sum_{b\in R}
         \cmcAbs{\cmcBra{x+b}\cmcKet{\vartheta_x(V)}}^2.
\]
Because $0\notin R$, the vectors
$\{\cmcKet{x+b}:b\in R\}$ are $r$ distinct orthonormal vectors
orthogonal to $\cmcKet{x}$.
Thus $s_x(V)$ is the squared norm of a rank-$r$ projection of
a Haar-random unit vector in dimension $\cmcDimension-1$.
Lemma~\ref{lem:cmc-haar-state-moment} gives
\[
    \cmcExpectation_V[s_x(V)^{\cmcSamples}]
      =\frac{\cmcRising{r}{\cmcSamples}}
             {\cmcRising{\cmcDimension-1}{\cmcSamples}}.
\]
Finally, convexity of $t\mapsto t^{\cmcSamples}$ implies
\begin{align*}
    \cmcExpectation_V
      \left[\cmcAcceptance{R}(V)^{\cmcSamples}\right]
      &=\cmcExpectation_V\left[
          \left(\frac1{\cmcDimension}\sum_x s_x(V)\right)^{\cmcSamples}
        \right]\\
      &\leq\frac1{\cmcDimension}\sum_x
          \cmcExpectation_V[s_x(V)^{\cmcSamples}]\\
      &=\frac{\cmcRising{r}{\cmcSamples}}
             {\cmcRising{\cmcDimension-1}{\cmcSamples}}.
\end{align*}
This proves \eqref{eq:cmc-haar-displacement-moment}.

If $r\leq\cmcDimension/2$ and
$\cmcSamples\leq\cmcDimension-2$, then every
$0\leq j\leq\cmcSamples-1$ satisfies $j\leq\cmcDimension-3$,
so
\[
    \frac{r+j}{\cmcDimension-1+j}
      \leq\frac{\cmcDimension/2+j}{\cmcDimension-1+j}
      \leq\frac34.
\]
Multiplying these inequalities gives
\[
    \frac{\cmcRising{r}{\cmcSamples}}
         {\cmcRising{\cmcDimension-1}{\cmcSamples}}
      =\prod_{j=0}^{\cmcSamples-1}
         \frac{r+j}{\cmcDimension-1+j}
      \leq\left(\frac34\right)^{\cmcSamples}.
\]
\end{proof}

\section{\texorpdfstring{\(\NP\)}{NP}-aided attack on CMC type constructions}
\label{sec:cmc}
\label{sec:cmc-attack}

In this section, we give an attack on CMC type constructions: a monomial
layer between two Clifford unitaries. We use the Bell-measurement and
Clifford machinery developed in Section~\ref{sec:clifford-bell} to turn
quantum measurements into classical constraints on a candidate key.
A single classical \(\NP\) query checks whether any key satisfies all
these constraints. A genuine CMC unitary always passes this check,
whereas the Haar bounds from Section~\ref{subsec:attack-haar} show that
a Haar-random unitary passes with small probability.

Fix $\cmcQubits\geq2$ and write $\cmcDimension=2^{\cmcQubits}$.
We consider a family indexed by a nonempty key space
$\cmcKeySpace\subseteq\cmcBitStrings{\cmcKeyLength}$, with
\[
    \cmcUnitary_k
      =\cmcClifford_{2,k}\cmcMonomial_k\cmcClifford_{1,k}.
\]
Here $\cmcClifford_{1,k}$ and $\cmcClifford_{2,k}$ are Clifford
unitaries, and $\cmcMonomial_k$ is a monomial layer as in
Definition~\ref{def:monomial-layer}, with permutation
$\cmcPermutation_k$ and phases $\cmcPhase_k$.
We assume uniform deterministic algorithms that, in time polynomial
in $\cmcQubits+\cmcKeyLength$, check whether $k\in\cmcKeySpace$
on input $(\cmcUnary{\cmcQubits},\cmcUnary{\cmcKeyLength},k)$,
produce both Clifford circuit descriptions up to global phase on
input $(\cmcUnary{\cmcQubits},k)$, and evaluate
$\cmcPermutation_k(x)$ on input $(\cmcUnary{\cmcQubits},k,x)$.
The phases may be arbitrary: the attack never evaluates
$\cmcPhase_k$ or $\cmcPermutation_k^{-1}$. Notice that using the \(\NP\) oracle, we can also learn the candidate key $k$.

\begin{theorem}[\(\NP\)-aided distinction of CMC and Haar unitaries]
\label{thm:cmc-np-attack}
For every CMC family satisfying the assumptions above, there exists
a uniform quantum algorithm $\cmcAdversary$ with only classical
access to an \(\NP\) oracle satisfying the following guarantee.

For every positive integer $\cmcSamples$, the algorithm receives
$(\cmcUnary{\cmcQubits},\cmcUnary{\cmcKeyLength},\cmcUnary{\cmcSamples})$
and oracle access to an unknown unitary $\cmcUnitary$.
It runs in time polynomial in
$\cmcQubits+\cmcKeyLength+\cmcSamples$, uses $2\cmcSamples$
nonadaptive forward queries to $\cmcUnitary$ and one classical query
to an \(\NP\)-complete language, and outputs a bit with these guarantees:
\begin{enumerate}
    \item \textbf{CMC case.} If $\cmcUnitary=\cmcUnitary_k$ for any
    $k\in\cmcKeySpace$, then $\cmcAdversary$ outputs $1$ with
    probability one.
    \item \textbf{Haar case.} If $\cmcUnitary$ is Haar-random, then
    \begin{equation}
    \label{eq:cmc-final-soundness}
        \cmcProbability_{\cmcUnitary\leftarrow
          \cmcHaar(\cmcUnitaryGroup{\cmcDimension})}
          [\cmcAdversary^{\cmcUnitary}=1]
        \leq 2^{\cmcKeyLength}
          \frac{\cmcRising{\cmcDimension/2}{\cmcSamples}}
               {\cmcRising{\cmcDimension-1}{\cmcSamples}}.
    \end{equation}
    For $\cmcSamples\leq\cmcDimension-2$, this is at most
    $2^{\cmcKeyLength}(3/4)^{\cmcSamples}$.
\end{enumerate}

In particular, suppose $\cmcKeyLength=\cmcKeyLength(\cmcQubits)$
is polynomially bounded.
For every positive integer-valued polynomial
$\cmcErrorExponent=\cmcErrorExponent(\cmcQubits)$, taking
$\cmcSamples=3(\cmcKeyLength+\cmcErrorExponent)$ gives distinguishing
advantage at least $1-2^{-\cmcErrorExponent}$ for all sufficiently
large $\cmcQubits$, under any distribution on the CMC key.
\end{theorem}

We prove the theorem by describing the algorithm that learns the key and its \(\NP\)
query, then analyzing the CMC and Haar cases.

\subsection{Algorithm}
\label{subsec:cmc-algorithm}

The algorithm collects $\cmcSamples$ Bell measurement outcomes and asks whether
one candidate key can explain all of them.
\begin{enumerate}
    \item For each $i=1,\ldots,\cmcSamples$, prepare a fresh copy of $\cmcProbe$ from
    \eqref{eq:cmc-probe}.
    \item For each $i$, apply $\cmcUnitary$ to both registers and measure
    $\cmcPsi{\cmcUnitary}$ in the Bell basis.\footnote{The Bell measurement uses CNOT gates from the first register to the
    second, followed by Hadamard gates on the first register and
    computational-basis measurements of the whole state. The first register yields $d_i$
    and the second yields $b_i$.} Record the full label
    $z_i=(b_i,d_i)$.
    \item Make the classical \(\NP\) query described below. Output
    $1$ if the answer is yes, and $0$ otherwise.
\end{enumerate}
Each repetition uses two forward oracle queries, for a total of
$2\cmcSamples$. All queries can be made in parallel before any
measurement.

\subsection{\texorpdfstring{\(\NP\)}{NP} query}
\label{subsec:cmc-np-query}

A candidate key $h$ specifies both Clifford layers. The first layer
maps $\cmcProbe$ to a Bell state with label $(a_h,c_h)$; the
second layer can be undone by classically correcting each observed
Bell label. More precisely, compute
\begin{equation}
\label{eq:cmc-correction}
\begin{aligned}
    (a_h,c_h)
      &=\cmcLabelMap{\cmcClifford_{1,h}}(\cmcProbeLabel),\\
    (\cmcCorrected{b_i}{h},\cmcCorrected{d_i}{h})
      &=\cmcLabelMap{\cmcClifford_{2,h}}^{-1}(z_i).
\end{aligned}
\end{equation}
These maps are efficiently computable by
Lemma~\ref{lem:cmc-clifford-bell}. They preserve Bell parity, so
$a_h\cdot c_h=1$ and, in particular, $a_h\neq0$.
Both coordinates of $z_i$ are needed to compute the correction,
even though the test below uses only the corrected displacement.

The monomial layer sends a basis pair $(x,x+a_h)$ to
$(\cmcPermutation_h(x),\cmcPermutation_h(x+a_h))$.
The query asks whether every corrected displacement could arise
this way for one common key:
\begin{equation}
\label{eq:cmc-np-predicate}
\begin{gathered}
    \exists h\in\cmcKeySpace\;
    \exists x_1,\ldots,x_{\cmcSamples}\in\cmcField^{\cmcQubits}:\\
    \cmcPermutation_h(x_i+a_h)+\cmcPermutation_h(x_i)
      =\cmcCorrected{b_i}{h}
    \quad\text{for every }i=1,\ldots,\cmcSamples.
\end{gathered}
\end{equation}
The witness consists of $h,x_1,\ldots,x_{\cmcSamples}$ and has
$\cmcKeyLength+\cmcSamples\cmcQubits$ bits. Since a poly-time algorithm can check the validity of the witness.
Thus one query to an \(\NP\)-complete language decides whether
such a witness exists. 

\subsection{CMC case}
\label{subsec:cmc-case}

We first record why the true key always passes the test, and why
each candidate allows at most half of all displacements.

\begin{lemma}[Bell support of a monomial layer]
\label{lem:cmc-monomial-support}
Let $\cmcMonomial$ be a monomial layer with permutation
$\cmcPermutation$. For $a\in\cmcField^{\cmcQubits}$, define
\[
    \cmcRange{\cmcPermutation}(a)
      =\{\cmcPermutation(x+a)+\cmcPermutation(x):
          x\in\cmcField^{\cmcQubits}\}.
\]
Every Bell outcome $(b,d)$ of
$(\cmcMonomial\otimes\cmcMonomial)\cmcKet{\cmcBellLabel{a,c}}$
has $b\in\cmcRange{\cmcPermutation}(a)$.
If $a\neq0$, then
$0\notin\cmcRange{\cmcPermutation}(a)$ and
$\cmcAbs{\cmcRange{\cmcPermutation}(a)}\leq\cmcDimension/2$.
\end{lemma}

\begin{proof}
By the definitions of a Bell state and a monomial layer,
\[
    (\cmcMonomial\otimes\cmcMonomial)
      \cmcKet{\cmcBellLabel{a,c}}
      =\frac1{\sqrt{\cmcDimension}}\sum_x
        (-1)^{c\cdot x}\cmcPhase(x)\cmcPhase(x+a)
        \cmcKet{\cmcPermutation(x),\cmcPermutation(x+a)}.
\]
Every basis pair in this sum has displacement in
$\cmcRange{\cmcPermutation}(a)$. A Bell state with any other
displacement is orthogonal to every such pair, so it cannot occur
as a measurement outcome.

If $a\neq0$, the inputs $x$ and $x+a$ are distinct, and their
images under the permutation are distinct. Their displacement
is therefore nonzero. Also, $x$ and $x+a$ give the same displacement.
There are $\cmcDimension/2$ such unordered pairs, so there are at
most $\cmcDimension/2$ possible displacements.
\end{proof}

Now suppose the oracle is $\cmcUnitary_k$ and use the true key
$h=k$ in the query. The first Clifford layer maps the probe to
$\cmcKet{\cmcBellLabel{a_k,c_k}}$ up to global phase.
By Lemma~\ref{lem:cmc-clifford-bell}, correcting the outcomes for
the second Clifford layer gives exactly the Bell distribution of
\[
    (\cmcMonomial_k\otimes\cmcMonomial_k)
      \cmcKet{\cmcBellLabel{a_k,c_k}}.
\]
The support lemma shows that every corrected displacement has the
form $\cmcPermutation_k(x_i+a_k)+\cmcPermutation_k(x_i)$ for some
$x_i$. Thus the true key and these $x_i$ satisfy
\eqref{eq:cmc-np-predicate} for every possible measurement record.
The algorithm outputs $1$ with probability one.

\subsection{Haar case}
\label{subsec:cmc-haar-case}

Fix a candidate key $h$ before sampling the Haar oracle, and let
$R_h=\cmcRange{\cmcPermutation_h}(a_h)$ be its allowed displacements.
Since $a_h\neq0$, the support lemma gives
\[
    0\notin R_h,
    \qquad 1\leq\cmcAbs{R_h}\leq\cmcDimension/2.
\]
Correcting a Bell outcome using the second Clifford layer of $h$
has the same effect on its distribution as applying
$\cmcClifford_{2,h}^{\cmcAdjoint}$ after the oracle.
Thus the probability that one outcome passes the test for $h$ is
\[
    \cmcAcceptance{R_h}
      (\cmcClifford_{2,h}^{\cmcAdjoint}\cmcUnitary).
\]
For this fixed $h$, the unitary
$\cmcClifford_{2,h}^{\cmcAdjoint}\cmcUnitary$ is Haar-random whenever
$\cmcUnitary$ is Haar-random.

Conditioned on the oracle unitary, the $\cmcSamples$ repetitions
use fresh states and have independent measurement outcomes.
The probability that this one key explains every outcome is therefore
bounded by Lemma~\ref{lem:cmc-haar-displacement}:
\begin{align*}
    \cmcProbability[\text{$h$ explains all $\cmcSamples$ outcomes}]
      &=\cmcExpectation_{\cmcUnitary}\!\left[
          \cmcAcceptance{R_h}
            (\cmcClifford_{2,h}^{\cmcAdjoint}\cmcUnitary)^{\cmcSamples}
        \right]\\
      &\leq\frac{\cmcRising{\cmcAbs{R_h}}{\cmcSamples}}
                    {\cmcRising{\cmcDimension-1}{\cmcSamples}}\\
      &\leq\frac{\cmcRising{\cmcDimension/2}{\cmcSamples}}
                    {\cmcRising{\cmcDimension-1}{\cmcSamples}}.
\end{align*}
The algorithm accepts only if some key explains every outcome.
There are at most $2^{\cmcKeyLength}$ keys, so a union bound proves
\eqref{eq:cmc-final-soundness}.
When $\cmcSamples\leq\cmcDimension-2$,
Lemma~\ref{lem:cmc-haar-displacement} gives the simpler bound
$2^{\cmcKeyLength}(3/4)^{\cmcSamples}$.

Finally, for $\cmcSamples=3(\cmcKeyLength+\cmcErrorExponent)$,
\[
    2^{\cmcKeyLength}(3/4)^{\cmcSamples}
      \leq 2^{\cmcKeyLength}2^{-(\cmcKeyLength+\cmcErrorExponent)}
      =2^{-\cmcErrorExponent},
\]
where we used $(3/4)^3\leq1/2$.
Combining this Haar bound with
acceptance probability one for every CMC key proves the claimed
advantage.

\section{Breaking known constructions}
\label{sec:constructions}

In this section, we apply our results to known constructions of
pseudorandom states, one-way state generators, pseudorandom unitaries,
and pseudorandom isometries. We show that many of these constructions
have computable samplable outputs, and hence admit our \(\NP\)-aided
tomography attack. For computable samplable OWSGs,
we invoke
Theorem~\ref{thm:computable-samplable-owf} to obtain one-way functions.
We then use the CMC attack to handle further unitary constructions
whose output amplitudes are not known to be computable.

Throughout this section, computability and samplability are given the
key, as in Definitions~\ref{def:computable-states}
and~\ref{def:computable-samplable-states}. For a unitary or isometry,
these properties refer to its output on a specified computational-basis
input. Uniformly computable entries or phases mean deterministic
classical evaluation to error $2^{-b}$ in $\poly(n+b)$ time.

We first prove some useful lemmas in Section~\ref{subsec:con-tools}.
We apply them to state constructions in Section~\ref{con:states} and to
unitary and isometry constructions in
Section~\ref{subsec:con-explicit-columns}. We discuss CMC constructions
in Section~\ref{subsec:con-cmc}, and the constructions for which our
attacks do not apply directly in Section~\ref{subsec:con-limits}.

\subsection{Useful lemmas}
\label{subsec:con-tools}

Given copies of a state from a computable family, our tomography
algorithm recovers a key for a nearby state. To distinguish a
pseudorandom state from Haar, we can use
this key to prepare a candidate state and compare it with fresh copies
of the challenge. We first formalize this argument, and observe that
it also works after a public change of basis. We then prove a few
facts about computability and samplability that will be used across
the constructions below. In particular, we show how to compute
amplitudes and sample measurement outcomes after applying phases,
permutations, and certain simple mixing operations.

\subsubsection{From tomography to attacks}

We begin with the distinguishing attack. The main observation is that
a Haar state has small overlap with every state in the keyed family,
except with small probability. Thus, even though the learner chooses
its output key after seeing the challenge, fresh swap tests will
reject a Haar state. For a unitary or isometry, we obtain the required
copies simply by querying the oracle repeatedly on one fixed basis input.

For pure OWSGs, we can return the learner's key directly and apply
Corollary~\ref{owsg:attack}, under its canonical or pointwise-correct
verifier hypotheses. These attacks need only computability. When the
family is also samplable, Theorem~\ref{thm:computable-samplable-owf}
shows that its one-wayness, with classical key generation, implies a
classical one-way function secure against quantum polynomial-time
inversion.

\begin{corollary}[Computable columns admit an \(\NP\)-aided attack]
\label{cor:con-column-attack}
For an $n$-qubit computable pure-state family with key length $L$ and
$D=2^n$, a quantum algorithm with classical $\NP$ access uses
$O(L+s)$ copies to distinguish the family from Haar. It rejects a real
state with probability $O(2^{-s})$ and accepts Haar with probability
at most $2^L(3/4)^{D-1}+O(2^{-s})$. Its running time is polynomial
in $n+s$ when $L=\poly(n)$. For a computable column of a unitary or isometry,
the copies can be obtained by parallel forward queries on one fixed
basis input.
\end{corollary}
\begin{proof}
We run the learner with target fidelity $15/16$ and failure probability
$2^{-s-2}$. Using the recovered key, we prepare the candidate state
and perform $128(s+1)$ swap tests against fresh challenge copies.
We accept if at least $3/4$ of these tests accept. On successful
learning in the real case, each test accepts with probability at
least $31/32$.

Now consider a Haar state $\ket\phi$. For each fixed key $h$, we have
\[
 \Pr[|\braket{\psi_h}{\phi}|^2\geq1/4]=(3/4)^{D-1}
\]
A union bound over all $2^L$ keys shows that, except with the stated
probability, no key gives a state with squared overlap at least $1/4$.
This holds in particular for the key chosen by the learner. Each swap
test then accepts with probability at most $5/8$. Applying Hoeffding's
inequality in the two cases gives the claimed error bounds.
Finally, the same argument applies to unitary and isometry columns,
since a Haar unitary or isometry maps a fixed basis input to a Haar
state in the output dimension.
\end{proof}

It is enough for the amplitudes to be computable after a known change
of basis. We can apply this change of basis to the challenge copies
before running the learner. This observation will be useful for the
IQP group-action states.

\begin{corollary}[Public changes of basis]
\label{cor:con-public-basis}
If $W$ is a known, key-independent efficient unitary and
$\{W\ket{\psi_k}\}_k$ is computable, the learning and distinguishing
attacks apply to $\{\ket{\psi_k}\}_k$ by first applying $W$ to the
copies.
\end{corollary}
\begin{proof}
Apply $W$ to each challenge copy and run the learner on the transformed
family. To prepare a candidate in the original basis, undo $W$ on the
state specified by the recovered key. Fidelity and Haar measure are
invariant under $W$, so both guarantees are preserved.
\end{proof}

The above corollary does not give computability or samplability in
the original basis. It also requires $W$ to be public and independent
of the key. In particular, it does not let us remove a secret output
Clifford.

\subsubsection{Computability and samplability}
\label{subsec:con-stabilizer-monomial}

Many of the constructions we consider begin with a stabilizer state
and then apply phases and permutations. Stabilizer states are
computable samplable. To see why these properties are preserved,
observe that a permutation only changes where an amplitude appears,
while a phase leaves the measurement probabilities unchanged.

\begin{lemma}[Stabilizers and monomial layers]
\label{lem:con-stabilizer}
\label{lem:con-monomial-closure}
Efficiently described stabilizer states are computable samplable.
If $M_k\ket z=\omega_k(z)\ket{\pi_k(z)}$ has efficient pure
implementation, classically efficient $\pi_k^{\pm1}$, and uniformly
computable phases, it preserves both properties. Sampling alone needs
only the forward permutation.
\end{lemma}
\begin{proof}
For stabilizer states, both tasks follow from standard classical
algorithms~\cite{DDM03,Van11}, using a fixed convention for the global
phase. Now let $\ket{\phi_k}$ be any computable samplable state and
apply a monomial unitary. The amplitude at $y$ is
\begin{equation}
 \bra yM_k\ket{\phi_k}
 =\omega_k(\pi_k^{-1}(y))\braket{\pi_k^{-1}(y)}{\phi_k}.
 \label{eq:con-monomial-amplitude}
\end{equation}
We can compute this amplitude by evaluating the inverse permutation,
the original amplitude, and the phase. To sample, we draw $z$ from the
original measurement distribution and output $\pi_k(z)$. This gives
the desired distribution because phases do not affect measurement
probabilities. The sampling procedure uses only the forward
permutation. We can repeat the argument for polynomially many
monomial layers.
\end{proof}

Another useful case is when, up to an efficiently computable phase,
an amplitude is a product of coefficients determined by successive
prefixes and the next bit. Computing an amplitude then
requires following a single path of prefixes. Sampling follows the
same path, choosing one bit at a time according to the corresponding
conditional probabilities.

\begin{lemma}[Prefix products]
\label{con:prefix-states}
Suppose a clean, efficient pure state has amplitudes
\begin{equation}
 a_k(x)=\omega_k(x)\prod_{j=1}^n c_{k,x_{<j}}(x_j),
 \qquad |\omega_k(x)|=1,\qquad |c_{k,v}(0)|^2+|c_{k,v}(1)|^2=1,
 \label{con:prefix-amplitude}
\end{equation}
where the factors are uniformly computable from the key and prefix
to error $2^{-b}$ in $\poly(n+b)$ time. Then the family is
computable samplable.
\end{lemma}
\begin{proof}
To compute $a_k(x)$, we evaluate its $n+1$ factors and multiply them.
Since every factor has absolute value at most one, it suffices to use
$b+O(\log(n+1))$ bits of precision per factor to obtain total error
at most $2^{-b}$.

To sample $x$, we choose its bits from left to right. Given the prefix
$v$, choose the next bit $z$ with probability $|c_{k,v}(z)|^2$.
The probability of obtaining $x$ is the product of these conditional
probabilities, which is exactly $|a_k(x)|^2$. Approximating each choice
to error at most $1/(nt)$ gives total variation error at most
$n/(nt)=1/t$. Both procedures run in polynomial time.
\end{proof}

We will also use the fact that these properties are preserved under
polynomial tensor products. In particular, repeating a state
polynomially many times does not prevent the attack.

\begin{lemma}[Tensor products]
\label{lem:con-tensor}
Polynomial tensor products of efficiently described computable samplable
pure-state families are computable samplable, provided their total
number of qubits is polynomially bounded.
\end{lemma}
\begin{proof}
An amplitude of the tensor product is the product of the corresponding
amplitudes of its $r$ factors. Computing each to $b+O(\log r)$ bits
of precision gives error at most $2^{-b}$. To sample the measurement
distribution, we sample the factors independently. Error at most
$1/(rt)$ for each factor gives total variation error at most $1/t$.
\end{proof}

Finally, we consider a permuted stabilizer state followed by a layer
that mixes pairs of basis states. Each output amplitude is a sum of
only two terms, so it remains efficiently computable. To sample, we
first choose a pair according to its total probability and then choose
one of its two outputs. The following lemma makes this argument
precise and will let us handle one step of the Kac construction.

\begin{lemma}[One layer of two-dimensional mixing]
\label{lem:con-two-dimensional}
The columns of $H_fPC$ are computable samplable when $C$ is an
efficiently described Clifford, $P$ and $P^{-1}$ are classically
polynomial-time computable, and
$H_f=\bigoplus_{r\in\{0,1\}^{n-1}}h_r$ is an efficiently implemented
layer of two-dimensional unitaries with uniformly computable entries.
\end{lemma}
\begin{proof}
Write $P\ket z=\ket{\pi(z)}$. For a basis input $x$, fix one
stabilizer normal form for $C\ket x$
and put $v_j(r)=\bra{\pi^{-1}(j,r)}C\ket x$. Then
\begin{equation}
 \bra{b,r}H_fPC\ket x=h_r[b,0]v_0(r)+h_r[b,1]v_1(r),
 \label{eq:con-kac-amplitude}
\end{equation}
Both terms are efficiently computable, which proves computability.
For sampling, first classically sample a computational-basis outcome
of $PC\ket x$ and retain the suffix $r$. This chooses $r$ with probability
$s_r=|v_0(r)|^2+|v_1(r)|^2$. We then choose $b$ with probability
\[
 q_r(b)=|(h_rv(r))_b|^2/s_r.
\]
Since $h_r$ is unitary, it preserves the total probability of the pair.
Thus, this is the required conditional distribution. It remains to
check that we can compute it efficiently even when $s_r$ is small.
In the fixed stabilizer normal form,
$v_j(r)=2^{-d/2}\eta_j(r)$ with
$\eta_j(r)\in\{0,\pm1,\pm i\}$. The common scale cancels from the
ratio, leaving denominator
$|\eta_0(r)|^2+|\eta_1(r)|^2\in\{1,2\}$ on every sampled suffix.
Hence we can approximate these probabilities to error $1/t$ in
polynomial time, uniformly in $r$.
\end{proof}

\subsection{State constructions}
\label{con:states}
\label{subsec:con-table}

We now show that the state constructions in
Table~\ref{tab:constructions} are computable samplable. In each case,
we describe how to evaluate an amplitude and how to sample the
computational-basis distribution, given the key. We consider the
specified finite-key instantiations of these constructions. For the
IQP group-action states, we first apply a public Hadamard transformation.

\begin{table}[!htbp]
\centering
\small
\setlength{\tabcolsep}{3pt}
\setlength{\belowcaptionskip}{6pt}
\renewcommand{\arraystretch}{1.16}
\caption{Computable samplable state constructions. For the IQP states,
we first apply the public Hadamard transformation.}
\label{tab:constructions}
\begin{tabular}{@{}>{\raggedright\arraybackslash}p{.43\textwidth}>{\raggedright\arraybackslash}p{.14\textwidth}>{\raggedright\arraybackslash}p{.38\textwidth}@{}}
\toprule
Construction & Primitive & Reason for coverage\\
\midrule
PRF phase and binary-phase states~\cite{JLS18,BSBinary,AGQY22}
 & PRS & Evaluable phases; uniform sampling\\
Finite-key Hamiltonian phase states~\cite{HPS}
 & PRS / OWSG candidates & Evaluable phases; uniform sampling\\
PRP subset and subset-phase states~\cite{Subset1,Subset2,Pseudoent23}
 & PRS & Invertible map from a known interval\\
Prefix-rotation scalable states~\cite{BCJZ26}
 & PRS & Prefix-product amplitudes\\
Disjoint PFC blocks on a Bell seed~\cite{MarkovGlue26}
 & PRS & Monomial image of a stabilizer state\\
Reversible-cipher states, Eq.~(13) of~\cite{Inflationary24}
 & PRS candidate & Evaluable binary phases; uniform sampling\\
IQP group-action states~\cite{QGA}
 & OWSG candidate & Phase states after public $H^{\otimes n}$\\
Pure OWSGs from these states~\cite{MorimaeYamakawa}
 & Pure OWSG & The same families, including tensor repetition\\
\bottomrule
\end{tabular}
\end{table}

\paragraph{Phase states.}
\label{con:phase-states}
For PRF phase states~\cite{JLS18}, binary-phase states~\cite{BSBinary,AGQY22}, and finite-key
Hamiltonian phase states~\cite{HPS}, the amplitudes are, respectively,
\[
 2^{-n/2}e^{2\pi i\mathsf{PRF}_k(x)/2^n},\qquad
 2^{-n/2}(-1)^{\mathsf{PRF}_k(x)},\qquad
 2^{-n/2}\exp\!\left(i\sum_{j=1}^m\theta_j(-1)^{A_j\cdot x}\right).
\]
Here $m=\poly(n)$, and the HPS angles are given by their prescribed
finite encoding. Given the key and $x$, we can evaluate each phase
efficiently to arbitrary precision. Moreover, all amplitudes have the
same magnitude, so sampling amounts to choosing a uniform basis string.
Thus these families are computable samplable; they are also immediate
instances of Lemma~\ref{con:prefix-states}.
The same argument applies to the structured reversible-cipher state
of~\cite[Eq.~(13)]{Inflationary24}, whose amplitude is
$2^{-n/2}(-1)^{\beta\cdot P^{-1}(x)}$.

\paragraph{Subset and subset-phase states.}
\label{con:subset-states}
Next, consider the subset constructions
of~\cite{Subset1,Subset2,Pseudoent23}. These apply a permutation $p_k$
to a uniform superposition over the known interval
$[s]=\{0,\ldots,s-1\}$, possibly with an additional binary phase.
The resulting amplitude at $y$ is
\begin{equation}
 a_k(y)=s^{-1/2}\omega_k(y)\mathbf1[p_k^{-1}(y)<s],
 \label{con:subset-amplitude}
\end{equation}
where $\omega_k=1$ or an efficiently evaluated binary phase, and
both $p_k$ and $p_k^{-1}$ are efficient. Given $y$, we compute
$p_k^{-1}(y)$ and check whether it lies in $[s]$. If it does, the
amplitude is $\omega_k(y)/\sqrt{s}$; otherwise it is zero.
To sample, we choose a uniform $z\in[s]$ and output $p_k(z)$.
We can obtain $z$ by rejection sampling from $\lceil\log_2s\rceil$
random bits. Truncating after $O(\log(t+1))$ attempts gives error
at most $1/t$.

For completeness, the uniform interval superposition can be prepared
by prefix rotations. If $c_v$ counts the elements of $[s]$ with prefix
$v$, use amplitudes $\sqrt{c_{v0}/c_v}$ and $\sqrt{c_{v1}/c_v}$ for
the next bit. These counts are efficiently computable, and a prefix
with $c_v=0$ is never reached.

\paragraph{Prefix rotations.}
The scalable construction of~\cite[Algorithms~1--2]{BCJZ26}
prepares the state one qubit at a time, using a rotation determined
by the preceding bits. In the notation of
Lemma~\ref{con:prefix-states}, its factors are
$c_{k,v}(0)=\cos\alpha_{k,v}$ and
$c_{k,v}(1)=\sin\alpha_{k,v}$.
The angle is produced by a classical Beta sampler using keyed PRF
bits. Thus, given the key and prefix, we can deterministically
compute the finitely encoded angle. The circuit also uncomputes its
workspace. Lemma~\ref{con:prefix-states} therefore applies: evaluating
one amplitude only requires the angles along one path of prefixes,
without computing a global normalization factor.

\paragraph{Disjoint Bell-seed gluing.}
Consider the instantiation of~\cite[Theorem~9]{MarkovGlue26} in which
disjoint PFC blocks act on a Bell seed. Since the blocks have disjoint
supports, we can write
\[
 (\bigotimes_i P_iF_iC_i)C_{\rm Bell}
 =(\bigotimes_i P_iF_i)(\bigotimes_i C_i)C_{\rm Bell}=MC.
\]
The product of the permutation and phase layers is monomial, while
the remaining circuit, including preparation of the Bell seed, is
Clifford. Hence Lemma~\ref{lem:con-stabilizer} gives computability
and samplability.

\paragraph{IQP group actions.}
\label{con:related}
The IQP candidates of~\cite[Definitions~5.2--5.3]{QGA} use the
fixed seed $\ket{0^n}$ and an action
$g_k=H^{\otimes n}D_kH^{\otimes n}$. Applying the public unitary
$H^{\otimes n}$ to the output gives
\begin{equation}
 H^{\otimes n}g_k\ket{0^n}
 =2^{-n/2}\sum_x e^{i\theta_k(x)}\ket x,
 \label{con:iqp-phase}
\end{equation}
where the circuit phase $\theta_k(x)$ is efficiently evaluated.
We are therefore back in the phase-state setting, and the transformed
family is computable samplable. By
Corollary~\ref{cor:con-public-basis}, we can learn in this basis
and return the recovered action. With learning error
$(\varepsilon,\delta)$, this wins the one-way game
of~\cite[Definition~3.5]{QGA} with probability at least
$(1-\delta)(1-\varepsilon)$. The weak-OWSG swap verifier accepts
with probability at least $(1-\delta)(1-\varepsilon/2)$.
We only need an action producing a nearby state; it need not be the
original sampled circuit.

\paragraph{Pure OWSGs and tensor repetition.}
The above arguments also apply when these families are used to build
pure OWSGs, with the public change of basis in the IQP case.
By Lemma~\ref{lem:con-tensor}, polynomial tensor repetition preserves
computability and samplability. We can therefore apply
Corollary~\ref{owsg:attack} to invert the resulting generator under
its stated verifier hypotheses. Furthermore,
Theorem~\ref{thm:computable-samplable-owf} implies one-way functions.

\FloatBarrier
\subsection{Unitary and isometry constructions}
\label{subsec:con-explicit-columns}

We next consider unitary and isometry constructions. Recall that it
suffices to find one fixed basis input whose output is computable:
we query the oracle repeatedly on that input and apply
Corollary~\ref{cor:con-column-attack}. For each construction in
Table~\ref{tab:con-columns}, we show that the output is both
computable and samplable.

\begin{table}[!htbp]
\centering
\small
\setlength{\tabcolsep}{3pt}
\setlength{\belowcaptionskip}{6pt}
\renewcommand{\arraystretch}{1.16}
\caption{Unitary and isometry constructions with computable samplable
basis-input columns.}
\label{tab:con-columns}
\begin{tabular}{@{}>{\raggedright\arraybackslash}p{.43\textwidth}>{\raggedright\arraybackslash}p{.17\textwidth}>{\raggedright\arraybackslash}p{.35\textwidth}@{}}
\toprule
Construction & Primitive & Reason for coverage\\
\midrule
Clifford PFC~\cite{PFC24,MH25}; phase-free PC~\cite{PC26}
 & PRU & Monomial image of a stabilizer column\\
One-sided LRFC; first two blocked forms~\cite{BlockedLRFC25,LRFC25}
 & PRU & Invertible shuffles and phases after a Clifford\\
$PF(\bigotimes_j V_j)$~\cite{Distinct26}
 & Nonadaptive PRU & Product columns with computable one-qubit entries\\
$PGH^{\otimes n}F$~\cite{Distinct26,BM24}
 & Nonadaptive / somewhat PRU & Explicit phases; uniform column distributions\\
One Kac step $H_fPC$~\cite{Kac26}
 & PRU & Two-dimensional mixing after a permuted stabilizer column\\
AGKL isometries~\cite{PRI24}
 & PRI & Invertible image of a fixed-prefix superposition\\
\bottomrule
\end{tabular}
\end{table}

\paragraph{One-sided constructions.}
\label{cor:con-one-sided}
Suppose first that $U_k=M_kC_k$, where $C_k$ is Clifford and
$M_k$ is monomial. On any basis input, $C_k$ produces a stabilizer
state, so Lemma~\ref{lem:con-stabilizer} applies whenever the
permutation and its inverse are efficiently computable, and the
phases are efficiently evaluable. This is the case for Clifford PFC~\cite{PFC24,MH25},
PC~\cite{PC26}, one-sided LRFC, and its first two blocked
forms~\cite{BlockedLRFC25,LRFC25}. Indeed, the PRPs have efficient
inverses, each XOR shuffle $(u,v)\mapsto(u\oplus h_k(v),v)$ is its
own inverse, and all keyed phases are efficiently evaluable.

\paragraph{Local designs.}
The same idea works for $PF(\bigotimes_jV_j)$~\cite{Distinct26}
even when the one-qubit unitaries are not Clifford, provided their
entries are uniformly computable. Before applying $PF$, the amplitude
on $y$ is $\prod_j\bra{y_j}V_j\ket{x_j}$, and we can sample each
output bit independently. Thus the product state is computable
samplable by Lemma~\ref{con:prefix-states}, and the monomial layer
$PF$ preserves these properties by
Lemma~\ref{lem:con-monomial-closure}.

\paragraph{Phase--Hadamard--phase columns.}
\label{prop:con-pghf}
Now consider $U_k=P_kG_kH^{\otimes n}F_k$, where $F_k,G_k$ are
diagonal. Write $f_k,g_k$ for their entries and $\pi_k$ for the
permutation. On a fixed basis input, the first phase is a scalar,
and the Hadamards produce a uniform superposition. In particular,
\begin{equation}
 \bra yU_k\ket x=2^{-n/2}g_k(\pi_k^{-1}(y))
          (-1)^{x\cdot\pi_k^{-1}(y)}f_k(x).
 \label{eq:con-pghf-amplitude}
\end{equation}
We can evaluate this amplitude using the inverse permutation and the
two phase functions. Its magnitude is constant, so the measurement
distribution is uniform. This proves computability and samplability
for the two-phase construction~\cite{Distinct26} and the somewhat
PRU construction~\cite{BM24}. For the latter, the restricted input
game permits polynomially many copies of one fixed basis input,
which is all Corollary~\ref{cor:con-column-attack} requires.

\paragraph{One Kac step after a Clifford.}
\label{prop:con-kac-one-step}
For the Clifford specialization of~\cite[Theorem~4.6]{Kac26}, the
unitary is $H_fPC$. The keyed permutation has an efficient inverse,
and the entries of each two-dimensional block are computable from
the finite PRF angle parameters using square roots and complex
exponentials. These are precisely the conditions of
Lemma~\ref{lem:con-two-dimensional}. Thus every basis-input column is
computable samplable, and Corollary~\ref{cor:con-column-attack} gives
an \(\NP\)-aided distinguisher. We use here that there is only one mixing
layer; we return to longer Kac walks in Section~\ref{subsec:con-limits}.

\paragraph{AGKL pseudorandom isometries.}
The construction of Ananth, Gulati, Kaleoglu, and Lin~\cite[Section~5,
Figure~1]{PRI24} maps a basis input $x\in\{0,1\}^n$ to
\begin{equation}
 \ket{\psi_{k,x}}=2^{-m/2}\sum_{z\in\{0,1\}^m}
 \omega_p^{f_{k_1}(x\Vert z)}\ket{g_{k_2}(x\Vert z)},
 \qquad \omega_p=e^{2\pi i/p},
 \label{con:pri-basis}
\end{equation}
where $f$ is an efficiently evaluable keyed function and $g$ is an
efficiently invertible keyed permutation. Fix the input $x$. To
compute the amplitude at $y$, we first invert the permutation and
write $g_{k_2}^{-1}(y)=x'\Vert z$. If $x'\ne x$, the amplitude is
zero; otherwise it is $2^{-m/2}\omega_p^{f_{k_1}(x\Vert z)}$.
To sample, we choose uniform $z$ and output $g_{k_2}(x\Vert z)$.
Hence this family is computable samplable.

The computational-basis security game of~\cite[Section~3]{PRI24}
allows repeated forward queries on this fixed input. Such queries
supply the copies needed for tomography. A Haar-random isometry sends
the same input to a Haar-random state, so
Corollary~\ref{cor:con-column-attack} gives the distinguishing attack
with output dimension $D=2^{n+m}$.
\FloatBarrier
\subsection{CMC attacks}
\label{subsec:con-cmc}
\label{prop:con-cmc-instances}

The preceding argument applies to $PFC$ because the monomial layer
preserves computability of the stabilizer column. What happens if
we add a secret output Clifford and consider $C_2PFC_1$ instead?
The output Clifford can mix the amplitudes, so we no longer obtain
computability from the same argument. However, the construction has
exactly the form required by Theorem~\ref{thm:cmc-np-attack}.

More generally, we can apply that theorem whenever the endpoints
are efficiently described Cliffords and the middle layer is monomial
with a classically efficiently evaluable forward permutation.
We do not need to evaluate the phases or invert the permutation.
Table~\ref{tab:con-cmc} lists the resulting applications, including
those already handled by tomography. To check each row, observe that
products of phases and permutations are monomial, and the forward
permutation can be evaluated by composing the individual permutations.
For the one-sided constructions, take the output Clifford to be the identity.

\begin{table}[!htbp]
\centering
\small
\setlength{\tabcolsep}{3pt}
\setlength{\belowcaptionskip}{6pt}
\renewcommand{\arraystretch}{1.16}
\caption{Constructions to which the CMC attack applies. The first
group also admits the tomography attack.}
\label{tab:con-cmc}
\begin{tabular}{@{}>{\raggedright\arraybackslash}p{.63\textwidth}>{\raggedright\arraybackslash}p{.33\textwidth}@{}}
\toprule
Construction & Monomial core\\
\midrule
\multicolumn{2}{@{}l}{\emph{Also covered by tomography}}\\
Clifford PFC~\cite{PFC24,MH25}; phase-free PC~\cite{PC26}
 & $PF$; $P$\\
One-sided LRFC; first two blocked forms~\cite{BlockedLRFC25,LRFC25}
 & Product of shuffles and phases\\
$PF(\bigotimes_j C_j)$, local Cliffords~\cite{Distinct26}
 & $PF$\\
\midrule
\multicolumn{2}{@{}l}{\emph{Constructions with secret output Cliffords}}\\
Strong Clifford-endpoint PFC: $C_2PFC_1$~\cite{MH25}
 & $PF$\\
Base strong LRFC (CLRFC): $C_2S_RF S_LC_1$~\cite{LRFC25}
 & $S_RF S_L$\\
Constant-time logical CMC branch~\cite{ConstantTime25}
 & Its implemented monomial core\\
\midrule
\multicolumn{2}{@{}l}{\emph{Additional variant considered here}}\\
CPC: $C_2PC_1$, extending PC~\cite{PC26}
 & $P$\\
\bottomrule
\end{tabular}
\end{table}

For the constant-time construction, we apply the attack to the
logical CMC unitary implemented in~\cite[Theorem~3.7 and the strong
branch of Theorem~3.26]{ConstantTime25}. We also include the natural
two-sided variant $C_2PC_1$ of PC. This variant is not considered or
proved secure in~\cite{PC26}; we include it because its middle layer
is simply the permutation $P$, and hence the CMC attack applies.

\FloatBarrier
\subsection{Limitations of the current techniques}
\label{subsec:con-limits}

It is perhaps most important to understand what candidates and architectures survive our attacks, and can be plausible routes to constructions in Microcrypt.
Table~\ref{tab:con-uncovered} lists several constructions for which
our arguments do not give attacks. We briefly explain the difficulty
in applying our techniques to each of them.
We emphasize that this does not necessarily imply security against \(\NP\)-aided adversaries, since it is plausible that more sophisticated \(\NP\)-aided attacks could succeed. Nevertheless, we believe it is important to list these candidates explicitly for future research.

For example, our analysis of one Kac step uses the fact that each
output amplitude is a sum of two stabilizer amplitudes. After many
mixing steps, we no longer have this description, so the argument
does not immediately extend. A similar issue arises with gluing.
For disjoint PFC blocks on a Bell seed, we combined all monomial
layers and all Clifford layers. Overlapping blocks need not allow
this rearrangement. In particular, the third blocked LRFC form
leaves a Clifford between two monomial layers, so the argument for
the first two forms does not apply directly.

To see the same issue in the ABGL-style composition, substitute the
strong PFC base $U=C_2PFC_1$ into $X^{k_1}UX^{k_3}UX^{k_4}$.
We obtain
\[
 (X^{k_1}C_2)\,PF\,(C_1X^{k_3}C_2)\,PF\,(C_1X^{k_4}).
\]
Here the two copies of $PF$ are separated by a Clifford. Thus this
expression does not directly give the single monomial layer required
by our CMC theorem. Whether another decomposition or a different
attack can handle this construction remains open in our analysis.

In the following table, PRSS denotes a pseudorandom state scrambler.
\FloatBarrier

\begin{table}[!htbp]
\centering
\small
\setlength{\tabcolsep}{3pt}
\setlength{\belowcaptionskip}{6pt}
\renewcommand{\arraystretch}{1.16}
\caption{Constructions for which our techniques do not give an attack,
and the difficulties in applying the preceding arguments.}
\label{tab:con-uncovered}
\begin{tabular}{@{}>{\raggedright\arraybackslash}p{.38\textwidth}>{\raggedright\arraybackslash}p{.15\textwidth}>{\raggedright\arraybackslash}p{.42\textwidth}@{}}
\toprule
Construction & Primitive & Main difficulty\\
\midrule
Full PRSS walk~\cite{PRSS24}
 & PRSS & Extending the single-step analysis to many mixing rounds\\
Full long Kac walks~\cite{Kac26}
 & PRU; strong PRU & Efficient amplitude evaluation through many mixing steps\\
Third blocked LRFC~\cite{BlockedLRFC25}
 & PRU & Handling an intermediate Clifford between monomial layers\\
Ancilla-free glued LRFC~\cite{LRFC25}
 & Strong PRU & Combining monomial layers separated by Cliffords after gluing\\
Overlapping shallow/glued constructions~\cite{SHH24,StrongGlue25,GridStrong26}
 & PRU; strong PRU & Obtaining a global CMC form or computable columns from overlapping blocks\\
Phased-permutation Hamiltonians~\cite{PermHam24}
 & Nonadaptive PRU & Evaluating columns or finding a CMC form after exponentiating a sum of monomials\\
Hidden-basis Hamiltonian dynamics~\cite{HiddenHam25}
 & PRU & Handling a secret general PRU basis in place of Clifford endpoints\\
ABGL-style composition~\cite{QHROM25,MH25}
 & Strong PRU & Handling two monomial layers separated by a Clifford\\
\bottomrule
\end{tabular}
\end{table}

\section{AI Disclosure}
Generative AI tools were used during the preparation of this manuscript to assist with writing,
editing, and polishing the exposition, as well as to help formalize some ideas and check
mathematical arguments and proofs. The main research ideas, technical insights, and proof
strategies originated from the authors and were presented by one of us at a workshop at the Simons Institute in July. 

AI assistance helped refine the presentation and
formalization of these ideas. The authors take full responsibility for the content and integrity of the final
manuscript, including the correctness of the results and proofs, the originality of the work, and
the accuracy of all claims and citations.

\section{Acknowledgments}
Part of this work was done when A.G. was at NTT Research.
D.K. and K.T. were supported in part by AFOSR, NSF
CNS-2247727, and a Google Research Scholar award.

\clearpage
\begingroup
\raggedright
\bibliographystyle{alpha}
\bibliography{references}

@inproceedings{KT24,
  author    = {Dakshita Khurana and Kabir Tomer},
  title     = {Commitments from Quantum One-Wayness},
  booktitle = {Proceedings of the 56th Annual ACM Symposium on Theory of Computing},
  series    = {STOC '24},
  year      = {2024},
  pages     = {968--978},
  publisher = {Association for Computing Machinery},
  doi       = {10.1145/3618260.3649654},
  url       = {https://doi.org/10.1145/3618260.3649654}
}

@misc{CEGHLS26,
  author        = {Carrasco, Jose and Eisert, Jens and Ghosh, Soumik
                   and Hangleiter, Dominik and Li, Nicky Kai Hong
                   and Sweke, Ryan},
  title         = {Instantiating {Microcrypt}: Obstacles and opportunities
                   via tailored state certification},
  year          = {2026},
  eprint        = {2609.15842},
  archivePrefix = {arXiv},
  primaryClass  = {quant-ph},
  doi           = {10.48550/arXiv.2609.15842},
  url           = {https://arxiv.org/abs/2609.15842}
}

@misc{KhuranaSimons26,
  author       = {Khurana, Dakshita},
  title        = {What Makes a Cryptographic Assumption a Quantum Assumption?},
  year         = {2026},
  month        = jul,
  howpublished = {Talk at the Simons Institute for the Theory of Computing},
  note         = {Cryptography 10 Years Later: Obfuscation, Proof Systems,
                  and Secure Computation Reunion, July 14, 2026},
  url          = {https://simons.berkeley.edu/talks/dakshita-khurana-ntt-research-university-illinois-urbana-champaign-2026-07-14-0}
}

@misc{JLS,
  author = {Ji, Z. and Liu, Y.-K. and Song, F.},
  title = {{Pseudorandom states, non-cloning theorems and quantum money}},
  year = {2017},
  note = {arXiv:1711.00385}
}

@misc{Kretschmer,
  author = {Kretschmer, W.},
  title = {{Quantum pseudorandomness and classical complexity}},
  year = {2021},
  note = {arXiv:2103.09320}
}

@misc{KQST,
  author = {Kretschmer, W. and Qian, L. and Sinha, M. and Tal, A.},
  title = {{Quantum Cryptography in Algorithmica}},
  year = {2022},
  note = {arXiv:2212.00879}
}

@misc{KQT,
  author = {Kretschmer, W. and Qian, L. and Tal, A.},
  title = {{Quantum-Computable One-Way Functions without One-Way Functions}},
  year = {2024},
  note = {arXiv:2411.02554}
}

@misc{HPS,
  author = {Bostanci, J. and Haferkamp, J. and Hangleiter, D. and Poremba, A.},
  title = {{Efficient Quantum Pseudorandomness from Hamiltonian Phase States}},
  year = {2025},
  note = {arXiv:2410.08073}
}

@misc{BSBinary,
  author = {Brakerski, Z. and Shmueli, O.},
  title = {{(Pseudo) random quantum states with binary phase}},
  year = {2019},
  note = {arXiv:1906.10611}
}

@inproceedings{AGQY22,
  author = {Ananth, P. and Gulati, A. and Qian, L. and Yuen, H.},
  title = {{Pseudorandom (Function-Like) Quantum State Generators: New Definitions and Applications}},
  booktitle = {Theory of Cryptography (TCC 2022), Part I},
  series = {Lecture Notes in Computer Science},
  volume = {13747},
  pages = {237--265},
  publisher = {Springer},
  year = {2022},
  doi = {10.1007/978-3-031-22318-1_9},
  url = {https://arxiv.org/abs/2211.01444}
}

@misc{Subset1,
  author = {Jeronimo, F. G. and Magrafta, N. and Wu, P.},
  title = {{Pseudorandom and pseudoentangled states from subset states}},
  year = {2023},
  note = {arXiv:2312.15285}
}

@misc{Subset2,
  author = {Giurgica-Tiron, T. and Bouland, A.},
  title = {{Pseudorandomness from subset states}},
  year = {2023},
  note = {arXiv:2312.09206}
}

@misc{PFC,
  author = {Metger, T. and Poremba, A. and Sinha, M. and Yuen, H.},
  title = {{Simple constructions of linear-depth $t$-designs and pseudorandom unitaries}},
  year = {2024},
  note = {arXiv:2404.12647}
}

@misc{QGA,
  author = {Morimae, T. and Xagawa, K.},
  title = {{Quantum Group Actions}},
  year = {2024},
  note = {arXiv:2410.04777}
}

@misc{AGL24,
  author = {Ananth, P. and Gulati, A. and Lin, Y.-T.},
  title = {{Cryptography in the Common Haar State Model: Feasibility Results and Separations}},
  year = {2024},
  note = {arXiv:2407.07908v1},
  url = {https://arxiv.org/abs/2407.07908v1}
}

@article{Ben21,
  author = {Bennink, R. S.},
  title = {{Efficient verification of anticoncentrated quantum states}},
  journal = {npj Quantum Information},
  volume = {7},
  pages = {127},
  year = {2021},
  doi = {10.1038/s41534-021-00455-6},
  url = {https://doi.org/10.1038/s41534-021-00455-6}
}

@article{DDM03,
  author = {Dehaene, J. and {De Moor}, B.},
  title = {{The Clifford group, stabilizer states, and linear and quadratic operations over GF(2)}},
  journal = {Physical Review A},
  volume = {68},
  pages = {042318},
  year = {2003},
  url = {https://arxiv.org/abs/quant-ph/0304125}
}

@inproceedings{JLS18,
  author = {Ji, Z. and Liu, Y.-K. and Song, F.},
  title = {{Pseudorandom quantum states}},
  booktitle = {CRYPTO 2018},
  series = {LNCS},
  volume = {10993},
  pages = {126--152},
  publisher = {Springer},
  year = {2018},
  note = {Full version: {Pseudorandom States, Non-Cloning Theorems and Quantum Money}},
  url = {https://arxiv.org/abs/1711.00385}
}

@misc{Kac26,
  author = {Lu, C. and Qin, M. and Song, F. and Yao, P. and Zhao, M.},
  title = {{Parallel Kac's Walk Generates PRU}},
  year = {2026},
  note = {arXiv:2504.14957v4},
  url = {https://arxiv.org/abs/2504.14957v4}
}

@misc{MH25,
  author = {Ma, F. and Huang, H.-Y.},
  title = {{How to Construct Random Unitaries}},
  year = {2025},
  note = {arXiv:2410.10116v3},
  url = {https://arxiv.org/abs/2410.10116v3}
}

@misc{PFC24,
  author = {Metger, T. and Poremba, A. and Sinha, M. and Yuen, H.},
  title = {{Simple constructions of linear-depth $t$-designs and pseudorandom unitaries}},
  year = {2024},
  note = {arXiv:2404.12647v1},
  url = {https://arxiv.org/abs/2404.12647v1}
}

@misc{MorimaeYamakawa,
  author = {Morimae, T. and Yamakawa, T.},
  title = {{One-Wayness in Quantum Cryptography}},
  year = {2024},
  note = {arXiv:2210.03394v3},
  url = {https://arxiv.org/abs/2210.03394v3}
}

@misc{LRFC25,
  author = {Schuster, T. and Ma, F. and Lombardi, A. and Brand{\~a}o, F. and Huang, H.-Y.},
  title = {{Strong random unitaries and fast scrambling}},
  year = {2025},
  note = {arXiv:2509.26310v1},
  url = {https://arxiv.org/abs/2509.26310v1}
}

@article{Van11,
  author = {{Van den Nest}, M.},
  title = {{Simulating quantum computers with probabilistic methods}},
  journal = {Quantum Information and Computation},
  volume = {11},
  number = {9--10},
  pages = {784--812},
  year = {2011},
  url = {https://arxiv.org/abs/0911.1624}
}

@misc{BCJZ26,
  author = {Rishabh Batra and Zhili Chen and Rahul Jain and YaoNan Zhang},
  title = {{Scalable, quantum-accessible, and adaptive pseudorandom quantum state and pseudorandom function-like quantum state generators}},
  year = {2026},
  eprint = {2507.22535v6},
  archivePrefix = {arXiv},
  url = {https://arxiv.org/abs/2507.22535v6},
  note = {arXiv:2507.22535v6; Version 6, August 2026; first posted in 2025}
}

@misc{Pseudoent23,
  author = {Scott Aaronson and Adam Bouland and Bill Fefferman and Soumik Ghosh and Umesh Vazirani and Chenyi Zhang and Zixin Zhou},
  title = {{Quantum Pseudoentanglement}},
  year = {2023},
  eprint = {2211.00747v2},
  archivePrefix = {arXiv},
  url = {https://arxiv.org/abs/2211.00747v2},
  note = {arXiv:2211.00747v2; Version 2, April 2023}
}

@misc{MarkovGlue26,
  author = {Wonjun Lee and Hyukjoon Kwon and Gil Young Cho},
  title = {{Gluing Randomness via Entanglement: Tight Bound from Second R{\'e}nyi Entropy}},
  year = {2026},
  eprint = {2601.16454v1},
  archivePrefix = {arXiv},
  url = {https://arxiv.org/abs/2601.16454v1},
  note = {arXiv:2601.16454v1}
}

@misc{PRSS24,
  author = {Chuhan Lu and Minglong Qin and Fang Song and Penghui Yao and Mingnan Zhao},
  title = {{Quantum Pseudorandom Scramblers}},
  year = {2024},
  eprint = {2309.08941v2},
  archivePrefix = {arXiv},
  url = {https://arxiv.org/abs/2309.08941v2},
  note = {arXiv:2309.08941v2; Version 2, September 2024}
}

@misc{Inflationary24,
  author = {Claudio Chamon and Eduardo R. Mucciolo and Andrei E. Ruckenstein and Zhi-Cheng Yang},
  title = {{Fast pseudorandom quantum state generators via inflationary quantum gates}},
  year = {2024},
  eprint = {2304.09885v3},
  archivePrefix = {arXiv},
  url = {https://arxiv.org/abs/2304.09885v3},
  note = {arXiv:2304.09885v3; Published in npj Quantum Information 10, 37 (2024)}
}

@misc{SHH24,
  author = {Thomas Schuster and Jonas Haferkamp and Hsin-Yuan Huang},
  title = {{Random unitaries in extremely low depth}},
  year = {2024},
  eprint = {2407.07754v2},
  archivePrefix = {arXiv},
  url = {https://arxiv.org/abs/2407.07754v2},
  note = {arXiv:2407.07754v2; Version 2, January 2025; first posted in July 2024}
}

@misc{PC26,
  author = {Ben Foxman and Alex Lombardi and Fermi Ma and Barak Nehoran and John Wright},
  title = {{Quantum Lazy Sampling and Path Recording for Any Group}},
  year = {2026},
  eprint = {2606.30281v1},
  archivePrefix = {arXiv},
  url = {https://arxiv.org/abs/2606.30281v1},
  note = {arXiv:2606.30281v1}
}

@misc{BlockedLRFC25,
  author = {Laura Cui and Thomas Schuster and Fernando Brand{\~a}o and Hsin-Yuan Huang},
  title = {{Unitary designs in nearly optimal depth}},
  year = {2025},
  eprint = {2507.06216v2},
  archivePrefix = {arXiv},
  url = {https://arxiv.org/abs/2507.06216v2},
  note = {arXiv:2507.06216v2}
}

@misc{Distinct26,
  author = {Asad Raza and Jens Eisert and Bill Fefferman},
  title = {{Distinctness threshold for pseudorandom unitaries}},
  year = {2026},
  eprint = {2609.03065v1},
  archivePrefix = {arXiv},
  url = {https://arxiv.org/abs/2609.03065v1},
  note = {arXiv:2609.03065v1}
}

@misc{BM24,
  author = {Zvika Brakerski and Nir Magrafta},
  title = {{Real-Valued Somewhat-Pseudorandom Unitaries}},
  year = {2024},
  eprint = {2403.16704v2},
  archivePrefix = {arXiv},
  url = {https://arxiv.org/abs/2403.16704v2},
  note = {arXiv:2403.16704v2}
}

@misc{ConstantTime25,
  author = {Ben Foxman and Natalie Parham and Francisca Vasconcelos and Henry Yuen},
  title = {{Random Unitaries in Constant (Quantum) Time}},
  year = {2025},
  eprint = {2508.11487v3},
  archivePrefix = {arXiv},
  url = {https://arxiv.org/abs/2508.11487v3},
  note = {arXiv:2508.11487v3}
}

@misc{StrongGlue25,
  author = {Prabhanjan Ananth and John Bostanci and Aditya Gulati and Yao-Ting Lin},
  title = {{Gluing Random Unitaries with Inverses and Applications to Strong Pseudorandom Unitaries}},
  year = {2025},
  eprint = {2510.04085v1},
  archivePrefix = {arXiv},
  url = {https://arxiv.org/abs/2510.04085v1},
  note = {arXiv:2510.04085v1; Part II of the full version of the merged CRYPTO 2025 paper}
}

@misc{GridStrong26,
  author = {Marten Folkertsma and Lorenzo Grevink and Jonas Helsen and Alicja Dutkiewicz},
  title = {{Arts \& crafts: Strong random unitaries and geometric locality}},
  year = {2026},
  eprint = {2605.03023v1},
  archivePrefix = {arXiv},
  url = {https://arxiv.org/abs/2605.03023v1},
  note = {arXiv:2605.03023v1}
}

@misc{PermHam24,
  author = {Chi-Fang Chen and Adam Bouland and Fernando G. S. L. Brand{\~a}o and Jordan Docter and Patrick Hayden and Michelle Xu},
  title = {{Efficient unitary designs and pseudorandom unitaries from permutations}},
  year = {2024},
  eprint = {2404.16751v2},
  archivePrefix = {arXiv},
  url = {https://arxiv.org/abs/2404.16751v2},
  note = {arXiv:2404.16751v2}
}

@misc{QHROM25,
  author = {Prabhanjan Ananth and John Bostanci and Aditya Gulati and Yao-Ting Lin},
  title = {{Pseudorandom Unitaries in the Haar Random Oracle Model}},
  year = {2025},
  eprint = {2509.24432v1},
  archivePrefix = {arXiv},
  url = {https://arxiv.org/abs/2509.24432v1},
  note = {arXiv:2509.24432v1; Part I of the full version of the merged CRYPTO 2025 paper}
}

@misc{HiddenHam25,
  author = {Laura Cui and Thomas Schuster and Liang Mao and Hsin-Yuan Huang and Fernando Brand{\~a}o},
  title = {{Random unitaries from Hamiltonian dynamics}},
  year = {2025},
  eprint = {2510.08434v1},
  archivePrefix = {arXiv},
  url = {https://arxiv.org/abs/2510.08434v1},
  note = {arXiv:2510.08434v1}
}

@misc{PRI24,
  author = {Prabhanjan Ananth and Aditya Gulati and Fatih Kaleoglu and Yao-Ting Lin},
  title = {{Pseudorandom Isometries}},
  year = {2023},
  eprint = {2311.02901v3},
  archivePrefix = {arXiv},
  url = {https://arxiv.org/abs/2311.02901v3},
  note = {arXiv:2311.02901v3; Cited arXiv version posted in November 2023}
}

@misc{BoundedCopies25,
  author = {Bruno Cavalar and Eli Goldin and Matthew Gray and Peter Hall and Yanyi Liu and Angelos Pelecanos},
  title = {{On the Computational Hardness of Quantum One-Wayness}},
  year = {2025},
  eprint = {2312.08363v4},
  archivePrefix = {arXiv},
  url = {https://arxiv.org/abs/2312.08363v4},
  note = {arXiv:2312.08363v4; Published in Quantum 9, 1679 (2025)}
}

@misc{Aaronsontomography,
  author = {Scott Aaronson},
  title = {{Shadow Tomography of Quantum States}},
  year = {2018},
  note = {arXiv:1711.01053; Proceedings of STOC 2018},
  url = {https://arxiv.org/abs/1711.01053}
}

@misc{HKP20,
  author = {Hsin-Yuan Huang and Richard Kueng and John Preskill},
  title = {{Predicting many properties of a quantum system from very few measurements}},
  year = {2020},
  note = {arXiv:2002.08953; Nature Physics 16, 1050--1057},
  url = {https://arxiv.org/abs/2002.08953}
}

@inproceedings{KT25,
  author = {Dakshita Khurana and Kabir Tomer},
  title = {{Founding Quantum Cryptography on Quantum Advantage, or, Towards Cryptography from \#P Hardness}},
  booktitle = {Proceedings of the 57th Annual ACM Symposium on Theory of Computing},
  series = {STOC '25},
  publisher = {Association for Computing Machinery},
  year = {2025},
  pages = {178--188},
  doi = {10.1145/3717823.3718145},
  url = {https://arxiv.org/abs/2409.15248},
  note = {Full version: arXiv:2409.15248}
}

@inproceedings{IL89,
  author = {Russell Impagliazzo and Michael Luby},
  title = {{One-way functions are essential for complexity based cryptography}},
  booktitle = {30th Annual Symposium on Foundations of Computer Science},
  publisher = {IEEE},
  year = {1989},
  pages = {230--235},
  doi = {10.1109/SFCS.1989.63483},
  url = {https://doi.org/10.1109/SFCS.1989.63483}
}

@article{KK07,
  author = {Elham Kashefi and Iordanis Kerenidis},
  title = {{Statistical Zero Knowledge and quantum one-way functions}},
  journal = {Theoretical Computer Science},
  volume = {378},
  number = {1},
  pages = {101--116},
  year = {2007},
  doi = {10.1016/j.tcs.2007.03.013},
  url = {https://arxiv.org/abs/quant-ph/0511266}
}

@misc{HirokaHsieh24,
  author = {Taiga Hiroka and Min-Hsiu Hsieh},
  title = {{Computational Complexity of Learning Efficiently Generatable Pure States}},
  year = {2024},
  eprint = {2410.04373},
  archivePrefix = {arXiv},
  url = {https://arxiv.org/abs/2410.04373},
  note = {arXiv:2410.04373}
}
\endgroup

\appendix
\section{One-way functions from computable samplable OWSGs}
\label{sec:owf-samplable}

In two separate steps, the tomography procedure from Section~\ref{owsg:section}
uses an $\NP$ oracle to choose a key that maximizes the likelihood of an
observed measurement transcript. In this section, we replace each such search
with a procedure that samples a key from the posterior distribution on keys
conditioned on the observed transcript.
We then show that when applying the learning procedure to computable samplable OWSGs, these posterior-sampling procedures can be implemented using distributional inversion, which yields a classical one-way function from OWSG security. 
\\

\noindent Fix the security parameter $n$.  Write
\[
    \mathcal K=\mathcal K_n,\qquad L=\ell(n),\qquad
    \mathcal X=\{0,1\}^n,\qquad D=2^n.
\]
For $k\in\mathcal K$, write
\[
    |\psi_k\rangle=\sum_{x\in\mathcal X}a_k(x)|x\rangle,
    \qquad
    p_k(x)=|a_k(x)|^2,
    \qquad
    |\psi_k^+\rangle=\sum_x|a_k(x)|\,|x\rangle.
\]
Recall the controlled-swap measurement $\mathcal M$ and the distribution
$R_{\psi,q}$ from Section~\ref{owsg:section}.  Thus for a distribution $q$, $R_{\psi,q}$ is the
outcome distribution of $\mathcal M$ on $|\psi\rangle|q^+\rangle$.  Define
\begin{equation}
    S_k:=R_{\psi_k,p_k}=R_{k,k}.
    \label{owf:matched-distribution}
\end{equation}
Lemma \ref{owsg:matched-fidelity} shows that
\begin{equation}
    1-|\langle\psi_{k'}|\psi_k\rangle|^2
    \le 4H^2(S_k,S_{k'})
    \qquad\text{for all }k,k'\in\mathcal K.
    \label{owf:matched-identification}
\end{equation}
For both results below, the unknown key has a known classically samplable
distribution $\pi_n$, i.e. there exists a deterministic polynomial time procedure $G_n$ such that $K \sim \pi_n$ where
\begin{equation}
    K=G_n(U),\qquad U\leftarrow\{0,1\}^{r(n)},
    \label{owf:key-seed}
\end{equation}
  We suppress the subscript
$n$ when clear.

\subsection{Posterior sampling}

\textbf{The posterior-fitting procedure.}
Let $\{Q_k\}_{k\in\mathcal K}$ be any keyed family of distributions and
let $s\ge1$.  Consider the joint experiment
\[
    K\sim\pi,
    \qquad
    (Y\mid K=k)\sim Q_k^{\otimes s}.
\]
Given a transcript $y$ of positive probability, define
$\mathsf{PostFit}_{Q,s}(y)$ to sample a key from
\begin{equation}
    \Pr[\mathsf{PostFit}_{Q,s}(y)=k']
    =
    \frac{\pi(k')Q_{k'}^{\otimes s}(y)}
         {\sum_j\pi(j)Q_j^{\otimes s}(y)}.
    \label{owf:postfit}
\end{equation}
On a zero-probability transcript, define its output arbitrarily.  Thus,
conditioned on $Y$, the original key and the output of
$\mathsf{PostFit}_{Q,s}$ are independent samples from the same posterior.

\begin{definition}[Posterior-sampling error]
\label{owf:posterior-error-def}
Let $(V,Y)\sim J$ be a finite joint distribution.  A randomized procedure
$\mathsf{Post}$ that takes input $Y$ has posterior-sampling error at
most $\tau$ for $J$ if
\[
    \operatorname{TV}\bigl(
        (Y,\mathsf{Post}(Y)),
        (Y,V)
    \bigr)\le\tau.
\]
\end{definition}

\begin{lemma}[Posterior fitting]
\label{owf:posterior}
Let $K\sim\pi$ take values in a set of at most $M$ keys, and let
$(Y\mid K=k)\sim Q_k^{\otimes s}$.  Suppose
$K'=\mathsf{Post}(Y)$ has posterior-sampling error at most $\tau$ for
$(K,Y)$ and always outputs a valid key.  Then, for every $0<\gamma<2$,
\begin{align}
    \Pr\!\left[H^2(Q_K,Q_{K'})>\gamma\right]
    &\le \frac{M}{2}e^{-s\gamma/2}+\tau,
    \label{owf:posterior-fit}\\
    \operatorname{TV}(\mathcal L(K'),\pi)
    &\le \tau.
    \label{owf:posterior-marginal}
\end{align}
\end{lemma}

\begin{proof}
First let $K'$ be an exact posterior sample.  Put
\[
    b_k(y):=\pi(k)Q_k^{\otimes s}(y),
    \qquad
    b(y):=\sum_j b_j(y).
\]
For distinct $k,k'$ and $b(y)>0$, since $b(y) \geq b_k(y) + b_{k'}(y) \geq 2\sqrt{b_k(y)b_{k'}(y)}$,
\[
    \Pr[K=k,K'=k',Y=y]
    =\frac{b_k(y)b_{k'}(y)}{b(y)}
    \le\frac12\sqrt{b_k(y)b_{k'}(y)}.
\]
Summing over $y$ gives
\[
    \Pr[K=k,K'=k']
    \le
    \frac12\sqrt{\pi(k)\pi(k')}\,B(Q_k,Q_{k'})^s.
\]
If $H^2(Q_k,Q_{k'})>\gamma$, then
$B(Q_k,Q_{k'})\le e^{-\gamma/2}$.  Summing over the relevant pairs and using
\[
    \left(\sum_k\sqrt{\pi(k)}\right)^2\le M
\]
gives the first bound for exact posterior sampling.  Moreover, averaging the
posterior distribution over $Y$ gives $K'\sim\pi$.

For a procedure with posterior-sampling error at most $\tau$, the joint distribution
of $(Y,K')$ is within total variation distance $\tau$ of the exact posterior
experiment.  Appending $K$ independently sampled from from its conditional
distribution given $Y$ cannot increase this distance.  Hence the probability of an event that depends on $K$ and $K'$ can only change by at most $\tau$ when we move from exact sampled $K'$ to $K'$ sampled from the procedure, giving the final bound we need. Taking the marginal on $K'$ gives
$\operatorname{TV}(\mathcal L(K'),\pi)\le\tau$.
\end{proof}

\textbf{Learning Algorithm.}
We now use posterior fitting to replace $\NP$ oracle calls in the tomography learning algorithm.
Note that here the second posterior call uses $\{S_k\}$ rather than $\{R_{k,a}\}$.

Fix positive integers $s_1,s_2$ and parameters
$0<\varepsilon,\gamma,\zeta<1$.  Suppose $\mathsf{Post}_1$ and
$\mathsf{Post}_2$ have posterior-sampling errors $\tau_1,\tau_2$
respectively for the joint distributions $(K, X)$ and $(K, Z)$ sampled according to the following experiments
\begin{align*}
    &K\sim\pi,
      \qquad X\mid K=k\sim p_k^{\otimes s_1},\\
    &K\sim\pi,
      \qquad Z\mid K=k\sim S_k^{\otimes s_2}.
\end{align*}
Given copies of an unknown state $|\psi_K\rangle$, the algorithm performs
the following steps.

\begin{enumerate}
    \item \textbf{Find a match for the basis distribution.}
    Measure the first $s_1$ copies in the computational basis, obtaining
    the transcript
    \[
        X=(X_1,\ldots,X_{s_1}).
    \]
    Set
    \[
        A:=\mathsf{Post}_1(X).
    \]

    \item \textbf{Prepare reference states.}
    Using the procedure of Lemma~\ref{owsg:positive}, prepare $s_2$
    independent reference states, each at trace distance at most $\zeta$
    from
    \[
        |\psi_A^+\rangle.
    \]

    \item \textbf{Collect measurement $\mathcal M$ outcomes.}
    For each of the remaining $s_2$ challenge copies, apply
    $\mathcal M$ to the challenge copy and a fresh reference state from
    Step~2.  Record the resulting transcript
    \[
        \widetilde Z=(\widetilde Z_1,\ldots,\widetilde Z_{s_2}).
    \]

    \item \textbf{Find a match for the matched-reference distribution.}
    Set
    \[
        K':=\mathsf{Post}_2(\widetilde Z)
    \]
    and return $K'$.
\end{enumerate}
The first posterior call replaces the first maximum-likelihood search in
the $\NP$ learning algorithm.  It produces a key whose basis distribution is close
to that of the challenge.  This key is used only to construct the
reference states.  The second posterior call is with respect to the
fixed family $\{S_k\}_{k\in\mathcal K}$.

\begin{proposition}
\label{owf:posterior-learning}
For the posterior learning algorithm above, let
\begin{align}
    \beta_1&:=\frac{M}{2}e^{-s_1\gamma/2},
    \label{owf:beta1}\\
    \beta_2&:=\frac{M}{2}e^{-s_2\varepsilon/8},
    \label{owf:beta2}\\
    E&:=\tau_1+\tau_2+\beta_1+\sqrt{s_2\gamma}+s_2\zeta,
    \label{owf:E}
\end{align}
where $M=|\mathcal K|$.  Then
\begin{align}
    \Pr\!\left[
        |\langle\psi_{K'}|\psi_K\rangle|^2<1-\varepsilon
    \right]
    &\le \beta_2+E,
    \label{owf:posterior-learning-fidelity}\\
    \operatorname{TV}(\mathcal{L}(K'),\pi)
    &\le E.
    \label{owf:posterior-learning-marginal}
\end{align}
\end{proposition}

\begin{proof}
We first show that the basis distribution of the key $A$ obtained in
Step~1 is close to that of the challenge.  We then replace the reference
states prepared from $A$ by the ideal reference $|\psi_K^+\rangle$.
This shows that the transcript collected in Step~3 is close to a sample
from $S_K^{\otimes s_2}$.  Finally, we apply the posterior-fitting lemma
a second time and use~\eqref{owf:matched-identification} to obtain high
fidelity.\\

\noindent Applying Lemma~\ref{owf:posterior} to the first transcript gives
\[
    \Pr\!\left[H^2(p_K,p_A)>\gamma\right]
    \le \beta_1+\tau_1.
\]
Suppose $H^2(p_K,p_A)\leq\gamma$.  Let $b:=B(p_K,p_A)$.
Since
\[
    H^2(p_K,p_A)=2-2b\le\gamma,
\]
we have
\[
    b\ge1-\frac{\gamma}{2},
    \qquad
    b^2\ge1-\gamma.
\]
By Lemma~\ref{lem:qsample-identities},
\[
    \langle\psi_K^+|\psi_A^+\rangle=b.
\]
Therefore
\begin{align*}
\mathsf{TD}\!\left(
    (|\psi_K^+\rangle\langle\psi_K^+|)^{\otimes s_2},
    (|\psi_A^+\rangle\langle\psi_A^+|)^{\otimes s_2}
\right)
=
\sqrt{1-b^{2s_2}}
\end{align*}
Since $b^2\ge 1-\gamma$, we have
\[
    b^{2s_2}
    =(b^2)^{s_2}
    \ge (1-\gamma)^{s_2}
    \ge 1-s_2\gamma,
\]
where the last inequality is Bernoulli's inequality. Therefore
\[
    1-b^{2s_2}\le s_2\gamma,
\]
and hence
\[
\mathsf{TD}\!\left(
    (|\psi_K^+\rangle\langle\psi_K^+|)^{\otimes s_2},
    (|\psi_A^+\rangle\langle\psi_A^+|)^{\otimes s_2}
\right)
\le \sqrt{s_2\gamma}.
\]

Each reference actually prepared in Step~2 is within trace distance
$\zeta$ of $|\psi_A^+\rangle$.  Replacing the $s_2$ prepared references with exact copies of $|\psi_A^+\rangle$
one at a time therefore changes the joint measurement transcript of the third step by
total variation distance at most $s_2\zeta$.
By the above, replacing the exact references $|\psi_A^+\rangle^{\otimes s_2}$ by
$|\psi_K^+\rangle^{\otimes s_2}$ costs at most
$\sqrt{s_2\gamma}$.

Including the probability that the first posterior fitting step fails, the actual
joint distribution of $(K,\widetilde Z)$ is therefore within
\[
    \Delta
    :=
    \tau_1+\beta_1+\sqrt{s_2\gamma}+s_2\zeta
\]
of the ideal distribution
\[
    K\sim\pi,
    \qquad
    Z\mid K=k\sim S_k^{\otimes s_2}.
\]

Next we show that the second posterior fitting step gives a high fidelity key.
Suppose $K'$ was generated by applying $\mathsf{Post}_2$ to the ideal distribution transcript. By
Lemma~\ref{owf:posterior}, with threshold $\varepsilon/4$,
\[
    \Pr\!\left[
        H^2(S_K,S_{K'})>\frac{\varepsilon}{4}
    \right]
    \le
    \beta_2+\tau_2
\]
and
\[
    \operatorname{TV}(\mathcal L(K'),\pi)
    \le\tau_2.
\]
Applying the same randomized map $\mathsf{Post}_2$ to the actual and
ideal transcript distributions cannot increase total variation distance.
Restoring the actual transcript therefore adds at most $\Delta$ to the failure
probability and to the marginal total variation distance.
Finally, whenever
\[
    H^2(S_K,S_{K'})\le\frac{\varepsilon}{4},
\]
equation~\eqref{owf:matched-identification} gives
\[
    1-|\langle\psi_{K'}|\psi_K\rangle|^2
    \le
    4H^2(S_K,S_{K'})
    \le\varepsilon.
\]
Thus
\[
    \Pr\!\left[
        |\langle\psi_{K'}|\psi_K\rangle|^2<1-\varepsilon
    \right]
    \le
    \beta_2+\tau_2+\Delta
    =
    \beta_2+E.
\]
and similarly
\[
    \operatorname{TV}(\mathcal L(K'),\pi)
    \le
    \tau_2+\Delta
    =
    E.
\]
\end{proof}

\begin{corollary}[Parameters for posterior learning]
\label{owf:posterior-parameters}
For $0<\varepsilon,\delta<1/4$, put
\begin{align}
    A_0&:=L\ln2+\ln\frac8\delta,
    \label{owf:A0}\\
    s_2&:=\left\lceil\frac{8A_0}{\varepsilon}\right\rceil,
    \label{owf:posterior-s2}\\
    \gamma&:=\frac{\delta^2}{64s_2},
    \label{owf:posterior-gamma}\\
    s_1&:=\left\lceil\frac{2A_0}{\gamma}\right\rceil,
    \label{owf:posterior-s1}\\
    \zeta&:=\frac{\delta}{8s_2}.
    \label{owf:posterior-zeta}
\end{align}
If $\tau_1,\tau_2\le\delta/8$, then the learning algorithm in
Proposition~\ref{owf:posterior-learning} satisfies
\begin{align}
    \Pr\!\left[|\langle\psi_{K'}|\psi_K\rangle|^2<1-\varepsilon\right]
    \le\delta
    \label{owf:posterior-param-fidelity}
\end{align}
It uses
\[
    s_1+s_2
    =O\!\left(
        \frac{(L+\log(1/\delta))^2}{\varepsilon\delta^2}
    \right)
\]
challenge copies.
\end{corollary}

\begin{proof}
The choices give
$\beta_1,\beta_2\le\delta/16$,
$\sqrt{s_2\gamma}=\delta/8$, and
$s_2\zeta=\delta/8$.  Hence
$E\le9\delta/16$ and
$\beta_2+E\le5\delta/8<\delta$.  The resource bounds follow directly
from the displayed parameters and Lemma~\ref{owsg:positive}.
\end{proof}

For OWSGs we also need to transfer average correctness to the posterior
output key.

\begin{lemma}[Verification with average correctness]
\label{owf:average-correctness}
Let $\rho_k=|\psi_k\rangle\langle\psi_k|$, and let $E_{k'}$ be the
verifier's single-copy acceptance effect for candidate key $k'$.  Suppose
\begin{equation}
    \mathbb E_{K\sim\pi}
    \operatorname{Tr}(E_K\rho_K)
    \ge1-\nu(n).
    \label{owf:average-correctness-assumption}
\end{equation}
If a learning algorithm returns $K'$ with fidelity-failure probability at most
$\delta$ at error $\varepsilon$ and
$\operatorname{TV}(\mathcal L(K'),\pi)\le\delta$, then its acceptance
probability on a fresh challenge copy is at least
\begin{equation}
    1-\nu(n)-2\delta-\sqrt{\varepsilon}.
    \label{owf:average-correctness-conclusion}
\end{equation}
\end{lemma}

\begin{proof}
The marginal bound transfers~\eqref{owf:average-correctness-assumption}
to $K'$ at cost at most $\delta$.  On the successful-learning event,
\[
    T(\rho_K,\rho_{K'})
    \le\sqrt{\varepsilon},
\]
so changing the verifier's input costs at most this amount.  The
complementary event costs at most another $\delta$.
\end{proof}

\subsection{Constructing one-way functions}

Under the classical sampling assumption, the two transcript distributions
used by the posterior learner can themselves be sampled classically. Therefore, in the absence of one-way functions, we can use distributional inversion to implement the posterior calls. Since this violates the security of the OWSGs, we obtain a construction of one-way functions.

We use the following standard consequence of distributional inversion,
often referred to as universal extrapolation.

\begin{lemma}[Universal extrapolation {\cite{IL89,IL90,KK07}}]
\label{owf:universal-extrapolation}
Let $\mathsf{S}$ be a classical PPT algorithm that, on input $1^n$, outputs
a pair $(Y,V)$ of polynomial-length classical strings.  Suppose that
classical one-way functions secure against QPT inversion do not exist.
Then, for every inverse-polynomial $\tau>0$, there is a QPT algorithm
$\mathsf{Ext}$ such that, for infinitely many $n$,
\[
    \operatorname{TV}\!\left(
        \mathcal L\bigl(Y,\mathsf{Ext}(1^n,Y)\bigr),
        \mathcal L(Y,V)
    \right)
    \le \tau.
\]
\end{lemma}
Next we show that the transcript distribution $S_k$ can be
sampled classically. Recall that
\[
    |\psi_k\rangle=\sum_{x\in\mathcal X}a_k(x)|x\rangle,
    \qquad
    p_k(x)=|a_k(x)|^2,
    \qquad
    |\psi_k^+\rangle=\sum_x |a_k(x)|\,|x\rangle,
\]
where $\mathcal X=\{0,1\}^n$ and $D=|\mathcal X|=2^n$.
The distribution
\[
    S_k=R_{\psi_k,p_k}
\]
is the outcome distribution of the controlled-swap measurement
$\mathcal M$ from Section~\ref{owsg:section}, applied to
$|\psi_k\rangle|\psi_k^+\rangle$. Its outcomes have the form
$(\alpha,\sigma,x,y)$, where
$\alpha\in\{0,\pi/2\}$, $\sigma\in\{-1,1\}$, and
$x,y\in\mathcal X$.

\begin{lemma}
\label{owf:sample-S}
For every computable samplable state family and every $0<\chi<1$, there
is a classical PPT algorithm which, on input $(1^n,k,1/\chi)$, samples a
distribution $\widetilde S_k$ satisfying
\[
    \operatorname{TV}(\widetilde S_k,S_k)\le\chi.
\]
\end{lemma}

\begin{proof}
For $a_k(x)\neq0$, write
\[
    u_k(x):=\frac{a_k(x)}{|a_k(x)|},
\]
and set $u_k(x)=1$ when $a_k(x)=0$. Thus
$a_k(x)=\sqrt{p_k(x)}\,u_k(x)$ whenever $p_k(x)>0$.

By Lemma~\ref{owsg:outcome-distribution}, an exact sample from $S_k$
can be generated as follows. Sample $x,y\sim p_k$ independently, choose
$\alpha\in\{0,\pi/2\}$ uniformly, and then sample
$\sigma\in\{-1,1\}$ according to
\begin{equation}
    \Pr[\sigma\mid x,y,\alpha]
    =
    \frac{
        1+\sigma\operatorname{Re}
        \!\left(e^{i\alpha}u_k(x)\overline{u_k(y)}\right)
    }{2}.
    \label{owf:classical-S}
\end{equation}

We approximate each of these steps classically. First use the sampler
from Definition~\ref{def:computable-samplable-states} to sample each of
$x$ and $y$ from a distribution within total variation distance
$\chi/8$ of $p_k$. This changes their joint distribution by at most
$\chi/4$.
Next use the algorithm that computes amplitudes to obtain, for every
required basis string $x$, an approximation $r_k(x)$ satisfying
\[
    |r_k(x)-a_k(x)|\le d,
    \qquad
    d\le\frac{\chi}{32\sqrt D}.
\]
Define
\[
    v_k(x):=
    \begin{cases}
        r_k(x)/|r_k(x)|,&r_k(x)\ne0,\\
        1,&r_k(x)=0.
    \end{cases}
\]
As in the proof of Lemma~\ref{owsg:positive},
\[
    |a_k(x)-|a_k(x)|v_k(x)|
    \le2|a_k(x)-r_k(x)|.
\]
Hence
\[
    \mathbb E_{x\sim p_k}|u_k(x)-v_k(x)|
    \le2\sqrt D\,d
    \le\frac{\chi}{16}.
\]
For fixed $x,y,\alpha$, replacing $u_k$ by $v_k$ in
\eqref{owf:classical-S} changes the conditional distribution of
$\sigma$ by at most
\[
    \frac12\left(
        |u_k(x)-v_k(x)|
        +
        |u_k(y)-v_k(y)|
    \right).
\]
Averaging over $x,y\sim p_k$ therefore contributes at most $\chi/16$
to the total variation distance.

Finally, approximate the probability in
\eqref{owf:classical-S} sufficiently well to sample $\sigma$ within
total variation distance $\chi/8$. The total error is at most
\[
    \frac{\chi}{4}
    +\frac{\chi}{16}
    +\frac{\chi}{8}
    =
    \frac{7\chi}{16}
    <\chi.
\]
Since $\log D=n$, the required amplitude precision is
$O(n+\log(1/\chi))$ bits. Thus the entire procedure runs in polynomial
time.
\end{proof}
We can now finally prove the main theorem of the section. 

\computablesamplableowf*
\begin{proof}[Proof of Theorem~\ref{thm:computable-samplable-owf}]
Suppose, towards a contradiction, that classical one-way functions secure
against QPT inversion do not exist.

Let $E_k$ be the verifier acceptance POVM element.  Correctness of
the OWSG gives
\[
    \mathbb E_{K\sim\pi_n}
        \operatorname{Tr}(E_K\rho_K)
    \ge1-\nu(n),
    \qquad
    \rho_k:=|\psi_k\rangle\langle\psi_k|,
\]
for negligible $\nu$.
Fix
\[
    \delta_0:=\frac1{64},
    \qquad
    \varepsilon_0:=\frac1{64}.
\]
Use Corollary~\ref{owf:posterior-parameters} with
$(\varepsilon_0,\delta_0)$ to obtain the sample counts $s_1,s_2$ and
reference-preparation accuracy $\zeta$.  Set
\[
    \xi:=\frac{\delta_0}{64},
    \qquad
    \tau:=\frac{\delta_0}{64}.
\]
We define one classical PPT sampler containing the two transcript
experiments.  Sample
\[
    B\leftarrow\{1,2\}
\]
uniformly and independently generate
\[
    K=G_n(U),
    \qquad
    U\leftarrow\{0,1\}^{r(n)}.
\]
If $B=1$, classically generate a transcript $\widehat X$ consisting of
$s_1$ independent approximate samples from $p_K$, using
$\mathsf{Samp}$ with error at most $\xi/s_1$ for each sample.  If $B=2$,
classically generate a transcript $\widehat Z$ consisting of $s_2$
independent approximate samples from $S_K$, using
Lemma~\ref{owf:sample-S} with error at most $\xi/s_2$ for each sample.
Write
\[
    \widehat T_1:=\widehat X,
    \qquad
    \widehat T_2:=\widehat Z.
\]
The sampler outputs
\[
    \bigl((B,\widehat T_B),K\bigr).
\]
Let $T_1=X$ and $T_2=Z$ denote the ideal transcripts
\[
    X\mid K=k\sim p_k^{\otimes s_1},
    \qquad
    Z\mid K=k\sim S_k^{\otimes s_2}.
\]
By a hybrid argument over the samples,
\begin{equation}
    \operatorname{TV}\!\left(
        \mathcal L(K,\widehat T_b),
        \mathcal L(K,T_b)
    \right)
    \le\xi
    \qquad (b\in\{1,2\}).
    \label{owf:transcript-sampling-error}
\end{equation}
Apply Lemma~\ref{owf:universal-extrapolation} to the combined sampler.
For infinitely many security parameters there is a QPT algorithm
$\mathsf{Ext}$ satisfying
\begin{equation}
    \operatorname{TV}\!\left(
        \mathcal L\bigl(
            B,\widehat T_B,
            \mathsf{Ext}(1^n,(B,\widehat T_B))
        \bigr),
        \mathcal L(B,\widehat T_B,K)
    \right)
    \le\tau.
    \label{owf:combined-extrapolation}
\end{equation}
Since $B$ is uniform and appears in both distributions, conditioning on
either value $B=b$ gives error at most $2\tau$ for that branch. For $b\in\{1,2\}$, define
\[
    \mathsf{Post}_b(t):=\mathsf{Ext}(1^n,(b,t)).
\]
Conditioning~\eqref{owf:combined-extrapolation} on $B=b$ gives
\[
    \operatorname{TV}\!\left(
        \mathcal L(\widehat T_b,\mathsf{Post}_b(\widehat T_b)),
        \mathcal L(\widehat T_b,K)
    \right)
    \le 2\tau.
\]
By~\eqref{owf:transcript-sampling-error}, replacing $\widehat T_b$ by
$T_b$ changes each of the two distributions above by at most $\xi$.
Hence
\[
    \operatorname{TV}\!\left(
        \mathcal L(T_b,\mathsf{Post}_b(T_b)),
        \mathcal L(T_b,K)
    \right)
    \le 2\tau+2\xi.
\]
Thus $\mathsf{Post}_b$ has posterior-sampling error at most
$2\tau+2\xi$ for the ideal experiment. The two procedures therefore satisfy the requirements of the posterior learning
algorithm.  By Proposition~\ref{owf:posterior-learning} and the parameter
choices of Corollary~\ref{owf:posterior-parameters}, its output $K'$
satisfies
\[
    \Pr\!\left[
        |\langle\psi_{K'}|\psi_K\rangle|^2
        <1-\varepsilon_0
    \right]
    \le\delta_0,
    \qquad
    \operatorname{TV}(\mathcal L(K'),\pi_n)
    \le\delta_0.
\]
Lemma~\ref{owf:average-correctness} therefore gives verification success
at least
\[
    1-\nu(n)-2\delta_0-\sqrt{\varepsilon_0}
    =
    \frac{27}{32}-\nu(n).
\]
This holds for infinitely many security parameters, contradicting the
one-wayness of the OWSG.
Therefore classical one-way functions secure against QPT inversion
exist.
\end{proof}

\end{document}